\documentclass[11pt,a4paper]{article}

\usepackage{amssymb,amsthm,amsmath,mathtools}
\usepackage{thmtools,thm-restate}
\usepackage[round]{natbib}
\usepackage{pifont}
\usepackage[usenames,svgnames,xcdraw,table]{xcolor}
\definecolor{DarkGreen}{rgb}{0.1,0.5,0.1}
\usepackage[backref=page]{hyperref}
\hypersetup{
	colorlinks=true,
	linkcolor=red,
	urlcolor=DarkGreen,
	citecolor=blue
}
\renewcommand*{\backref}[1]{}
\renewcommand*{\backrefalt}[4]{%
    \ifcase #1 (Not cited.)%
    \or        (Cited on page~#2)%
    \else      (Cited on pages~#2)%
    \fi}
\usepackage[capitalise,noabbrev]{cleveref}

\Crefname{prop}{Proposition}{Propositions}
\Crefname{property}{Property}{Properties}
\Crefname{example}{Example}{Examples}
\Crefname{table}{Table}{Tables}

\usepackage{nicefrac}
\usepackage{enumerate}
\usepackage{etoolbox}
\usepackage[margin=1in]{geometry}
\usepackage[T1]{fontenc}
\usepackage[tt=false]{libertine}

\usepackage{url}
\usepackage[utf8]{inputenc}
\usepackage[small]{caption}
\usepackage{graphicx}
\usepackage{booktabs}
\usepackage{mathrsfs,amsfonts,dsfont,authblk}
\usepackage{array,multirow,graphicx,bigdelim}

\usepackage[linesnumbered,lined,boxed,ruled,vlined]{algorithm2e}

\SetKwInOut{Parameters}{Parameters}
\SetKwComment{Comment}{$\triangleright$\ }{}
\SetAlFnt{\small}
\SetAlCapFnt{\small}
\SetAlCapNameFnt{\small}
\SetAlCapHSkip{0pt}
\IncMargin{-\parindent}

\usepackage{tikz}
\usetikzlibrary{fit,calc}

\colorlet{mygray}{gray!40}

\newcommand*\circled[1]{\tikz[baseline=(char.base)]{
            \node[shape=circle,draw,inner sep=1pt] (char) {#1};}}

\makeatletter
\renewcommand{\paragraph}{%
  \@startsection{paragraph}{4}%
  {\z@}{1.0ex \@plus 1ex \@minus .2ex}{-1em}%
  {\normalfont\normalsize\bfseries}%
}
\let\oldnl\nl
\newcommand{\nonl}{\renewcommand{\nl}{\let\nl\oldnl}}
\makeatother

  {\list{}{\leftmargin=#1\rightmargin=#1}\item[]}%
  {\endlist}

\usepackage[labelfont={normalfont,bf},textfont=it]{caption}
\usepackage{subcaption}

\theoremstyle{definition}
\newtheorem{property}{Property}
\newtheorem{definition}{Definition}

\theoremstyle{remark}
\newtheorem{remark}{Remark}

\allowdisplaybreaks

\newcommand{\Alg}{\textup{\textsc{Alg}}}

\newcommand{\Ber}{\textrm{\textup{Ber}}}

\newcommand{\EF}[1]{\ifstrempty{#1}{\textrm{\textup{EF}}}
{\textrm{\textup{EF{$#1$}}}}}
\newcommand{\EFExistence}{\textup{\textsc{\EF{}-Existence}}}
\newcommand{\EEF}{\textrm{\textup{EEF}}}

\newcommand{\Envygraph}{\textup{\texttt{EnvyGraph}}}
\newcommand{\eps}{{\varepsilon}}
\newcommand{\EQcoloring}{\textup{\textsc{Equitable Coloring}}}

\newcommand{\F}{\mathcal{F}}

\newcommand{\GraphkColorability}{\textup{\textsc{Graph $k$-Colorability}}}
\newcommand{\HEF}[1]{\ifstrempty{#1}{\textrm{\textup{HEF}}}
{\textrm{\textup{HEF-{$#1$}}}}}
\newcommand{\HEFkExistence}{\textup{\textsc{\HEF{k}-Existence}}}
\newcommand{\HEFkVerification}{\textup{\textsc{\HEF{k}-Verification}}}

\newcommand{\HittingSet}{\textup{\textsc{Hitting Set}}}

\newcommand{\kopt}{k^{\textup{opt}}}

\newcommand{\I}{\mathcal{I}}
\newcommand{\kappaopt}{\kappa^{\textup{\texttt{opt}}}}

\newcommand{\Market}{\textup{\texttt{Alg-EF1+PO}}}

\newcommand{\MNW}{\textup{\texttt{MNW}}}

\newcommand{\N}{{\mathbb{N}}}

\newcommand{\NP}{\textrm{\textup{NP}}}
\newcommand{\NPC}{\textrm{\textup{NP-complete}}}
\newcommand{\NPH}{\textrm{\textup{NP-hard}}}
\newcommand{\NPhard}{\text{NP-hard}}
\newcommand{\NW}{\textrm{\textup{NSW}}}

\newcommand{\Partition}{\textrm{\textsc{Partition}}}

\newcommand{\PO}{\textrm{\textup{PO}}}

\newcommand{\RR}{\textup{\texttt{RoundRobin}}}
\newcommand{\reg}{\texttt{\textup{reg}}}

\newcommand{\sEF}[1]{\ifstrempty{#1}{\textrm{\textup{sEF}}}
{\textrm{\textup{sEF{$#1$}}}}}
\newcommand{\uHEF}[1]{\ifstrempty{#1}{\textrm{\textup{uHEF}}}
{\textrm{\textup{uHEF-{$#1$}}}}}

\newcommand{\V}{\mathcal{V}}

\title{Fair Division through Information Withholding}

\author[1]{Hadi Hosseini}
\author[2]{Sujoy Sikdar}
\author[3]{Rohit Vaish}
\author[3]{Jun Wang}
\author[3]{Lirong Xia}

\affil[1]{Rochester Institute of Technology\\
	{\small\texttt{hhvcs@rit.edu}}}
\affil[2]{Washington University in St. Louis\\
	{\small\texttt{sujoy@wustl.edu}}}
\affil[3]{Rensselaer Polytechnic Institute\\ 
	{\small\texttt{\{vaishr2,wangj38,xial\}@rpi.edu}}}

\begin{document}

\maketitle

\begin{abstract}
Envy-freeness up to one good (\EF{1}) is a well-studied fairness notion for indivisible goods that addresses pairwise envy by the removal of at most one good. In the worst case, each pair of agents might require the (hypothetical) removal of a different good, resulting in a weak \emph{aggregate} guarantee. We study allocations that are nearly envy-free in aggregate, and define a novel fairness notion based on \emph{information withholding}. Under this notion, an agent can withhold (or hide) some of the goods in its bundle and reveal the remaining goods to the other agents. We observe that in practice, envy-freeness can be achieved by withholding only a small number of goods overall. We show that finding allocations that withhold an optimal number of goods is computationally hard even for highly restricted classes of valuations. In contrast to the worst-case results, our experiments on synthetic and real-world preference data show that existing algorithms for finding \EF{1} allocations withhold close-to-optimal number of goods.
\end{abstract}

\section{Introduction}
\label{sec:Introduction}

When dividing discrete objects, one often strives for a fairness notion called \emph{envy-freeness} \citep{F67resource}, under which no agent prefers the allocation of another agent to its own. Envy-free outcomes might not exist in general (even with only two agents and a single indivisible good), motivating the need for approximations. Among the many approximations of envy-freeness proposed in the literature \citep{LMM+04approximately,B11combinatorial,NR14minimizing,CKM+16unreasonable}, the notion called \emph{envy-freeness up to one good} (\EF{1}) has received significant attention recently. \EF{1} requires that pairwise envy can be eliminated by the removal of some good in the envied bundle. It is known that an \EF{1} allocation always exists and can be computed in polynomial time~\citep{LMM+04approximately}.

On closer scrutiny, however, we find that \EF{1} is not as strong as one might think. Indeed, an \EF{1} allocation could entail the (hypothetical) removal of \emph{many} goods, because the elimination of envy for different pairs of agents may require the removal of distinct goods. To see this, consider an instance with six goods $g_1,\dots,g_6$ and three agents $a_1,a_2,a_3$ whose (additive) valuations are as shown below:

\begin{table}[ht]
	\label{tab:Motivating_HEF}
	\centering
	\small
	\begin{tabular}{ c|cccccc }
		& $g_1$ & $g_2$ & $g_3$ & $g_4$ & $g_5$ & $g_6$\\ \hline
		$a_1$ & $\circled{1}$ & $\circled{1}$ & $\underline{4}$ & $1$ & $1$ & $\underline{4}$\\
		$a_2$ & $1$ & $\underline{4}$ & $\circled{1}$ & $\circled{1}$ & $\underline{4}$ & $1$\\
		$a_3$ & $\underline{4}$ & $1$ & $1$ & $\underline{4}$ & $\circled{1}$ & $\circled{1}$\\
	\end{tabular}
\end{table}

Observe that the allocation shown via circled goods is \EF{1}, since any pairwise envy can be addressed by removing an underlined good. However, each pair of agents requires the removal of a \emph{different} good (e.g., $a_1$'s envy towards $a_2$ is addressed by removing $g_3$ whereas $a_3$'s envy towards $a_2$ is addressed by removing $g_4$, and so on), resulting in a weak approximation overall (since all goods need to be removed over all pairs of agents).

The above example shows that \EF{1}, on its own, is too \emph{coarse} to distinguish between allocations that remove a \emph{large} number of goods (such as the one with circled entries) and those that remove only a \emph{few} (such as the one with underlined entries, which, in fact, is envy-free). This limitation highlights the need for a fairness notion that (a) can distinguish between allocations in terms of their \emph{aggregate} approximation, and (b) retains the ``up to one good'' style approximation of \EF{1} that has proven to be practically useful~\citep{GP15spliddit}. Our work aims to fill this important gap.

We propose a new fairness notion called \emph{envy-freeness up to $k$ hidden goods} (\HEF{k}) defined as follows: Say there are $n$ agents, $m$ goods, and an allocation $A = (A_1,\dots,A_n)$. Suppose there is a set $S$ of $k$ goods (called the \emph{hidden} set) such that each agent $i$ withholds the goods in $A_i \cap S$ (i.e., the hidden goods owned by $i$) and only discloses the goods in $A_i \setminus S$ to the other agents. Any other agent $h \neq i$ only observes the goods disclosed by $i$ (i.e., those in $A_i \setminus S$), and its valuation for $i$'s bundle is therefore $v_h(A_i \setminus S)$ instead of $v_h(A_i)$. Additionally, agent $h$'s valuation for its own bundle is $v_h(A_h)$ (and not $v_h(A_h \setminus S)$) because it can observe its own hidden goods. If, under the disclosed allocation, no agent prefers the bundle of any other agent (i.e., if $v_h(A_h) \geq v_h(A_i \setminus S)$ for every pair of agents $i,h$), then we say that $A$ is \emph{envy-free up to $k$ hidden goods} (\HEF{k}). In other words, by withholding the information about $S$, allocation $A$ can be made free of envy.

Notice how \HEF{k} addresses the limitations associated with \EF{1}: Like \EF{1}, \HEF{k} is a relaxation of envy-freeness that is defined in terms of the \emph{number of goods}. However, unlike \EF{1}, \HEF{k} offers a \emph{precise quantification} of the extent of information that must be withheld in order to achieve envy-freeness.

Clearly, any allocation can be made envy-free by hiding all the goods (i.e., if $k = m$). The real strength of \HEF{k} lies in $k$ being \emph{small}; indeed, an \HEF{0} allocation is envy-free. As we will demonstrate below, there are natural settings that admit \HEF{k} allocations with a small $k$ (i.e., hide only a small number of goods) even when (exact) envy-freeness is unlikely.

\subsection*{Information Withholding is Meaningful in \mbox{Practice}.}

To understand the usefulness of \HEF{k}, we generated a synthetic dataset where we varied the number of agents $n$ from $5$ to $10$, and the number of goods $m$ from $5$ to $20$ (we ignore the cases where $m < n$). For every fixed $n$ and $m$, we generated $100$ instances with \emph{binary} valuations. Specifically, for every agent $i$ and every good $j$, the valuation $v_{i,j}$ is drawn i.i.d. from $\textup{Bernoulli}(0.7)$. \Cref{subfig:HEFk_motivation_0.7_NOTEF} shows the heatmap of the number of instances out of $100$ that \emph{do not} admit envy-free outcomes. (Thus, a 	`hot' cell indicated by red color is one where \emph{none} of the $100$ instances admits an envy-free allocation.) \Cref{subfig:HEFk_motivation_0.7_maxk} shows the heatmap of the number of goods that must be hidden in the worst-case. That is, the color of each cell denotes the smallest $k$ such that each of the corresponding $100$ instances admits some \HEF{k} allocation.

\begin{figure}[h]
	\centering
	\begin{subfigure}[b]{0.47\linewidth}
		\centering
		\includegraphics[width=0.94\linewidth]{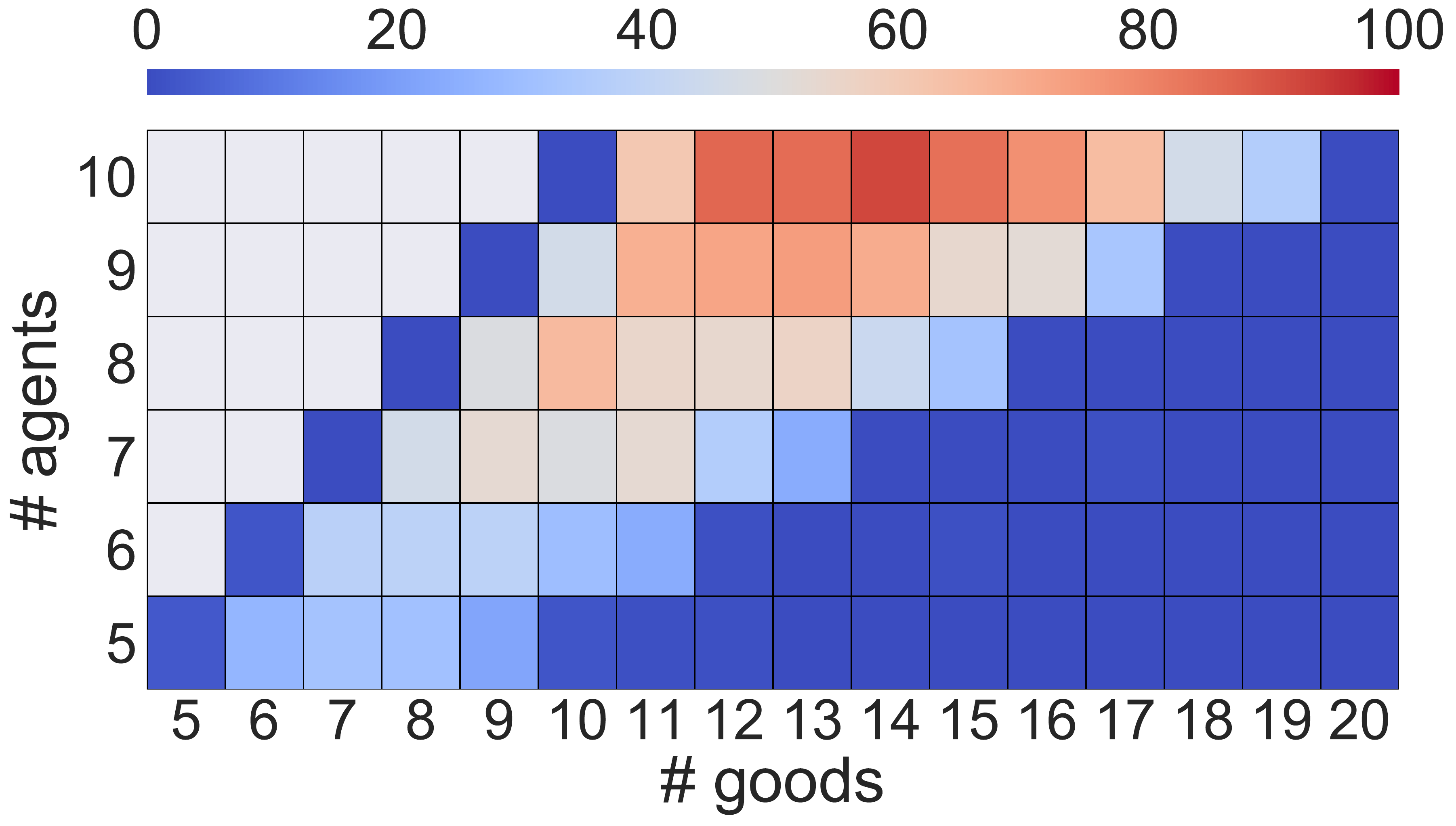}%
		\caption{Heatmap of the fraction of instances that are not envy-free.}
		\label{subfig:HEFk_motivation_0.7_NOTEF}
	\end{subfigure}
	\hfill
	\begin{subfigure}[b]{0.47\linewidth}
		\centering
		\includegraphics[width=0.94\linewidth]{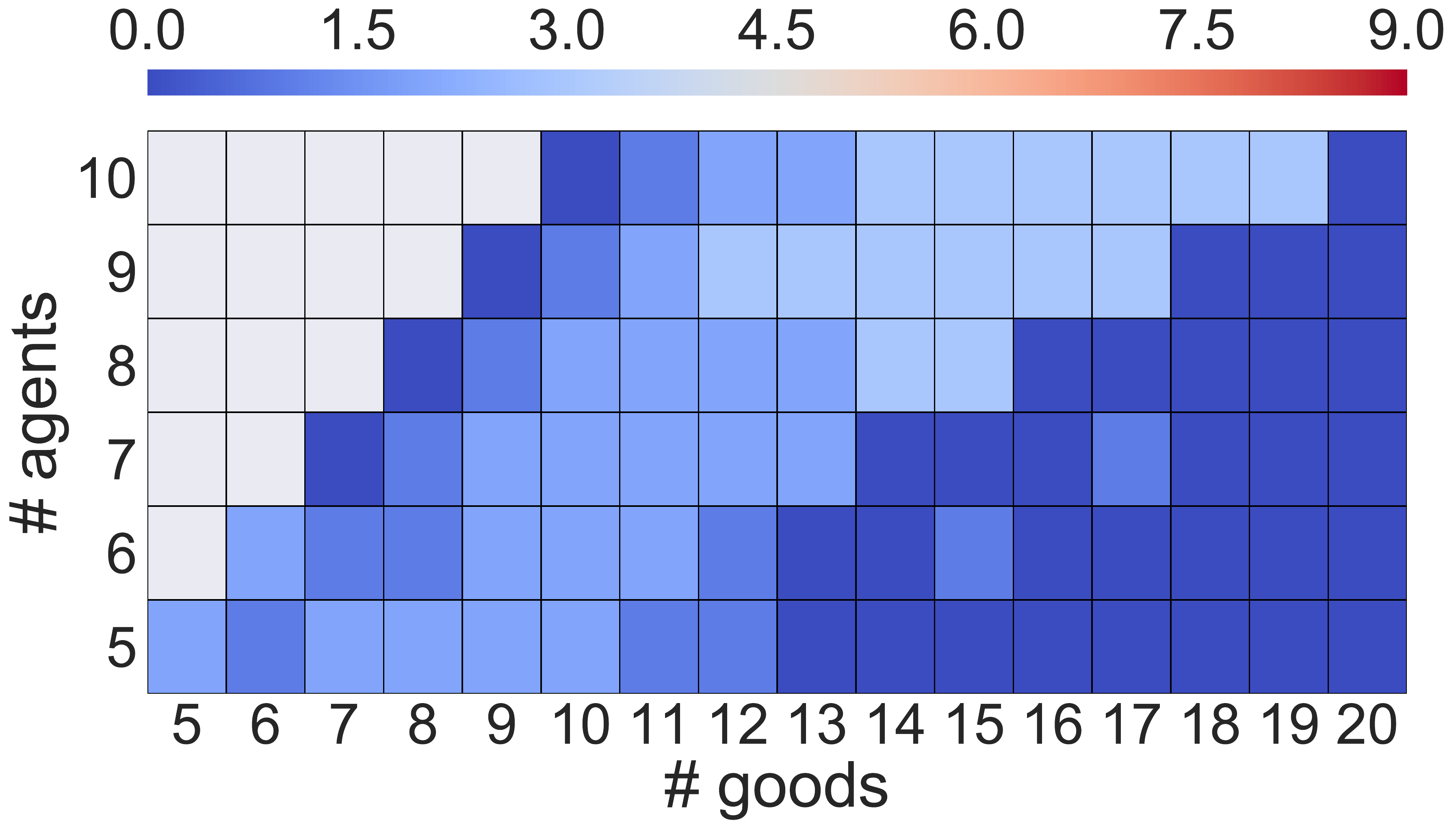}%
		\caption{Heatmap of the number of goods that must be hidden.}
		\label{subfig:HEFk_motivation_0.7_maxk}
	\end{subfigure}
	\caption{In both figures, each cell corresponds to $100$ instances with binary valuations for a fixed number of goods $m$ (on X-axis) and a fixed number of agents $n$ (on Y-axis).}
	\label{fig:HEFk_motivation_0.7}
\end{figure}

It is evident from \Cref{fig:HEFk_motivation_0.7} that even in the regime where envy-free outcomes are unlikely (in particular, the red-colored cells in \Cref{subfig:HEFk_motivation_0.7_NOTEF}), there exist \HEF{k} allocations with $k \leq 3$ (the light blue-colored cells in \Cref{subfig:HEFk_motivation_0.7_maxk}). This observation, along with the foregoing discussion, suggests that fairness through information withholding is a well-motivated approach towards approximate envy-freeness that could provide promising results in practice.

\paragraph{Our Contributions}
We make contributions on three fronts.
\begin{itemize}
	\item On the \emph{conceptual} side, we propose a novel fairness notion called envy-freeness up to $k$ hidden goods (\HEF{k}) as a fine-grained generalization of envy-freeness in terms of aggregate approximation.
	\item Our \emph{theoretical} results (\Cref{sec:Theoretical_Results}) show that computing \HEF{k} allocations is computationally hard even for highly restricted classes of valuations (\Cref{thm:HEFk_Existence_NPcomplete_IdenticalVals,cor:HEFk_Existence_NPcomplete_BinaryVals}). We show a similar result when \HEF{k} is required alongside Pareto optimality (\Cref{prop:HEFkPOExistence_NPcomplete_BinaryVals}). A related technical contribution is an alternative proof of \NPC{}ness of determining the existence of an envy-free allocation for binary valuations (\Cref{prop:EF_Existence_NPcomplete_BinaryValuations}).
	\item Our \emph{experiments} show that \HEF{k} allocations with a small $k$ often exist, even when (exact) envy-free allocations do not (\Cref{fig:HEFk_motivation_0.7}). We also compare several known algorithms for computing \EF{1} allocations on synthetic and real-world preference data, and find that the round-robin algorithm and an algorithm of \citet{BKV18Finding} withhold close-to-optimal number of goods, often hiding no more than three items (\Cref{sec:Experiments}).
\end{itemize}

\section{Related Work}
\label{sec:RelatedWork}
An emerging line of work in the fair division literature considers relaxations of envy-freeness by limiting the information available to the agents. Notably, \citet{ABC+18knowledge} consider a setting where each agent is aware only of its own bundle and has no knowledge about the allocations of the other agents. They propose the notion of \emph{epistemic envy-freeness} (\EEF{}) under which each agent believes that an envy-free allocation of the remaining goods among the other agents is possible. Note that in \EEF{}, each agent might consider a different hypothetical assignment of the remaining goods, and each of these could be significantly different from the \emph{actual} underlying allocation. By contrast, under \HEF{k}, each agent evaluates its valuation with respect to the same (underlying) allocation. \citet{CS17ignorance} study a related model where agents have probabilistic beliefs about the allocations of the other agents, and envy is defined in expectation. \citet{CCL+19maximin} study a setting similar to \citet{ABC+18knowledge} wherein each agent is unaware of the allocations of the other agents, with the guarantee that it does not get the worst bundle.

Another related line of work considers settings where the agents constitute a social network and can only observe the allocations of their neighbors \citep{AKP17fair,BQZ17networked,CEM17distributed,ABC+18knowledge,BCG+18local,BKN18envy}. These works place an informational constraint on the \emph{set of agents}, whereas our model restricts the \emph{set of revealed goods} per agent.

Several other forms of fairness approximations have been proposed recently, such as introducing side payments~\citep{HS19fair}, permitting sharing of some goods~\citep{SE19fair}, or donating a small fraction of goods~\citep{CGH19envy,CKM+20little}.

\section{Preliminaries}
\label{sec:Preliminaries}
\paragraph{Problem instance}
An \emph{instance} $\I = \langle [n], [m], \V \rangle$ of the fair division problem is defined by a set of $n \in \N$ \emph{agents} $[n] = \{1,2,\dots,n\}$, a set of $m \in \N$ \emph{goods} $[m] = \{1,2,\dots,m\}$, and a \emph{valuation profile} $\V = \{v_1,v_2,\dots,v_n\}$ that specifies the preferences of every agent $i \in [n]$ over each subset of the goods in $[m]$ via a \emph{valuation function} $v_i: 2^{[m]} \rightarrow \N \cup \{0\}$. Notice that each agent's valuation for any subset of goods is assumed to be a non-negative integer. We will assume that the valuation functions are \emph{additive}, i.e., for any $i \in [n]$ and $G \subseteq [m]$, $v_i(G) \coloneqq \sum_{j \in G} v_i(\{j\})$, where $v_i(\emptyset) = 0$. We will write $v_{i,j}$ instead of $v_i(\{j\})$ for a singleton good $j \in [m]$. We say that an instance has \emph{binary valuations} if for every $i \in [n]$ and every $j \in [m]$, $v_{i,j} \in \{0,1\}$.

\paragraph{Allocation}
An \emph{allocation} $A \coloneqq (A_1,\dots,A_n)$ refers to an $n$-partition of the set of goods $[m]$, where $A_i \subseteq [m]$ is the \emph{bundle} allocated to agent $i$. Given an allocation $A$, the utility of agent $i \in [n]$ for the bundle $A_i$ is $v_i(A_i) = \sum_{j \in A_i} v_{i,j}$.

\begin{definition}[\textbf{Envy-freeness}]
	An allocation $A$ is \emph{envy-free} (\EF{}) if for every pair of agents $i,h \in [n]$, $v_i(A_i) \geq v_i(A_h)$. An allocation $A$ is \emph{envy-free up to one good} (\EF{1}) if for every pair of agents $i,h \in [n]$ such that $A_h \neq \emptyset$, there exists some good $j \in A_h$ such that $v_i(A_i) \geq v_i(A_h \setminus \{j\})$. An allocation $A$ is \emph{strongly envy-free up to one good} (\sEF{1}) if for every agent $h \in [n]$ such that $A_h \neq \emptyset$, there exists a good $g_h \in A_h$ such that for all $i \in [n]$, $v_i(A_i) \geq v_i(A_h \setminus \{g_h\})$. The notions of \EF{}, \EF{1}, and \sEF{1} are due to \citet{F67resource}, \citet{B11combinatorial}, and \citet{CFS+19group}, respectively.\footnote{A slightly weaker notion than \EF{1} was previously studied by \citet{LMM+04approximately}. However, their algorithm can be shown to compute an \EF{1} allocation.}
\end{definition}

\begin{definition}[\textbf{Envy-freeness with hidden goods}]
	An allocation $A$ is said to be \emph{envy-free up to $k$ hidden goods} (\HEF{k}) if there exists a set $S \subseteq [m]$ of at most $k$ goods such that for every pair of agents $i,h \in [n]$, we have $v_i(A_i) \geq v_i(A_h \setminus S)$. An allocation $A$ is \emph{envy-free up to $k$ uniformly hidden goods} (\uHEF{k}) if there exists a set $S \subseteq [m]$ of at most $k$ goods satisfying $|S \cap A_i| \leq 1$ for every $i \in [n]$ such that for every pair of agents $i,h \in [n]$, we have $v_i(A_i) \geq v_i(A_h \setminus S)$. We say that allocation $A$ \emph{hides} the goods in $S$ and \emph{reveals} the remaining goods. Notice that a \uHEF{k} allocation is also \HEF{k} but the converse is not necessarily true. Indeed, in \Cref{prop:HEFk_vs_uHEF}, we will present an instance that, for some $k \in \N$, admits an \HEF{k} allocation but no \uHEF{k} allocation.
	\label{defn:HEF-k}
\end{definition}

\begin{remark}
	It follows from the definitions that \HEF{0} $\Rightarrow$ \HEF{1} $\Rightarrow$ \HEF{2} $\dots$, and that an allocation satisfies \HEF{0} if and only if it satisfies \EF{}. It is also easy to verify that an allocation is \sEF{1} if and only if it is \uHEF{n}. This is because the unique hidden good for every agent is also the one that is (hypothetically) removed under \sEF{1}. Additionally, as discussed in \Cref{sec:Introduction}, an \EF{1} allocation might not be \uHEF{k} for any $k \leq n$.
	\label{rem:EF_HEF_Relationship}
\end{remark}

We say that allocation $A$ is \emph{\HEF{} with respect to set $S$} if $A$ becomes envy-free after hiding the goods in $S$, i.e., for every pair of agents $i,h \in [n]$, we have $v_i(A_i) \geq v_i(A_h \setminus S)$. We say that $k$ goods \emph{must be hidden} under $A$ if $A$ is \HEF{} with respect to some set $S$ such that $|S|=k$, and there is no set $S'$ with $|S'| < k$ such that $A$ is \HEF{} with respect to $S'$.

\begin{definition}[\textbf{Pareto optimality}]
	An allocation $A$ is Pareto dominated by another allocation $B$ if $v_i(B_i) \geq v_i(A_i)$ for every agent $i \in [n]$ with at least one of the inequalities being strict. A \emph{Pareto optimal} (\PO{}) allocation is one that is not Pareto dominated by any other allocation.
\end{definition}

\begin{definition}[\textbf{\EF{1} algorithms}]
	\label{defn:EF1_algorithms}
	We will now describe four known algorithms for finding \EF{1} allocations that are relevant to our work.
	
	\paragraph{Round-robin algorithm (\RR{}):} Fix a permutation $\sigma$ of the agents. The \RR{} algorithm cycles through the agents according to $\sigma$. In each round, an agent gets its favorite good from the pool of remaining goods.

	\paragraph{Envy-graph algorithm (\Envygraph{}):} This algorithm, proposed by \citet{LMM+04approximately}, works as follows: In each step, one of the remaining goods is assigned to an agent that is not envied by any other agent. The existence of such an agent is guaranteed by resolving cyclic envy relations (if any exists) in a combinatorial structure called the \emph{envy-graph} of an allocation.
	
	\paragraph{Fisher market-based algorithm (\Market{}):} This algorithm, due to \citet{BKV18Finding}, uses local search and price-rise subroutines in a Fisher market associated with the fair division instance, and returns an \EF{1} and \PO{} allocation. The bound on running time of this algorithm is pseudopolynomial, i.e., has a polynomial dependence on $v_{i,j}$ instead of $\log v_{i,j}$.
	
	\paragraph{Maximum Nash Welfare solution (\MNW{}):} The \emph{Nash social welfare} of an allocation $A$ is defined as $\NW(A) \coloneqq \left( \prod_{i \in [n]} v_i(A_i) \right)^{1/n}$. The \MNW{} algorithm computes an allocation with the highest Nash social welfare (called a \emph{Nash optimal} allocation). It is known that a Nash optimal allocation is both \EF{1} and \PO{}~\citep{CKM+16unreasonable}.
\end{definition}

\begin{remark}
	\citet{CFS+19group} observed that \RR{}, \Market{}, and \MNW{} algorithms all satisfy \sEF{1}. It is easy to see that \Envygraph{} algorithm is also \sEF{1}. Among these four algorithms, only \MNW{} and \Market{} are provably also \PO{}.\footnote{It is also known that \RR{} and \Envygraph{} fail to satisfy \PO{}; see, e.g., \citep{CFS17fair}.} The allocations computed by all four algorithms have the property that there exists some agent that is not envied by any other agent. Indeed, \MNW{} and \Market{} are both \PO{} and therefore cannot have cyclic envy relations, and \RR{} and \Envygraph{} algorithms have this property by design. For such an agent (not necessarily the same agent for all four algorithms), no good needs to be removed under \sEF{1}. Therefore, from \Cref{rem:EF_HEF_Relationship}, all these algorithms are also envy-free up to $n-1$ uniformly hidden goods, or $\uHEF{(n-1)}$.
	\label{rem:EF1_algorithms_uHEF_n-1}
\end{remark}

\begin{restatable}{prop}{uHEFExistence}
	\label{prop:uHEF(n-1)}
	Given an instance with additive valuations, a $\uHEF{(n-1)}$ allocation always exists and can be computed in polynomial time, and a $\uHEF{(n-1)}+\PO{}$ allocation always exists and can be computed in pseudopolynomial time.
\end{restatable}

\begin{remark}
	Note that for any $k < n-1$, an \HEF{k} allocation might fail to exist. Indeed, with $n$ agents that have identical and positive valuations for $m = n-1$ goods, some agent will surely miss out and force the allocation to hide all $n-1$ (i.e., $k+1$ or more) goods. Therefore, the bound in \Cref{prop:uHEF(n-1)} for \uHEF{k} (and hence, for \HEF{k}) is tight in terms of $k$.
	\label{rem:HEFk_Tight_Bound_For_k}
\end{remark}

\subsection{Relevant Computational Problems}
\label{subsec:Computational_Problems}

\Cref{defn:HEFkExistence} formalizes the decision problem of checking whether a given instance admits a fair (i.e., \HEF{k}) allocation.

\begin{definition}[\textbf{\HEFkExistence}]
	Given an instance $\I$, does there exist an allocation $A$ and a set $S \subseteq [m]$ of at most $k$ goods such that $A$ is \HEF{} with respect to $S$?
	\label{defn:HEFkExistence}
\end{definition}

Notice that a certificate for \HEFkExistence{} consists of an allocation $A$ as well as a set $S$ of at most $k$ hidden goods.

Another relevant computational question involves checking whether a given allocation $A$ is \HEF{} with respect to some set $S \subseteq [m]$ of at most $k$ goods.

\begin{definition}[\textbf{\HEFkVerification}]
	Given an instance $\I$ and an allocation $A$, does there exist a set $S \subseteq [m]$ of $k$ goods such that $A$ is \HEF{} with respect to $S$?
	\label{defn:HEFkVerification}
\end{definition}

For additive valuations, both \HEFkExistence{} and \HEFkVerification{} are in \NP{}. The next problem pertains to the existence of envy-free allocations.

\begin{definition}[\textbf{\EFExistence}]
	Given an instance $\I$, does there exist an envy-free allocation for $\I$?
	\label{defn:EFExistence}
\end{definition}

\EFExistence{} is known to be \NPC{}~\citep{LMM+04approximately}. From \Cref{rem:EF_HEF_Relationship}, it follows that \HEFkExistence{} is \NPC{} when $k=0$ for additive valuations.

\section{Theoretical Results}
\label{sec:Theoretical_Results}

This section presents our theoretical results concerning the existence and computation of \HEF{k} and \uHEF{k} allocations. We will first show that \uHEF{k} is a strictly more demanding notion than \HEF{k} (\Cref{prop:HEFk_vs_uHEF}).

\begin{restatable}{prop}{HEFvsuHEF}
	\label{prop:HEFk_vs_uHEF}
	There exists an instance $\I$ that, for some fixed $k \in \N$, admits an \HEF{k} allocation but no \uHEF{k} allocation.
\end{restatable}
\begin{proof}
	Consider the fair division instance $\I$ with five agents $a_1,\dots,a_5$ and six goods $g_1,\dots,g_6$ shown in \Cref{tab:HEFk_vs_uHEFk}. Observe that the allocation $A = (A_1,\dots,A_5)$ with $A_1 = \{g_1,g_2\}$, $A_2 = \{g_3\}$, $A_3 = \{g_4\}$, $A_4 = \{g_5\}$, $A_5 = \{g_6\}$ satisfies \HEF{2} with respect to the set $S = \{g_1,g_2\}$.
	
	\begin{table}[h]
		\centering
		\small
		\begin{tabular}{ c|cccccc }
			& $g_1$ & $g_2$ & $g_3$ & $g_4$ & $g_5$ & $g_6$\\ \hline
			$a_1$ & $1$ & $1$ & $2$ & $0$ & $0$ & $0$ \\
			$a_2$ & $1$ & $1$ & $2$ & $0$ & $0$ & $0$ \\
			$a_3$ & $10$ & $10$ & $1$ & $1$ & $1$ & $1$ \\
			$a_4$ & $10$ & $10$ & $1$ & $1$ & $1$ & $1$ \\
			$a_5$ & $10$ & $10$ & $1$ & $1$ & $1$ & $1$ \\
		\end{tabular}
		\caption{The instance used in the proof of \Cref{prop:HEFk_vs_uHEF}.}
		\label{tab:HEFk_vs_uHEFk}
	\end{table}
	
	We will show that $\I$ does not admit a \uHEF{2} allocation. Suppose, for contradiction, that there exists an allocation $B$ satisfying \uHEF{2}. Then, $B$ must hide $g_1$ and $g_2$ (otherwise, at least one of $a_3$, $a_4$ or $a_5$ will envy the owner(s) of these goods). Thus, in particular, the good $g_3$ must be revealed by $B$. Assume, without loss of generality, that $g_3$ is \emph{not} assigned to $a_1$ in $B$ (otherwise, a similar argument can be carried out for $a_2$). Then, $B$ must assign both $g_1$ and $g_2$ to $a_1$ (so that $a_1$ does not envy the owner of $g_3$). However, this violates the one-hidden-good-per-agent property of \uHEF{k}, which is a contradiction.
\end{proof}

Recall from \Cref{subsec:Computational_Problems} that \HEFkExistence{} is \NPC{} when $k=0$. This still leaves open the question whether \HEFkExistence{} is \NPC{} for \emph{any} fixed $k \in \N$. Our next result (\Cref{thm:HEFk_Existence_NPcomplete_IdenticalVals}) shows that this is indeed the case, even under the restricted setting of \emph{identical} valuations (i.e., for every $j \in [m]$, $v_{i,j}=v_{h,j}$ for every $i,h \in [n]$).

\begin{restatable}[\textbf{Hardness of \HEFkExistence{}}]{theorem}{HEFkExistenceNPcompleteIdenticalVals}
	\label{thm:HEFk_Existence_NPcomplete_IdenticalVals}
	For any fixed $k \in \N$, \HEFkExistence{} is \NPC{} even for identical valuations.
\end{restatable}
\begin{proof}
	We will show a reduction from \Partition{}, which is known to be \NPC{}~
	\citep{GJ79computers}. An instance of \Partition{} consists of a multiset $X = \{x_1,x_2,\dots,x_n\}$ with $x_i \in \N$ for all $i \in [n]$. The goal is to determine whether there exists $Y \subset X$ such that $\sum_{x_i \in Y} x_i = \sum_{x_i \in X \setminus Y} x_i = T$, where $T \coloneqq \frac{1}{2} \sum_{x_i \in X} x_i$.
	
	We will construct a fair division instance with $k+3$ agents $a_1,\dots,a_{k+3}$ and $n+k+1$ goods. The goods are classified into $n+1$ \emph{main goods} $g_1,\dots,g_{n+1}$ and $k$ \emph{dummy goods} $d_1,\dots,d_k$. The (identical) valuations are defined as follows: Every agent values the goods $g_1,\dots,g_n$ at $x_1,\dots,x_n$ respectively; the good $g_{n+1}$ at $T$, and each dummy good at $4T$.
	
	($\Rightarrow$) Suppose $Y$ is a solution of \Partition{}. Then, an \HEF{k} allocation can be constructed as follows: Assign the main goods corresponding to the set $Y$ to agent $a_1$ and those corresponding to $X \setminus Y$ to agent $a_2$. The good $g_{n+1}$ is assigned to agent $a_3$. Each of the remaining $k$ agents is assigned a unique dummy good. Note that every agent in the set $\{a_1,a_2,a_3\}$ envies every agent in the set $\{a_4,\dots,a_{k+3}\}$, and these are the only pairs of agents with non-zero envy. Therefore, the allocation can be made envy-free by hiding the $k$ dummy goods, i.e., the allocation is \HEF{} with respect to the set $\{d_1,\dots,d_k\}$.
	
	($\Leftarrow$) Now suppose there exists an \HEF{k} allocation $A$. Since there are $k$ dummy goods and $k+3$ agents, there must exist at least three agents that do not receive any dummy good in $A$. Without loss of generality, let these agents be $a_1$, $a_2$ and $a_3$ (otherwise, we can reindex). We claim that all dummy goods must be hidden under $A$. Indeed, agent $a_1$ does not receive any dummy good, and therefore its maximum possible valuation can be $v(g_1 \cup \dots \cup g_{n+1}) = 3T < v(d_j)$ for any dummy good $d_j$. If some dummy good $d_j$ is not hidden, then $a_1$ will envy the owner of $d_j$, contradicting \HEF{k}. Therefore, all dummy goods must be hidden, and since there are $k$ such goods, these are the only ones that can be hidden.
	
	The above observation implies that the good $g_{n+1}$ must be revealed by $A$. Furthermore, $g_{n+1}$ must be assigned to one of $a_1$, $a_2$ or $a_3$ (otherwise, by pigeonhole principle, one of these agents will have valuation at most $\frac{2T}{3}$ and will envy the owner of $g_{n+1}$). If $g_{n+1}$ is assigned to $a_3$, then the remaining main goods $g_1,\dots,g_n$ must be divided between $a_1$ and $a_2$ such that $v(A_1) \geq T$ and $v(A_2) \geq T$. This gives a partition of the set $X$.
\end{proof}

Another commonly used preference restriction is that of \emph{binary} valuations (i.e., for every $i \in [n]$ and $j \in [m]$, $v_{i,j} \in \{0,1\}$). We note that even under this restriction, \HEFkExistence{} remains \NPC{} when $k=0$ (\Cref{cor:HEFk_Existence_NPcomplete_BinaryVals}). This observation follows from a result of \citet{AGM+15fair}, who showed that determining the existence of an envy-free allocation is \NPC{} even for binary valuations (\Cref{prop:EF_Existence_NPcomplete_BinaryValuations}). We provide an alternative proof of this statement in \Cref{subsec:Proof_EF_Existence_NPcomplete_BinaryValuations} in the appendix.

\begin{restatable}[\citealp{AGM+15fair}; Theorem 11]{prop}{EFExistenceNPcompleteBinaryVals}
	\label{prop:EF_Existence_NPcomplete_BinaryValuations}
	\EFExistence{} is \NPC{} even for binary valuations.
\end{restatable}

\begin{restatable}{corollary}{HEFkExistenceNPcompleteBinaryVals}
	\label{cor:HEFk_Existence_NPcomplete_BinaryVals}
	For $k=0$, \HEFkExistence{} is \NPC{} even for binary valuations.
\end{restatable}

\Cref{prop:EF_Existence_NPcomplete_BinaryValuations} is also useful in establishing the computational hardness of finding an \HEF{k}+\PO{} allocation. Note that unlike \Cref{cor:HEFk_Existence_NPcomplete_BinaryVals}, \Cref{prop:HEFkPOExistence_NPcomplete_BinaryVals} holds for every fixed $k \in \N$.

\begin{restatable}[\textbf{Hardness of \HEF{k}+\PO{}}]{theorem}{HEFkPOExistenceNPcompleteBinaryVals}
	\label{prop:HEFkPOExistence_NPcomplete_BinaryVals}
	Given any instance $\I$ with binary valuations and any fixed $k \in \N \cup \{0\}$, it is \NPH{} to determine if $\I$ admits an allocation that is envy-free up to $k$ hidden goods $(\HEF{k})$ and Pareto optimal $(\PO{})$.
\end{restatable}
\begin{proof} (Sketch)
	Starting from any instance of \EFExistence{} with binary valuations (\Cref{prop:EF_Existence_NPcomplete_BinaryValuations}), we add to it $k$ new goods and $k+1$ new agents such that all new goods are approved by all new agents (and no one else). Also, the new agents have zero value for the existing goods. In the forward direction, an arbitrary allocation of new goods among the new agents works. In the reverse direction, \PO{} forces each new (respectively, existing) good to be assigned among new (respectively, existing) agents only. The imbalance between new agents and new goods means that all (and only) the new goods must be hidden. Then, the restriction of the \HEF{k} allocation to the existing agents/goods gives the desired \EF{} allocation.
\end{proof}

We will now proceed to analyzing the computational complexity of \HEFkVerification{}. Here, we show a hardness-of-approximation result (\Cref{prop:HEFk_Verification_Hardness_Of_Approximation_BinaryVals}). The inapproximability factor is stated in terms of the aggregate envy, defined as follows: Given any allocation $A$, the \emph{aggregate envy} in $A$ is the sum of all pairwise envy values, i.e., 
\begin{align*}
	E \coloneqq \sum_{h \in [n]} \sum_{i \neq h} \max\{0, v_i(A_h) - v_i(A_i)\}.
\end{align*}

Note that \HEFkVerification{} is stated as a decision problem (\Cref{defn:HEFkVerification}). However, one can consider an approximation version of this problem as follows: A $c$-approximation algorithm for \HEFkVerification{} takes as input a fair division instance and an allocation, and computes a set of goods of size at most $c \cdot \kopt$, where $\kopt$ is the size of the smallest hidden set for the given allocation. Under this definition, \Cref{prop:HEFk_Verification_Hardness_Of_Approximation_BinaryVals} can be interpreted as follows: Given any $\eps>0$, there is no polynomial-time $(1-\eps).\ln E$-approximation algorithm for \HEFkVerification{}, unless P=NP.

\begin{restatable}[\textbf{\HEFkVerification{} inapproximability}]{theorem}{HEFkVerificationHardnessOfApproximationBinaryVals}
	\label{prop:HEFk_Verification_Hardness_Of_Approximation_BinaryVals}
	Given any $\eps > 0$, it is \NPhard{} to approximate \HEFkVerification{} to within $(1-\eps) \cdot \ln E$ even for binary valuations, where $E$ is the aggregate envy in the given allocation.
\end{restatable}
\begin{proof}
	We will show a reduction from \HittingSet{}. An instance of \HittingSet{} consists of a finite set $X = \{x_1,\dots,x_p\}$, a collection $\F = \{F_1,\dots,F_q\}$ of subsets of $X$, and some $k \in \N$. The goal is to determine whether there exists $Y \subseteq X$, $|Y| \leq k$ that intersects every member of $\F$ (i.e., for every $F \in \F$, $Y \cap F \neq  \emptyset$). It is known that given any $\eps > 0$, it is \NPhard{} to approximate \HittingSet{} to within a factor $(1-\eps) \cdot \ln |\F|$ \citep{DS14analytical}.
	
	We will construct a fair division instance with $n = q+1$ agents and $m = p + \sum_{i=1}^q (|F_i|-1)$ goods. The agents are classified into $q$ \emph{dummy agents} $a_1,\dots,a_q$ and one \emph{main agent} $a_{q+1}$. The goods are classified into $p$ \emph{main goods} $g_1,\dots,g_p$ and $q$ distinct sets of dummy goods, where the $i^\text{th}$ set consists of the goods $f^i_{1},\dots,f^i_{|F_i|-1}$.
	
	The valuations are as follows: The main agent approves all the main goods, i.e., for all $j \in [p]$, $v_{q+1}(\{g_j\}) = 1$. Each dummy agent $a_i$ approves the dummy goods in the $i^\text{th}$ set as well as those main goods that intersect with $F_i$, i.e., for every $i \in [q]$, $v_i(\{f^i_j\}) = 1$ for all $j \in [|F_i|-1]$, and $v_i(\{g_j\}) = 1$ whenever $x_j \in F_i$. All other valuations are set to $0$. 
	
	The input allocation $A = (A_1,\dots,A_{q+1})$ is defined as follows: The main agent $a_{q+1}$ is assigned all the main goods, i.e., $A_{q+1} \coloneqq \{g_1,\dots,g_p\}$. For every $i \in [q]$, the dummy agent $a_i$ is assigned the $|F_i|-1$ dummy goods in the $i^\text{th}$ set, i.e., $A_i \coloneqq \{f^i_{1},\dots,f^i_{|F_i|-1}\}$. Note that in the allocation $A$, each dummy agent envies the main agent by one approved good, and these are the only pairs of agents with envy.
	
	($\Rightarrow$) Suppose $Y \subseteq X$, $|Y| \leq k$ is solution of the \HittingSet{} instance. We claim that the allocation $A$ is \HEF{} with respect to the set $S \coloneqq \{g_j : x_j \in Y\}$ with $|S| \leq k$. Indeed, since $S$ is induced by a hitting set, each dummy agent approves at least one good in $S$. Therefore, by hiding the goods in $S$, the envy from the dummy agents can be eliminated.
	
	($\Leftarrow$) Now suppose there exists $S \subseteq [m]$, $|S| \leq k$ such that $A$ is \HEF{} with respect to $S$. Then, for every $i \in [q]$, the set $S$ must contain at least one good that is approved by the dummy agent $a_i$ (otherwise $A$ will not be envy-free after hiding the goods in $S$). It is easy to see that the set $Y \coloneqq \{x_j : g_j \in S\}$ constitutes the desired hitting set of cardinality at most $k$.
	
	Finally, to show the hardness-of-approximation, notice that the aggregate envy in $A$ is $q$ because each dummy agent envies the main agent by one unit of utility. The claim now follows by substituting $|\F| = q = E$ in the inapproximability result of \HittingSet{} stated above.
\end{proof}

Our next result (\Cref{thm:HEFk_Verification_ApproxAlgo}) provides an approximation algorithm that (nearly) matches the hardness-of-approximation result in \Cref{prop:HEFk_Verification_Hardness_Of_Approximation_BinaryVals}. We remark that the algorithm in \Cref{thm:HEFk_Verification_ApproxAlgo} applies to \emph{any} instance with additive and possibly non-binary valuations.

\begin{restatable}[\textbf{Approximation algorithm}]{theorem}{HEFkVerificationApproxAlgo}
	\label{thm:HEFk_Verification_ApproxAlgo}
	There is a polynomial-time algorithm that, given as input any instance of \HEFkVerification{}, finds a set $S \subseteq [m]$ with $|S| \leq \kopt \cdot \ln E + 1$ such that the given allocation is \HEF{} with respect to $S$. Here, $E$ and $\kopt$ denote the aggregate envy and the number of goods that must be hidden under the given allocation, respectively.
\end{restatable}

The proof of \Cref{thm:HEFk_Verification_ApproxAlgo} is deferred to \Cref{subsec:Proof_HEFk_Verification_ApproxAlgo} in the appendix but a brief idea is as follows: For any set $S \subseteq [m]$, define the \emph{residual envy function} $f : 2^{[m]} \rightarrow \mathbb{R}$ so that $f(S)$ is the aggregate envy in allocation $A$ after hiding the goods in $S$. That is,
\begin{align*}
f(S) \coloneqq \sum_{h \in [n]} \sum_{i \neq h} \max\{0, v_i(A_h \setminus S) - v_i(A_i)\}.
\end{align*}
The relevant observation is that $f$ is \emph{supermodular}. Given this observation, the approximation guarantee in \Cref{thm:HEFk_Verification_ApproxAlgo} can be obtained by the standard greedy algorithm for submodular maximization, or, equivalently, supermodular minimization~ \citep{NWF78analysis}; see Algorithm~\ref{alg:Greedy_HEFk_ApproxAlgo} in \Cref{subsec:Proof_HEFk_Verification_ApproxAlgo}.

\section{Experimental Results}
\label{sec:Experiments}

\begin{table*}
	\centering
	\begin{tabular}{|cccc|}
		\multicolumn{4}{c}{\textbf{Normalized average-case regret}}\\
		\hline
		\footnotesize{\Market{}} & \footnotesize{\RR{}} & \footnotesize{\MNW{}} & \footnotesize{\Envygraph{}}\\
		\includegraphics[width=0.22\textwidth]{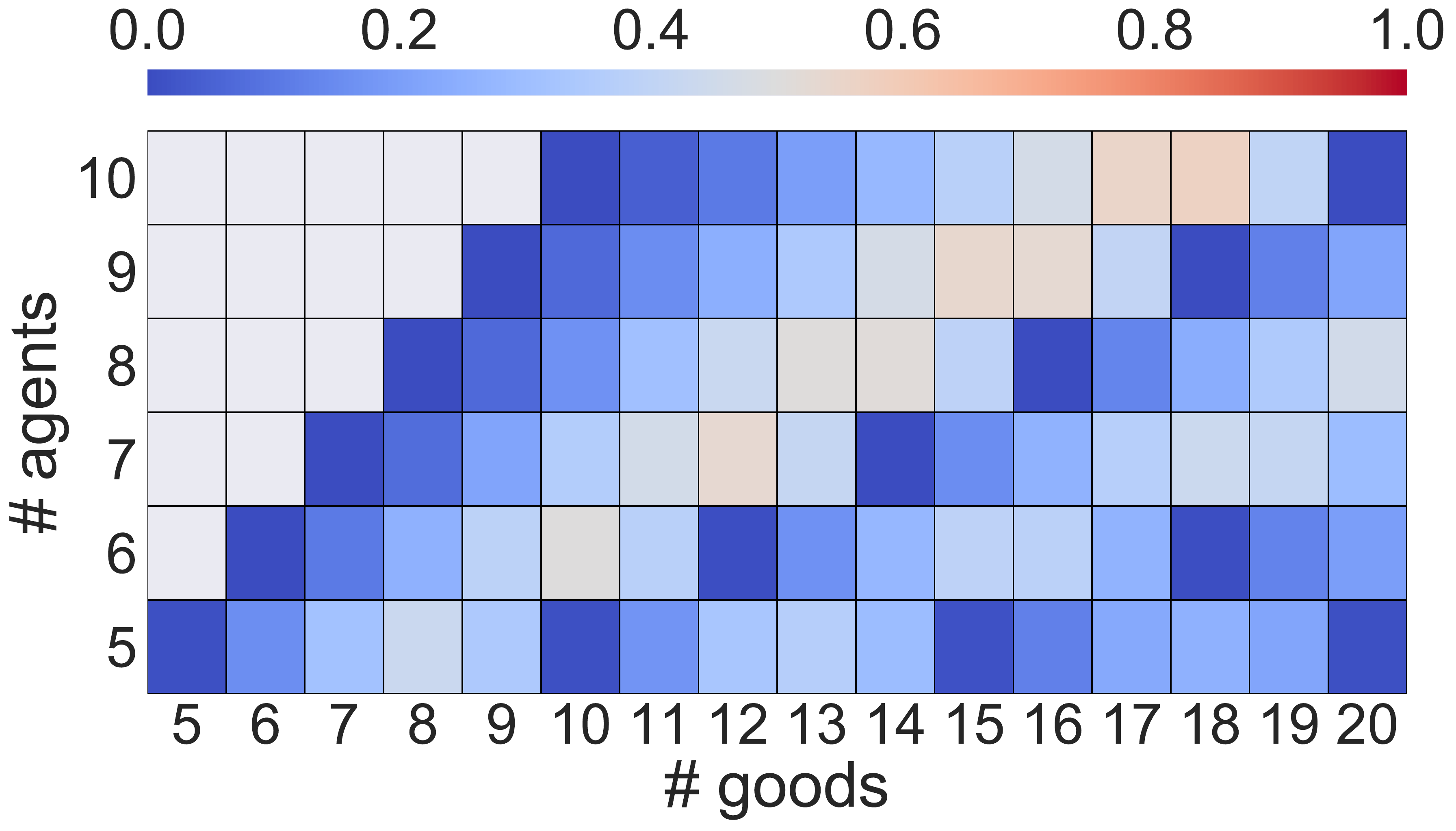} &
		\includegraphics[width=0.22\textwidth]{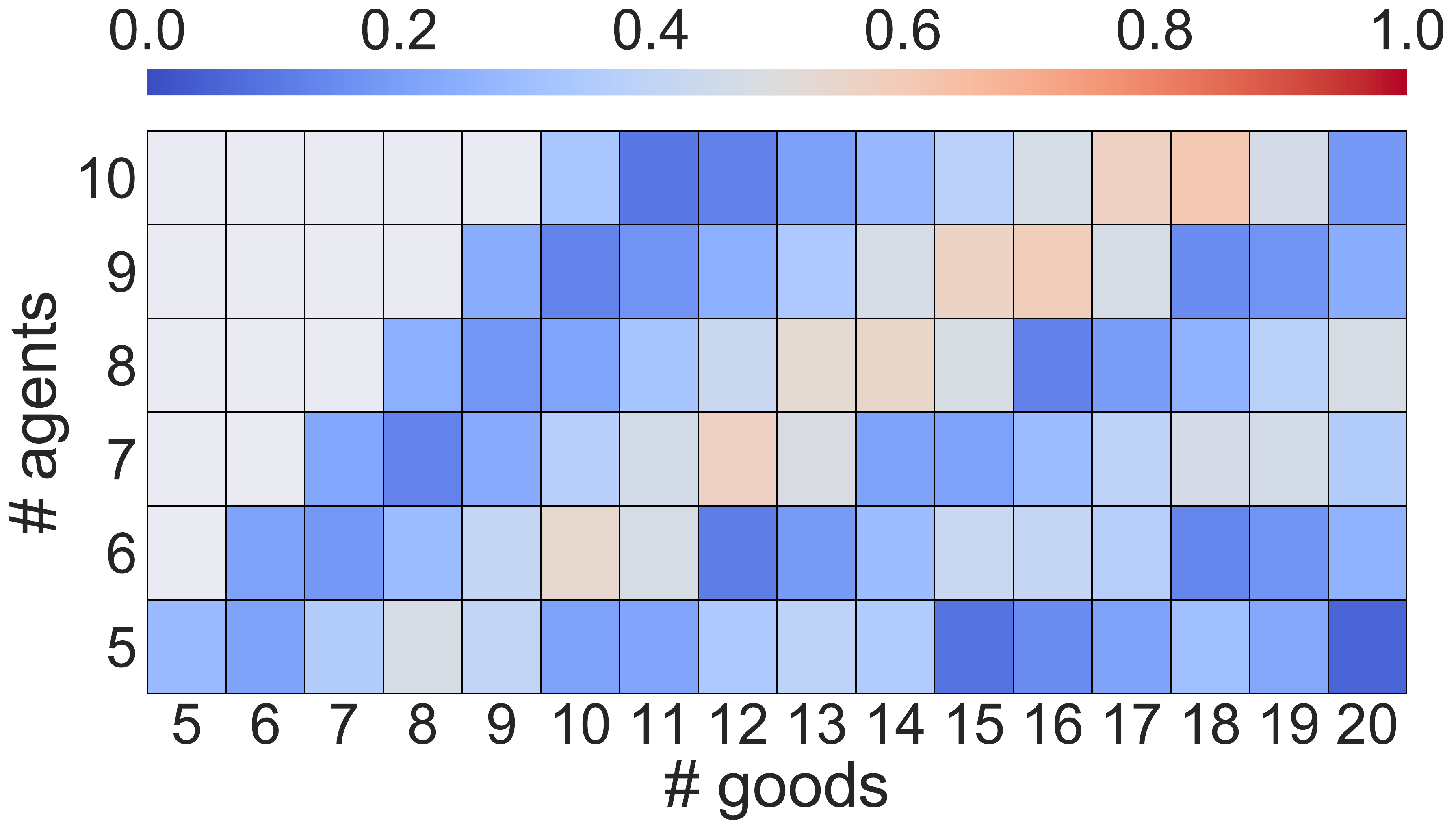} &
		\includegraphics[width=0.22\textwidth]{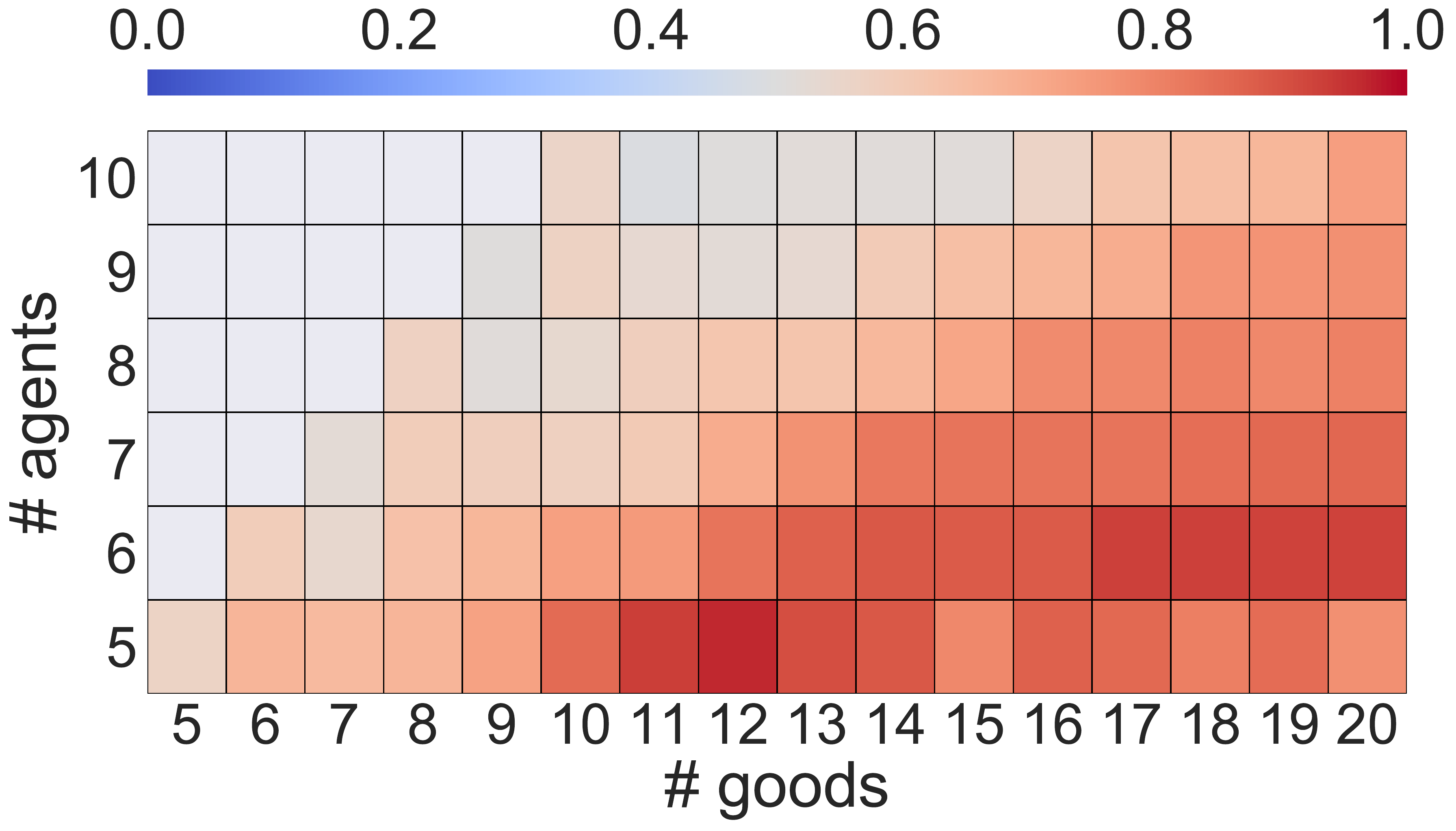} &
		\includegraphics[width=0.22\textwidth]{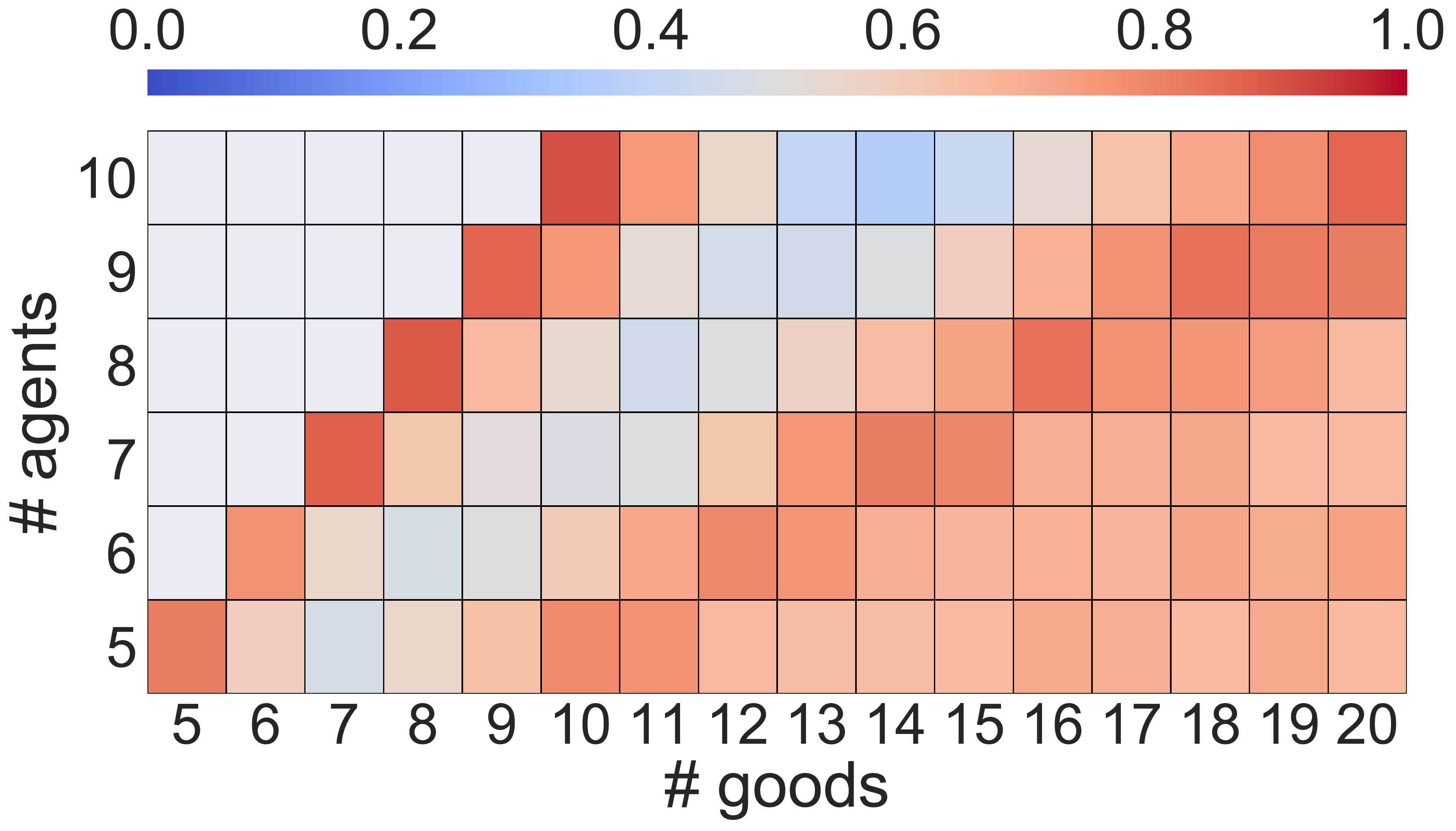} \\
		\hline
		\multicolumn{4}{c}{}\\
		\multicolumn{4}{c}{\textbf{Number of goods that must be hidden on average} (averaged over non-\EF{} instances only)}\\
		\hline
		\footnotesize{\Market{}} & \footnotesize{\RR{}} & \footnotesize{\MNW{}} & \footnotesize{\Envygraph{}}\\
		\includegraphics[width=0.22\textwidth]{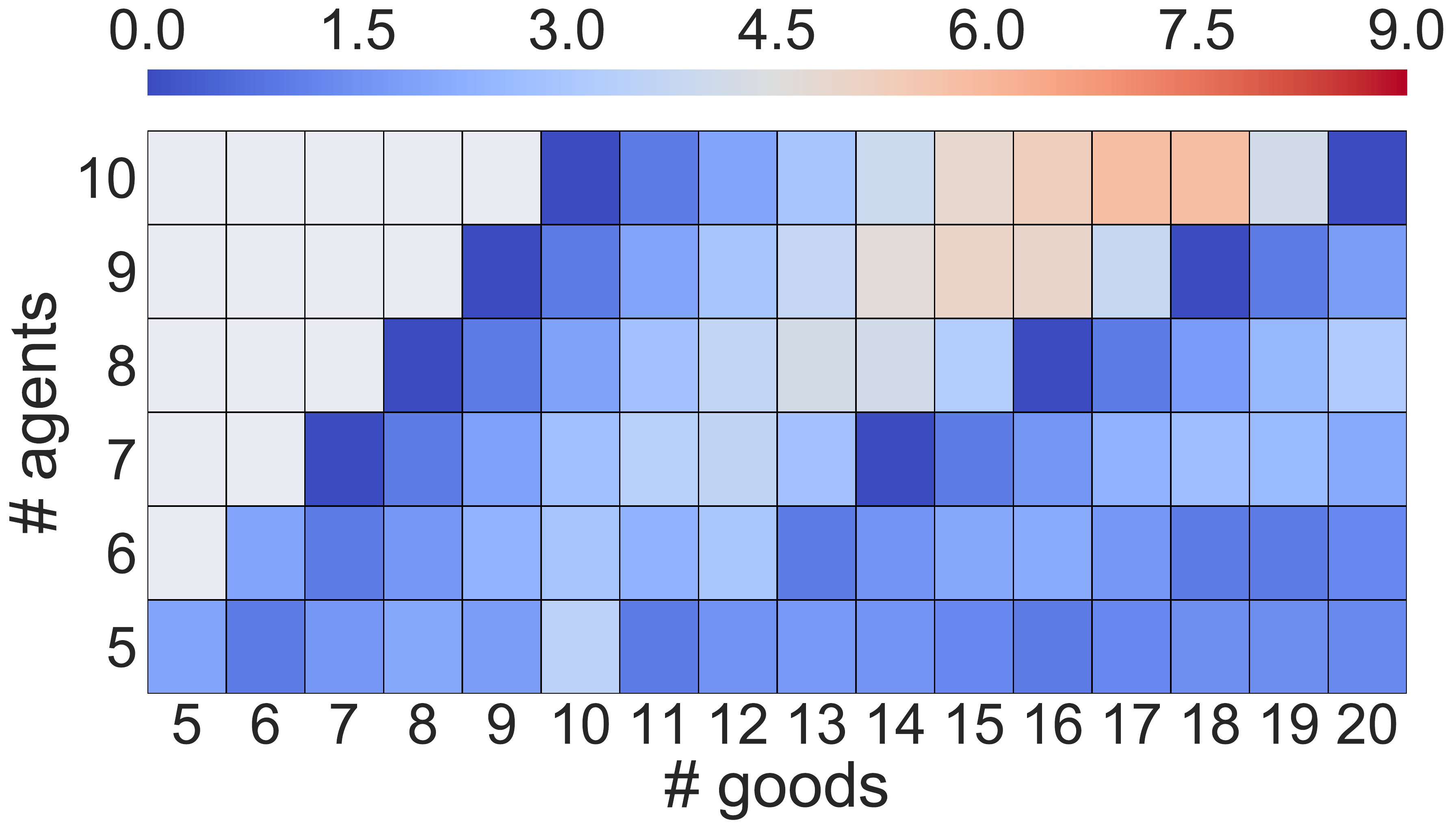} & \includegraphics[width=0.22\textwidth]{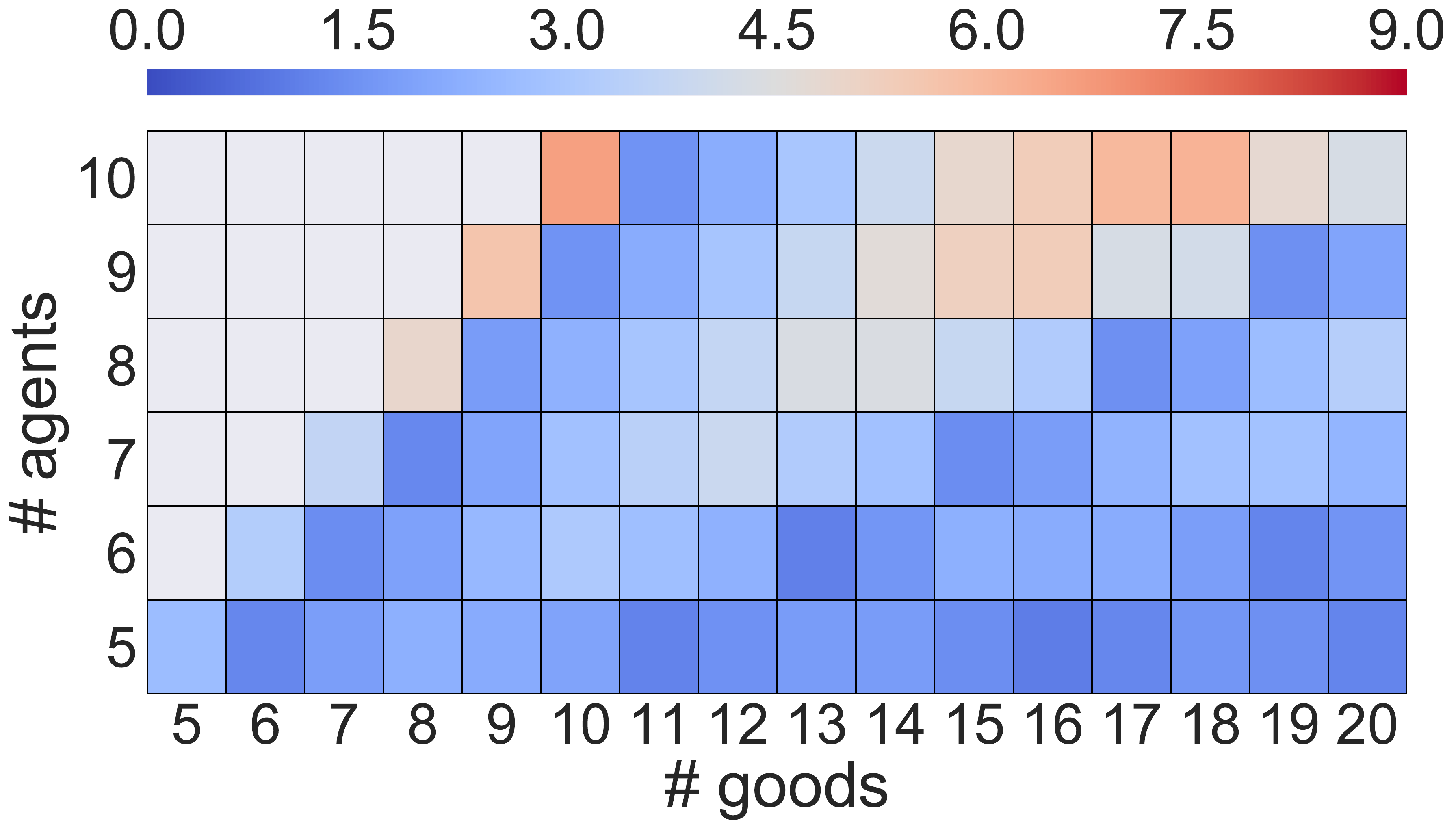} &
		\includegraphics[width=0.22\textwidth]{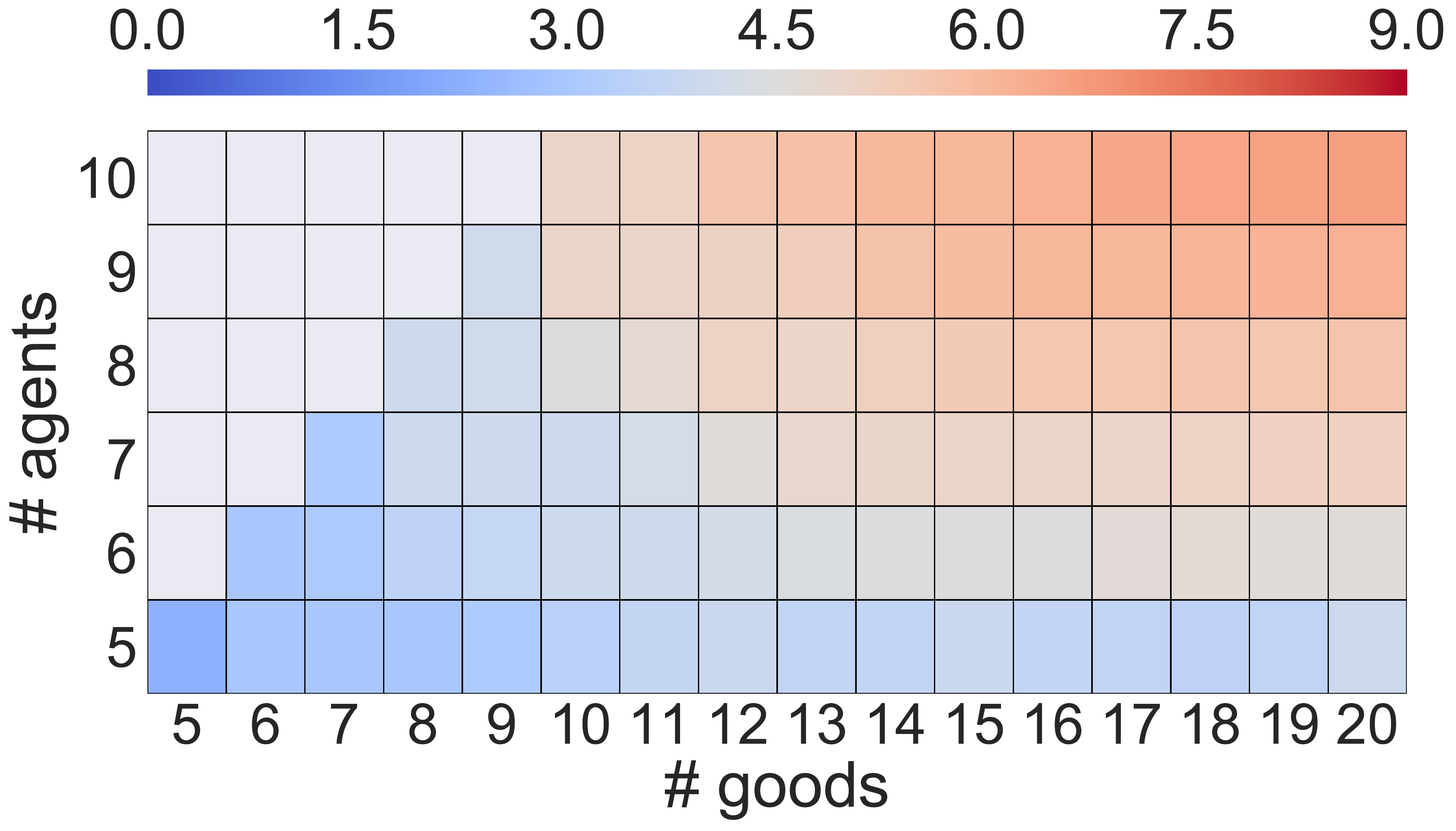} &
		\includegraphics[width=0.22\textwidth]{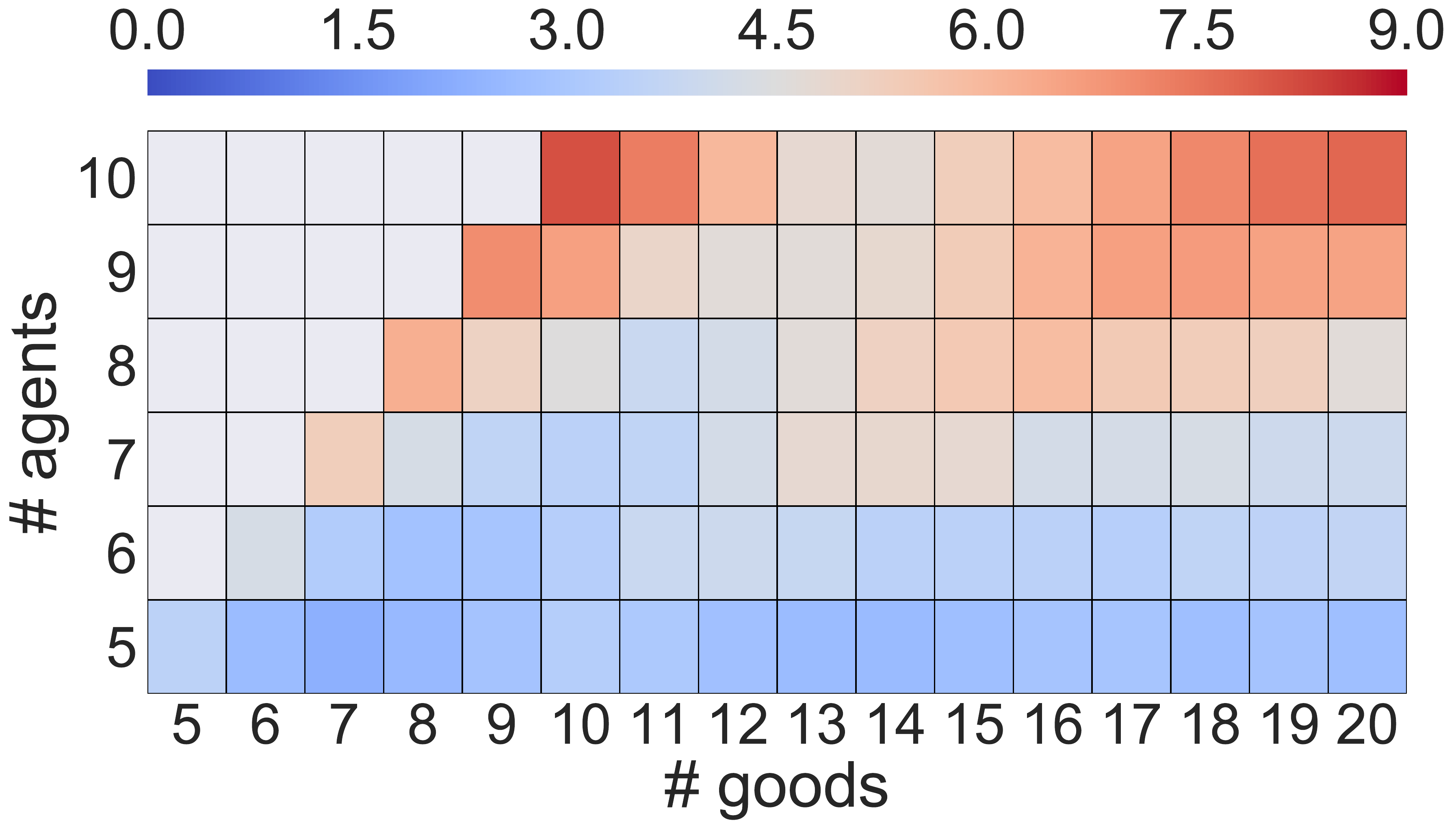} \\
		\hline
	\end{tabular}
	\caption{Results for synthetic data.}
	\label{tab:Expt_BinaryVals_bias_0.7}
\end{table*}

We have seen that the worst-case computational results for \HEF{k}, even in highly restricted settings, are largely negative (\Cref{sec:Theoretical_Results}). In this section, we will examine whether the known algorithms for computing approximately envy-free allocations---in particular, the four \EF{1} algorithms described in \Cref{defn:EF1_algorithms} in \Cref{sec:Preliminaries}---can provide meaningful approximations to \HEF{k} in practice. Recall from \Cref{rem:EF1_algorithms_uHEF_n-1} that all four discussed algorithms---\RR{}, \MNW{}, \Market{}, and \Envygraph{}---satisfy $\uHEF{(n-1)}$.

We evaluate each algorithm in terms of (a) its \emph{regret} (defined below), and (b) the \emph{number of goods that the algorithm must hide}. Given an instance $\I$ and an allocation $A$, let $\kappa(A,\I)$ denote the number of goods that must be hidden under $A$. The \emph{regret} of allocation $A$ is the number of extra goods that must be hidden under $A$ compared to the optimal. That is, $\reg(A,\I) \coloneqq \kappa(A,\I) - \min_{B} \kappa(B,\I)$. Similarly, given an algorithm \Alg{}, the regret of \Alg{} is given by $\reg(\Alg(\I),\I)$, where $\Alg(\I)$ is the allocation returned by \Alg{} for the input instance $\I$. Note that the regret can be large due to the suboptimality of an algorithm, but also due to the size of the instance. To negate the effect of the latter, we normalize the regret value by $n-1$, which is the worst-case upper bound on the number of hidden goods for all four algorithms of interest.

\subsection{Experiments on Synthetic Data}
\label{subsec:Expt_Synthetic}
The setup for synthetic experiments is similar to that used in \Cref{fig:HEFk_motivation_0.7}. Specifically, the number of agents, $n$, is varied from $5$ to $10$, and the number of goods, $m$, is varied from $5$ to $20$ (we ignore the cases where $m < n$). For every fixed $n$ and $m$, we generated $100$ instances with \emph{binary} valuations drawn i.i.d. from Bernoulli distribution with parameter $0.7$ (i.e., $v_{i,j} \sim \Ber(0.7)$). \Cref{tab:Expt_BinaryVals_bias_0.7} shows the heatmaps of the normalized regret (averaged over $100$ instances) and the number of goods that must be hidden (averaged over non-\EF{} instances, i.e., whenever $k \geq 1$) for all four algorithms.\footnote{Additional results for $v_{i,j} \sim \Ber(0.7)$, and $v_{i,j} \sim \Ber(0.5)$ can be found in \Cref{subsec:Additional_Experiments} in the appendix.}

It is clear that \Market{} and \RR{} algorithms have a superior performance than \MNW{} and \Envygraph{}. In particular, both \Market{} and \RR{} have small normalized regret, suggesting that they hide close-to-optimal number of goods. Additionally, the number of hidden goods itself is small for these algorithms (in most cases, no more than \emph{three} goods need to be hidden), suggesting that the worst-case bound of $n-1$ is unlikely to arise in practice. Overall, our experiments suggest that \Market{} and \RR{} can achieve useful approximations to \HEF{k} in practice, especially in comparison to \MNW{} and \Envygraph{}.\footnote{In \Cref{subsec:MNW_Large_Regret} in the appendix, we provide two families of instances where the normalized worst-case regret of \MNW{} is large.}

\subsection{Experiments on Real-World Data}
\label{subsec:Expt_Spliddit}
For experiments with real-world data, we use the data from the popular fair division website \emph{Spliddit} \citep{GP15spliddit}. The Spliddit data has $2212$ instances in total, where the number of agents $n$ varies between $3$ and $10$, and the number of goods $m \geq n$ varies between $3$ and $93$. Unlike the synthetic data, the distribution of instances here is rather uneven (see \Cref{fig:Spliddit_data_distribution} in \Cref{subsec:Additional_Experiments} in the appendix); in fact, $1821$ of the $2212$ instances have $n=3$ agents and $m=6$ goods. Therefore, instead of using heatmaps, we compare the algorithms in terms of their normalized regret (averaged over the entire dataset) and the cumulative distribution function of the hidden goods (see \Cref{fig:Results_Spliddit}).

\Cref{fig:Results_Spliddit} presents an interesting twist: \MNW{} is now the best performing algorithm, closely followed by \RR{} and \Market{}. For any fixed $k$, the fraction of instances for which these three algorithms compute an \HEF{k} allocation is also nearly identical. As can be observed, these algorithms almost never need to hide more than \emph{three} goods. By contrast, \Envygraph{} has the largest regret and significantly worse cumulative performance. Therefore, once again, \Market{} and \RR{} algorithms perform competitively with the optimal solution, making them attractive options for achieving fair outcomes without withholding too much information.

\begin{figure}
	\centering
	\begin{subfigure}[b]{0.46\linewidth}
		\centering
		\includegraphics[width=\linewidth]{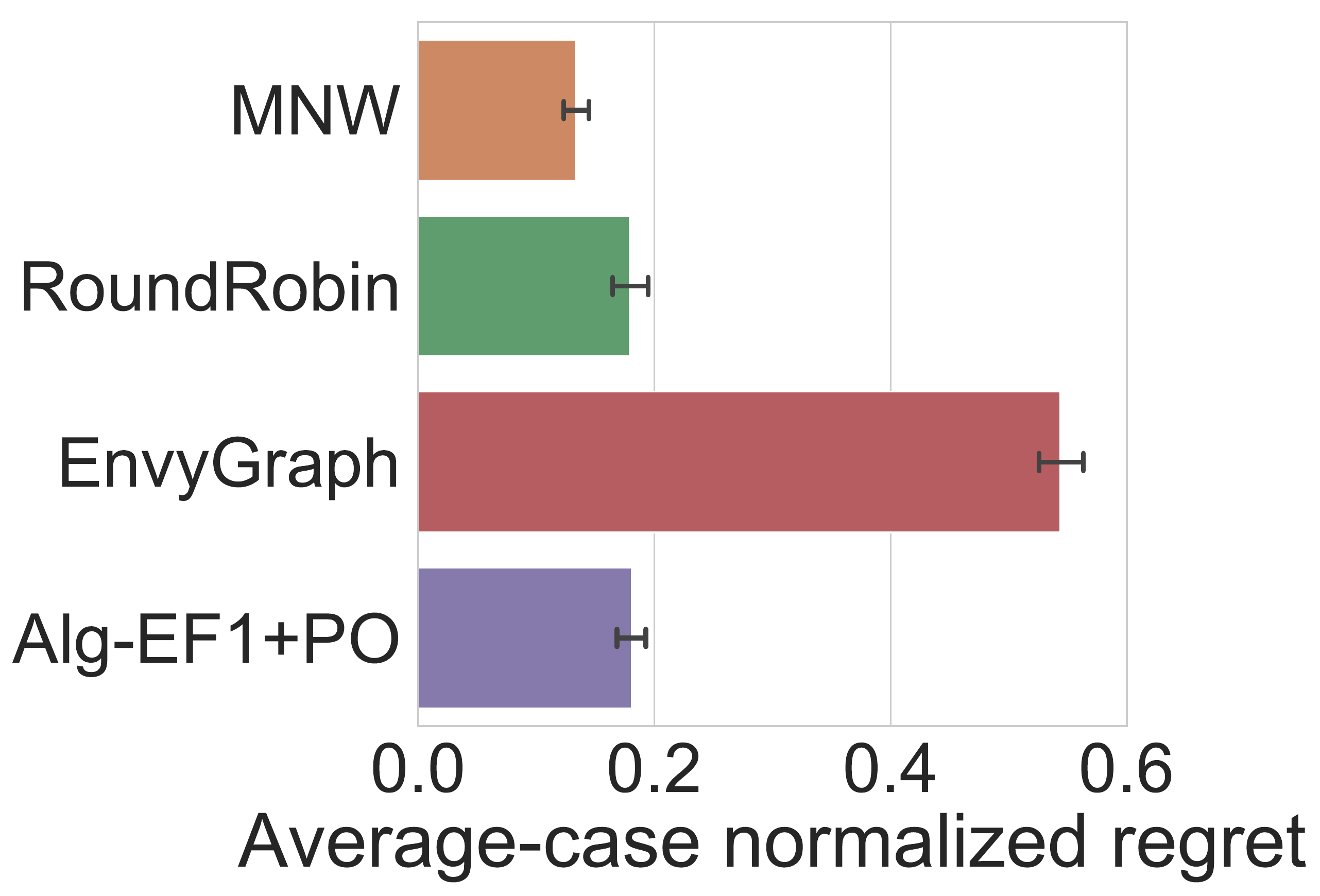}
	\end{subfigure}
	~
	\begin{subfigure}[b]{0.51\linewidth}
		\centering
		\includegraphics[width=\linewidth]{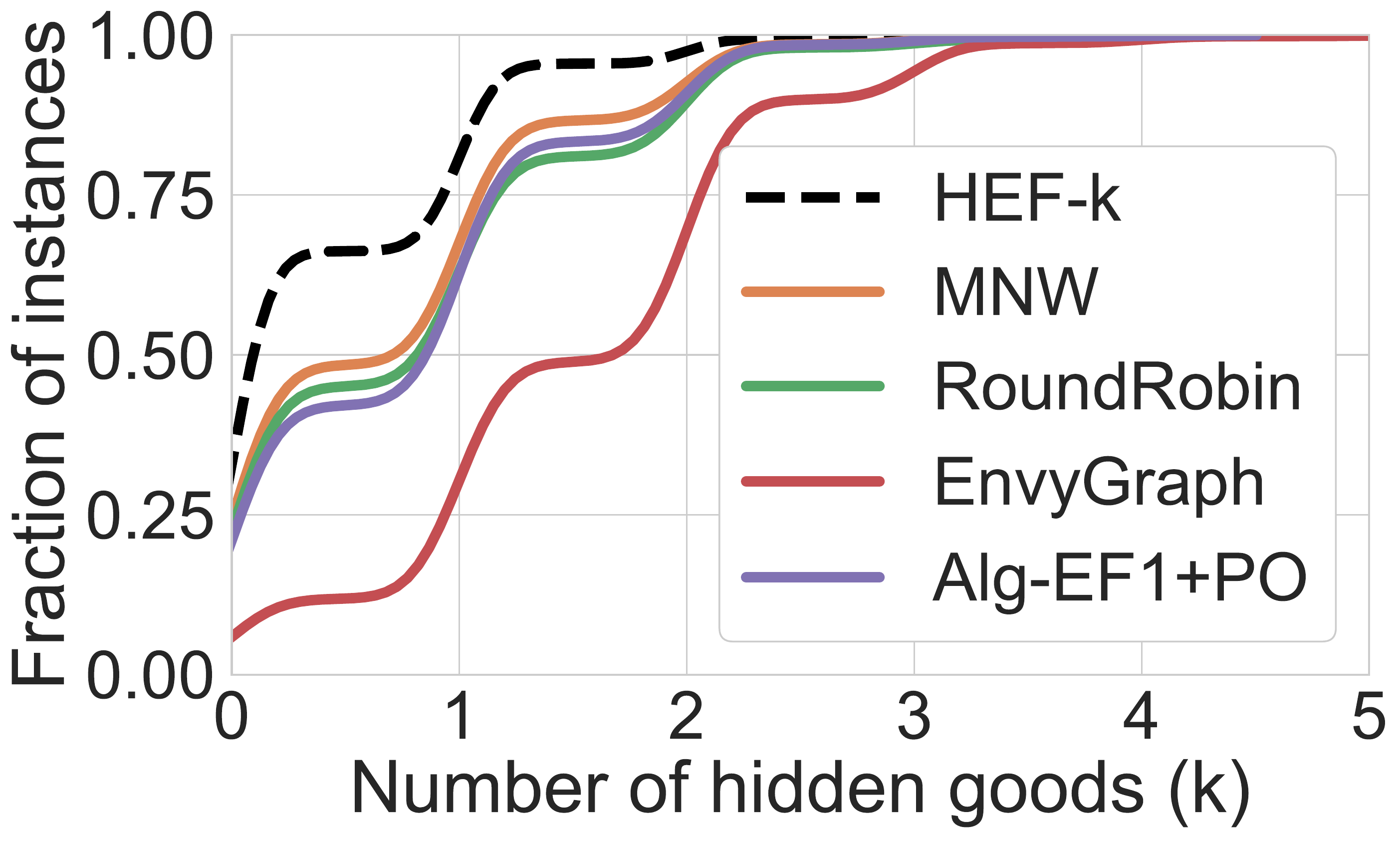}
	\end{subfigure}
	\caption{Results for Spliddit data.}
	\label{fig:Results_Spliddit}
\end{figure}

\section{Future Work}
Analyzing the asymptotic behavior of \HEF{k} allocations, as has been done for envy-free allocations~\citep{DGK+14computational,MS19when}, is an interesting direction for future work. It would also be interesting to explore the connections with other recently proposed relaxations that involve discarding goods~\citep{CGH19envy,CKM+20little} or sharing a small subset of goods~\citep{SE19fair}.

\section*{Acknowledgments}
We thank the anonymous conference reviewers for their helpful comments. We are grateful to Ariel Procaccia and Nisarg Shah for sharing with us the data from Spliddit, and to Haris Aziz for bringing to our attention the proof of \EFExistence{} for binary valuations in \citep{AGM+15fair}. RV thanks Rupert Freeman, Nick Gravin, and Neeldhara Misra for very helpful discussions and several useful suggestions for improving the presentation of the paper. Thanks also to Erel Segal-Halevi for many helpful comments on \Cref{subsec:Proof_EF_Existence_NPcomplete_BinaryValuations}. LX acknowledges NSF \#1453542 and \#1716333, and HH acknowledges NSF \#1850076 for support.

\bibliographystyle{plainnat}
\bibliography{References}

\begin{thebibliography}{28}
\providecommand{\natexlab}[1]{#1}
\providecommand{\url}[1]{\texttt{#1}}
\expandafter\ifx\csname urlstyle\endcsname\relax
  \providecommand{\doi}[1]{doi: #1}\else
  \providecommand{\doi}{doi: \begingroup \urlstyle{rm}\Url}\fi

\bibitem[Abebe et~al.(2017)Abebe, Kleinberg, and Parkes]{AKP17fair}
Rediet Abebe, Jon Kleinberg, and David~C Parkes.
\newblock {Fair Division via Social Comparison}.
\newblock In \emph{Proceedings of the 16th Conference on Autonomous Agents and
  Multiagent Systems}, pages 281--289, 2017.

\bibitem[Aziz et~al.(2015)Aziz, Gaspers, Mackenzie, and Walsh]{AGM+15fair}
Haris Aziz, Serge Gaspers, Simon Mackenzie, and Toby Walsh.
\newblock {Fair Assignment of Indivisible Objects under Ordinal Preferences}.
\newblock \emph{Artificial Intelligence}, 227:\penalty0 71--92, 2015.

\bibitem[Aziz et~al.(2018)Aziz, Bouveret, Caragiannis, Giagkousi, and
  Lang]{ABC+18knowledge}
Haris Aziz, Sylvain Bouveret, Ioannis Caragiannis, Ira Giagkousi, and
  J{\'e}r{\^o}me Lang.
\newblock {Knowledge, Fairness, and Social Constraints}.
\newblock In \emph{Thirty-Second AAAI Conference on Artificial Intelligence},
  pages 4638--4645, 2018.

\bibitem[Barman et~al.(2018)Barman, Krishnamurthy, and Vaish]{BKV18Finding}
Siddharth Barman, Sanath~Kumar Krishnamurthy, and Rohit Vaish.
\newblock {Finding Fair and Efficient Allocations}.
\newblock In \emph{Proceedings of the 2018 ACM Conference on Economics and
  Computation}, pages 557--574, 2018.

\bibitem[Bei et~al.(2017)Bei, Qiao, and Zhang]{BQZ17networked}
Xiaohui Bei, Youming Qiao, and Shengyu Zhang.
\newblock {Networked Fairness in Cake Cutting}.
\newblock In \emph{Proceedings of the 26th International Joint Conference on
  Artificial Intelligence}, pages 3632--3638, 2017.

\bibitem[Beynier et~al.(2018)Beynier, Chevaleyre, Gourv{\`e}s, Lesca, Maudet,
  and Wilczynski]{BCG+18local}
Aur{\'e}lie Beynier, Yann Chevaleyre, Laurent Gourv{\`e}s, Julien Lesca,
  Nicolas Maudet, and Ana{\"e}lle Wilczynski.
\newblock {Local Envy-Freeness in House Allocation Problems}.
\newblock In \emph{Proceedings of the 17th International Conference on
  Autonomous Agents and Multiagent Systems}, pages 292--300, 2018.

\bibitem[Bredereck et~al.(2018)Bredereck, Kaczmarczyk, and
  Niedermeier]{BKN18envy}
Robert Bredereck, Andrzej Kaczmarczyk, and Rolf Niedermeier.
\newblock {Envy-Free Allocations Respecting Social Networks}.
\newblock In \emph{Proceedings of the 17th International Conference on
  Autonomous Agents and Multiagent Systems}, pages 283--291, 2018.

\bibitem[Budish(2011)]{B11combinatorial}
Eric Budish.
\newblock {The Combinatorial Assignment Problem: Approximate Competitive
  Equilibrium from Equal Incomes}.
\newblock \emph{Journal of Political Economy}, 119\penalty0 (6):\penalty0
  1061--1103, 2011.

\bibitem[Caragiannis et~al.(2016)Caragiannis, Kurokawa, Moulin, Procaccia,
  Shah, and Wang]{CKM+16unreasonable}
Ioannis Caragiannis, David Kurokawa, Herv{\'e} Moulin, Ariel~D Procaccia,
  Nisarg Shah, and Junxing Wang.
\newblock {The Unreasonable Fairness of Maximum Nash Welfare}.
\newblock In \emph{Proceedings of the 2016 ACM Conference on Economics and
  Computation}, pages 305--322, 2016.

\bibitem[Caragiannis et~al.(2019)Caragiannis, Gravin, and Huang]{CGH19envy}
Ioannis Caragiannis, Nick Gravin, and Xin Huang.
\newblock {Envy-Freeness Up to Any Item with High Nash Welfare: The Virtue of
  Donating Items}.
\newblock In \emph{Proceedings of the 2019 ACM Conference on Economics and
  Computation}, pages 527--545. ACM, 2019.

\bibitem[Chan et~al.(2019)Chan, Chen, Li, and Wu]{CCL+19maximin}
Hau Chan, Jing Chen, Bo~Li, and Xiaowei Wu.
\newblock {Maximin-Aware Allocations of Indivisible Goods}.
\newblock In \emph{Proceedings of the Twenty-Eighth International Joint
  Conference on Artificial Intelligence}, pages 137--143, 2019.

\bibitem[Chaudhury et~al.(2020)Chaudhury, Kavitha, Mehlhorn, and
  Sgouritsa]{CKM+20little}
Bhaskar~Ray Chaudhury, Telikepalli Kavitha, Kurt Mehlhorn, and Alkmini
  Sgouritsa.
\newblock {A Little Charity Guarantees Almost Envy-Freeness}.
\newblock In \emph{Proceedings of the Fourteenth Annual ACM-SIAM Symposium on
  Discrete Algorithms}, pages 2658--2672. SIAM, 2020.

\bibitem[Chen and Shah(2017)]{CS17ignorance}
Yiling Chen and Nisarg Shah.
\newblock {Ignorance is Often Bliss: Envy with Incomplete Information}.
\newblock Technical report, 2017.

\bibitem[Chevaleyre et~al.(2017)Chevaleyre, Endriss, and
  Maudet]{CEM17distributed}
Yann Chevaleyre, Ulle Endriss, and Nicolas Maudet.
\newblock {Distributed Fair Allocation of Indivisible Goods}.
\newblock \emph{Artificial Intelligence}, 242:\penalty0 1--22, 2017.

\bibitem[Conitzer et~al.(2017)Conitzer, Freeman, and Shah]{CFS17fair}
Vincent Conitzer, Rupert Freeman, and Nisarg Shah.
\newblock {Fair Public Decision Making}.
\newblock In \emph{Proceedings of the 2017 ACM Conference on Economics and
  Computation}, pages 629--646. ACM, 2017.

\bibitem[Conitzer et~al.(2019)Conitzer, Freeman, Shah, and
  Vaughan]{CFS+19group}
Vincent Conitzer, Rupert Freeman, Nisarg Shah, and Jennifer~Wortman Vaughan.
\newblock {Group Fairness for Indivisible Good Allocation}.
\newblock In \emph{Thirty-Third AAAI Conference on Artificial Intelligence},
  pages 1853--1860, 2019.

\bibitem[Dickerson et~al.(2014)Dickerson, Goldman, Karp, Procaccia, and
  Sandholm]{DGK+14computational}
John~P Dickerson, Jonathan Goldman, Jeremy Karp, Ariel~D Procaccia, and Tuomas
  Sandholm.
\newblock {The Computational Rise and Fall of Fairness}.
\newblock In \emph{Proceedings of the Twenty-Eighth AAAI Conference on
  Artificial Intelligence}, pages 1405--1411, 2014.

\bibitem[Dinur and Steurer(2014)]{DS14analytical}
Irit Dinur and David Steurer.
\newblock {Analytical Approach to Parallel Repetition}.
\newblock In \emph{Proceedings of the Forty-Sixth Annual ACM Symposium on
  Theory of Computing}, pages 624--633, 2014.

\bibitem[Foley(1967)]{F67resource}
Duncan Foley.
\newblock {Resource Allocation and the Public Sector}.
\newblock \emph{Yale Economic Essays}, pages 45--98, 1967.

\bibitem[Garey and Johnson(1979)]{GJ79computers}
Michael~R. Garey and David~S. Johnson.
\newblock \emph{{Computers and Intractability: A Guide to the Theory of
  {NP}-Completeness}}.
\newblock W. H. Freeman \& Co., 1979.

\bibitem[Goldman and Procaccia(2014)]{GP15spliddit}
Jonathan Goldman and Ariel~D Procaccia.
\newblock {Spliddit: Unleashing Fair Division Algorithms}.
\newblock \emph{ACM SIGecom Exchanges}, 13\penalty0 (2):\penalty0 41--46, 2014.

\bibitem[Halpern and Shah(2019)]{HS19fair}
Daniel Halpern and Nisarg Shah.
\newblock {Fair Division with Subsidy}.
\newblock In \emph{International Symposium on Algorithmic Game Theory}, pages
  374--389. Springer, 2019.

\bibitem[Krause and Golovin(2014)]{KG14submodular}
Andreas Krause and Daniel Golovin.
\newblock {Submodular Function Maximization}.
\newblock In \emph{Tractability: Practical Approaches to Hard Problems}, pages
  71--104. Cambridge University Press, 2014.

\bibitem[Lipton et~al.(2004)Lipton, Markakis, Mossel, and
  Saberi]{LMM+04approximately}
Richard~J Lipton, Evangelos Markakis, Elchanan Mossel, and Amin Saberi.
\newblock {On Approximately Fair Allocations of Indivisible Goods}.
\newblock In \emph{Proceedings of the 5th ACM Conference on Electronic
  Commerce}, pages 125--131, 2004.

\bibitem[Manurangsi and Suksompong(2019)]{MS19when}
Pasin Manurangsi and Warut Suksompong.
\newblock {When Do Envy-Free Allocations Exist?}
\newblock In \emph{Proceedings of the AAAI Conference on Artificial
  Intelligence}, volume~33, pages 2109--2116, 2019.

\bibitem[Nemhauser et~al.(1978)Nemhauser, Wolsey, and Fisher]{NWF78analysis}
George~L Nemhauser, Laurence~A Wolsey, and Marshall~L Fisher.
\newblock {An Analysis of Approximations for Maximizing Submodular Set
  Functions--I}.
\newblock \emph{Mathematical Programming}, 14\penalty0 (1):\penalty0 265--294,
  1978.

\bibitem[Nguyen and Rothe(2014)]{NR14minimizing}
Trung~Thanh Nguyen and J{\"o}rg Rothe.
\newblock {Minimizing Envy and Maximizing Average Nash Social Welfare in the
  Allocation of Indivisible Goods}.
\newblock \emph{Discrete Applied Mathematics}, 179:\penalty0 54--68, 2014.

\bibitem[Sandomirskiy and Segal-Halevi(2019)]{SE19fair}
Fedor Sandomirskiy and Erel Segal-Halevi.
\newblock {Fair Division with Minimal Sharing}.
\newblock \emph{arXiv preprint arXiv:1908.01669}, 2019.

\end{thebibliography}

\clearpage
\newpage
\section{Appendix}
 \label{sec:Appendix}

\subsection{Proof of Proposition~\ref{prop:EF_Existence_NPcomplete_BinaryValuations}}
\label{subsec:Proof_EF_Existence_NPcomplete_BinaryValuations}

Recall the statement of \Cref{prop:EF_Existence_NPcomplete_BinaryValuations}.
\EFExistenceNPcompleteBinaryVals*

Our proof of \Cref{prop:EF_Existence_NPcomplete_BinaryValuations} uses a reduction from \EQcoloring{}, which is defined below.

\begin{definition}[\EQcoloring]
Given a  graph $G$ and a number $\ell \in \N$, does there exist a proper $\ell$-coloring of $G$ such that all color classes are of equal size?
\label{defn:EquitableColoring}
\end{definition}

The standard definition of \EQcoloring{} requires the color classes to differ in size by at most one. We overload the term to refer to the version where all color classes are of the same size. \EQcoloring{} can be shown to be \NPC{} by a straightforward reduction from \GraphkColorability{}~\citep{GJ79computers}. In addition, we can assume that $\ell \geq 3$ without loss of generality.

\begin{proof} (of \Cref{prop:EF_Existence_NPcomplete_BinaryValuations})
We will show a reduction from \EQcoloring{}. Recall from \Cref{defn:EquitableColoring} that an instance of \EQcoloring{} consists of a graph $G = (V,E)$ and a number $\ell \in \N$. The goal is to determine if $G$ admits a proper $\ell$-coloring wherein the color classes are of the same size. For simplicity, we will write $n \coloneqq |V|$ and $m \coloneqq |E|$.\footnote{Not be confused with the number of agents, $n$, and the number of goods, $m$, as defined in \Cref{sec:Preliminaries}.} Note that we can assume, without loss of generality, that $G$ is connected.\footnote{Given any graph $G$, we can construct another connected graph $G' = (V',E')$ as follows: Let $V' \coloneqq V \cup \{x_1,\dots,x_\ell\} \cup \{y_1,\dots,y_{\frac{n}{\ell}+1}\}$. There is an edge between every pair of vertices in $\{x_1,\dots,x_\ell\}$ so as to induce an $\ell$-clique. In addition, each $y_i$ is connected to every vertex in $\{x_1,\dots,x_\ell\}$ as well as to every vertex in $V$. It is easy to see that $G$ admits an equitable $\ell$-coloring if and only if $G'$ admits an equitable $(\ell+1)$-coloring. Indeed, the $x_i$'s consume $\ell$ colors, and, as a result, all $y_i$'s must have the same color. This, in turn, leaves exactly $\ell$ colors for the vertices in $V$. Furthermore, there are $\frac{n}{\ell}+1$ vertices in each color class, implying that the coloring is equitable.} Since a connected graph with $n$ vertices has at least $n-1$ edges, we have that
\begin{align}
m \geq n-1.
\label{eqn:Connected_Graph_Inequality}
\end{align}

In addition, we will also assume that each vertex in $G$ has degree at least two. Indeed, for any vertex $v$ with degree at most one, we can add $\ell$ new vertices $v_1, v_2,\dots,v_{\ell}$ that are connected as follows: The vertices $v_1,\dots,v_{\ell}$ constitute an $\ell$-clique (that is, for every pair of distinct $i,j \in [\ell]$, $v_i$ is connected to $v_j$), and $v$ is connected to each vertex in $\{v_2,\dots,v_{\ell-1}\}$ but not $v_1$. Call the new graph $G'$. It is easy to see that $G$ has an equitable $\ell$-coloring if and only if $G'$ does.

We will construct a fair division instance with $m+n$ goods and $m+\ell$ agents. The agents are classified into $m$ \emph{edge} agents $a_1, \dots, a_m$ and $\ell$ \emph{dummy} agents $d_1,\dots,d_\ell$. The goods are classified into $n$ \emph{vertex} goods $v_1,\dots,v_n$ and $m$ \emph{edge} goods $e_1,\dots,e_m$. Note that we use the same notation for the vertices (edges) and the corresponding vertex (edge) goods.

The preferences of the agents are defined as follows: For every edge $e = (v_i,v_j)$, an edge agent $a_e$ approves all the edge goods and exactly two vertex goods $v_i$ and $v_j$. Each dummy agent approves all the vertex goods and has zero value for the edge goods.

($\Rightarrow$) Suppose $G$ admits an equitable coloring with each color class of size $\frac{n}{\ell}$. Then, an envy-free allocation $A$ can be constructed as follows: Assign each edge good $e$ to the edge agent $a_e$ and each vertex good $v$ to the dummy agent $d_i$ if vertex $v$ has color $i$. Notice that all goods are allocated under $A$. Also note that no two edge agents envy each other since each of them gets exactly one edge good. Furthermore, due to the proper coloring condition, for any edge $e = (v_i,v_j)$ in $G$, the corresponding vertex goods $v_i$ and $v_j$ are assigned to distinct dummy agents. Hence, no edge agent envies a dummy agent. The dummy agents have zero value for the edge goods and therefore do not envy the edge agents. Finally, since all color classes are of the same size, each dummy agent gets exactly $\frac{n}{\ell}$ approved goods, and therefore does not envy any other dummy agent. Overall, the allocation is envy-free.

($\Leftarrow$) Now suppose there exists an envy-free allocation $A$. We will show that $A$ satisfies \Cref{property:Edge_agent_cannot_get_two_edge_goods,property:Dummy_agent_gets_at_least_one_vertex_good,property:Dummy_agents_cannot_get_any_edge_good,property:Edge_agent_gets_exactly_one_edge_good,property:Edge_agents_cannot_get_any_vertex_good,property:Dummy_agents_cannot_get_adjacent_vertices} that will help us infer an equitable coloring of $G$.

\begin{property}
No edge agent can get two or more edge goods under $A$.
\label{property:Edge_agent_cannot_get_two_edge_goods}
\end{property}
\begin{proof} (of
\Cref{property:Edge_agent_cannot_get_two_edge_goods})
Suppose, for contradiction, that some edge agent $a_e$ gets two or more edge goods. Then, any other edge agent $a_{e'}$ has a utility of at least $2$ for the bundle of $a_e$. For $A$ to be envy-free, $a_{e'}$ must have a utility of at least $2$ for its own bundle. For binary valuations, this means that $a_{e'}$ must be assigned two or more goods that it approves. Therefore, we need at least $2m$ goods to satisfy the edge agents. The total number of available goods is $m + n$, which, using \Cref{eqn:Connected_Graph_Inequality}, evaluates to at most $2m+1$. This leaves at most one good to be allocated among $\ell$ dummy agents. Since $\ell \geq 3$, some dummy agent is bound to be envious, contradicting the envy-freeness of $A$. 
\end{proof}

\begin{property}
Every dummy agent gets at least one vertex good under $A$.
\label{property:Dummy_agent_gets_at_least_one_vertex_good}
\end{property}
\begin{proof} (of
\Cref{property:Dummy_agent_gets_at_least_one_vertex_good})
Fix a vertex good $v$ and a dummy agent $d$. Then, either $v$ is assigned to $d$, or $d$ gets some other (approved) good to prevent it from envying the owner of $v$. Since the only goods approved by the dummy agents are the vertex goods, the claim follows.
\end{proof}

\begin{property}
No dummy agent can get an edge good under $A$.\label{property:Dummy_agents_cannot_get_any_edge_good}
\end{property}
\begin{proof} (of
\Cref{property:Dummy_agents_cannot_get_any_edge_good})
Suppose, for contradiction, that a dummy agent $d$ gets an edge good $e$ under $A$. From \Cref{property:Dummy_agent_gets_at_least_one_vertex_good}, we know that $d$ also gets some vertex good, say $v_0$. By assumption, the graph $G$ has minimum degree two, so there must exist some edge $e_1 = (v_0,v_1)$ adjacent to the vertex $v_0$. Notice that the edge agent $a_{e_1}$ has a utility of (at least) $2$ for the bundle of $d$. Therefore, for $A$ to be envy-free, $a_{e_1}$ must get at least two goods that it approves. \Cref{property:Edge_agent_cannot_get_two_edge_goods} limits the number of edge goods assigned to any edge agent to at most one. Therefore, in addition to some edge good, $a_{e_1}$ must also get the vertex good $v_1$. Once again using the bound on minimum degree of $G$, we get that there must exist some edge $e_2 = (v_1,v_2)$ adjacent to the vertex $v_1$. A similar argument shows that the vertex good $v_2$ must be assigned to the edge agent $a_{e_2}$. Continuing in this manner, we will eventually encounter an edge $e_i = (v_{i-1},v_i)$ such that $v_{i-1}$ is already assigned to $a_{e_{i-1}}$ and $v_{i}$ is already assigned to either $d$ or some other edge agent. This would imply that $a_{e_i}$ is envious of some other agent under $A$---a contradiction.
\end{proof}

\begin{property}
Every edge agent gets exactly one edge good under $A$.
\label{property:Edge_agent_gets_exactly_one_edge_good}
\end{property}
\begin{proof} (of
\Cref{property:Edge_agent_gets_exactly_one_edge_good})
Follows from \Cref{property:Dummy_agents_cannot_get_any_edge_good,property:Edge_agent_cannot_get_two_edge_goods}.
\end{proof}

\begin{property}
No edge agent can get a vertex good under $A$.
\label{property:Edge_agents_cannot_get_any_vertex_good}
\end{property}
\begin{proof} (of
\Cref{property:Edge_agents_cannot_get_any_vertex_good})
Suppose, for contradiction, that some edge agent $a_{e_0}$ is assigned a vertex good $v_0$. 
%
%
Let $e_1 = (v_0,v_1)$ be an edge incident to the vertex $v_0$ in $G$ (such an edge must exist due to the bound on minimum degree). From \Cref{property:Edge_agent_gets_exactly_one_edge_good}, we know that each edge agent gets exactly one edge good. Thus, the edge agent $a_{e_1}$ has a utility of (at least) $2$ for the bundle of the agent $a_{e_0}$. For $A$ to be envy-free, $a_{e_1}$ must receive two or more goods that it approves, only one of which can be an edge good. Therefore, agent $a_{e_1}$ must also receive the vertex good $v_1$. Now let $e_2 = (v_1,v_2)$ be another edge incident to the vertex $v_1$ in $G$ (again, such an edge exists because $v_1$ has degree at least two). A similar argument implies that the vertex good $v_2$ must be assigned to the edge agent $a_{e_2}$. Continuing in this manner, let $e_i = (v_{i-1},v_i)$ denote the first edge in the sequence for which one of the following mutually exclusive conditions is true:

\begin{enumerate}
	\item Either, the vertex good $v_i$ is assigned to an agent different from $a_{e_i}$, or
	
	\item the vertex good $v_i$ is assigned to $a_{e_i}$ and $v_i = v_0$ (thus $a_{e_i}=a_{e_0}$).
\end{enumerate}

Notice that due to the finiteness of the graph $G$, there must exist an edge $e_i$ satisfying one of the aforementioned conditions. We will now argue that each of these conditions leads to a contradiction.

\begin{enumerate}
	\item First, suppose that the vertex good $v_i$ is assigned to an agent different from $a_{e_i}$. Then, the edge agent $a_{e_i}$ has a utility of (at least) $2$ for the bundle of edge agent $a_{e_{i-1}}$ and a utility of $1$ for its own bundle, contradicting the envy-freeness of $A$.
	
	\item Next, suppose that $v_i = v_0$. That is, there exists a cycle $C = \{(v_0,v_1),(v_1,v_2),\dots,(v_{i-1},v_0)\}$ in the graph $G$ such that for every $j \in \{0,1,\dots,i-1\}$, the vertex good $v_j$ is assigned to the edge agent $a_{e_j}$.
	
Recall from \Cref{property:Dummy_agent_gets_at_least_one_vertex_good} that each dummy agent gets at least one vertex good. Since all vertex goods corresponding to the vertices in $C$ are assigned to the edge agents, there must exist at least one vertex outside the cycle $C$. Furthermore, since $G$ is connected, there must exist a path from this vertex to a vertex in $C$, say $v_1$. Thus, there must exist a vertex $v'_1 \in V$ such that $(v_1,v'_1) \in E$ and $v'_1 \notin C$. Let $e'_1 \coloneqq (v_1,v'_1)$. Then, the edge agent $a_{e'_1}$ has a utility of $2$ for the bundle of agent $a_{e_1}$ (recall that $a_{e_1}$ receives an edge good and the vertex good $v_1$), and must therefore be assigned the vertex good $v'_1$. Since the vertex $v'_1$ has degree at least two, there must exist another edge $e'_2 = (v'_1,v'_2)$ in $G$. By a similar argument as before, the edge agent $a_{e'_2}$ must be assigned the vertex good $v'_2$. Continuing in this manner, we will encounter an edge, say $e'_i = (v'_{i-1},v'_i)$ such that the vertex good $v'_i$ is assigned to an agent different from $a_{e'_i}$. This means that the edge agent $a_{e'_i}$ has a utility of (at least) $2$ for the bundle of the edge agent $a_{e'_{i-1}}$ and a utility of $1$ for its own bundle, contradicting the envy-freeness of $A$.
\end{enumerate}

This completes the proof of \Cref{property:Edge_agents_cannot_get_any_vertex_good}.
\end{proof}

\begin{property}
For any edge $e = (v_i,v_j)$, no dummy agent is assigned both vertex goods $v_i$ and $v_j$ under $A$.
\label{property:Dummy_agents_cannot_get_adjacent_vertices}
\end{property}
\begin{proof} (of
\Cref{property:Dummy_agents_cannot_get_adjacent_vertices})
Suppose, for contradiction, that for some edge $e = (v_i,v_j)$, a dummy agent $d$ is assigned both $v_i$ and $v_j$. \Cref{property:Edge_agent_gets_exactly_one_edge_good} implies that the utility of $a_e$ for its own bundle is exactly $1$. However, the utility of $a_e$ for the bundle of $d$ is $2$, contradicting the envy-freeness of $A$.
\end{proof}

It follows from \Cref{property:Edge_agents_cannot_get_any_vertex_good} that all vertex goods must be allocated among the dummy agents. Now consider the following coloring of the graph $G$: For each vertex $v$, the color of $v$ is the index of the dummy agent that gets the vertex good $v$. \Cref{property:Dummy_agents_cannot_get_adjacent_vertices} implies that the coloring is proper. Furthermore, due to envy-freeness of $A$, agents with identical valuations must have equal utilities. Therefore, each dummy agent gets the same number of vertex goods, implying that the coloring is equitable. This completes the proof of \Cref{prop:EF_Existence_NPcomplete_BinaryValuations}.
\end{proof}

\subsection{Proof of Theorem~\ref{thm:HEFk_Verification_ApproxAlgo}}
\label{subsec:Proof_HEFk_Verification_ApproxAlgo}

Recall the statement of \Cref{thm:HEFk_Verification_ApproxAlgo}.

\HEFkVerificationApproxAlgo*

Recall from \Cref{sec:Theoretical_Results} that given any allocation $A$, the \emph{residual envy function} $f : 2^{[m]} \rightarrow \mathbb{R}$ is defined as follows:
\begin{align*}
	f(S) \coloneqq \sum_{h \in [n]} \sum_{i \neq h} \max\{0, v_i(A_h \setminus S) - v_i(A_i)\}.
\end{align*}

Here, $f(S)$ is the aggregate envy in $A$ after hiding the goods in $S \subseteq [m]$. We will show in \Cref{lem:Supermodularity} that $f$ is \emph{supermodular}, i.e., for any pair of sets $S,T \subseteq [m]$ such that $S \subseteq T$ and any good $j \notin T$, $f(S) - f(S \cup \{j\}) \geq f(T) - f(T \cup \{j\})$. The proof of \Cref{thm:HEFk_Verification_ApproxAlgo} will then follow from the standard greedy algorithm for submodular maximization, or, equivalently, supermodular minimization~ \citep{NWF78analysis}.

\begin{restatable}{lemma}{Supermodularity}
 \label{lem:Supermodularity}
 The residual envy function $f$ is supermodular.
\end{restatable}
\begin{proof}
We will start with the necessary notation. For any agent $h \in [n]$ and any other agent $i \in [n] \setminus \{h\}$, define $f_{h,i}(S) \coloneqq \max\{0, v_i(A_h \setminus S) - v_i(A_i)\}$ as the envy of $i$ towards $h$ after hiding the goods in $S$. Also, let $f_h(S) \coloneqq \sum_{i \neq h} f_{h,i}(S)$ denote the total (aggregate) envy towards $h$. We therefore have $f(S) = \sum_{h \in [n]} f_h(S) = \sum_{h \in [n]} \sum_{i \neq h} f_{h,i}(S)$. 

Notice that $f_{h,i}$ is a monotone non-increasing set function, i.e., for any $S \subseteq T$, we have $f_{h,i}(S) \geq f_{h,i}(T)$. Also notice that for any $T \subseteq [m]$ and any $j \in [m] \setminus T$, we have that $f_{h,i}(T) - f_{h,i}(T \cup \{j\}) \leq v_{i,j}$.

For any set of goods $S \subseteq [m]$ and any agent $h \in [n]$, define $E_h(S) \coloneqq \{i \in [n] : f_{h,i}(S) > 0\}$ as the set of agents that envy agent $h$ even after the goods in $S$ are hidden. Notice that if $S \subseteq T$, then $E_h(T) \subseteq E_h(S)$. Thus, if for some agent $i$ we have that $i \notin E_h(S)$, then $i \notin E_h(T)$, and therefore $f_{h,i}(S) = f_{h,i}(T) = 0$.

Define $N_j \coloneqq \{i \in [n] : v_{i,j} > 0\}$ as the set of agents that have a strictly positive valuation for the good $j$.

We will now prove that $f$ is supermodular, i.e., for any $S \subseteq T$ and any good $j \notin T$, $f(S) - f(S \cup \{j\}) \geq f(T) - f(T \cup \{j\})$. Let $r \in [n]$ be the owner of good $j$ under $A$, i.e., $j \in A_r$. Notice that if $i \notin N_j$ (i.e., $v_{i,j}=0$), then additivity of valuations implies $v_i(A_r \setminus S) = v_i(A_r \setminus S \cup \{j\})$. Thus,
\begin{align}
    f(S) - f(S \cup \{j\}) & = f_r(S) - f_r(S \cup \{j\}) \nonumber \\
    & = \sum_{i \neq r} f_{r,i}(S) - f_{r,i}(S \cup \{j\}) \nonumber \\
     & = \sum_{i \in E_r(S)} f_{r,i}(S) - f_{r,i}(S \cup \{j\}) \nonumber \\
     & = \sum_{i \in E_r(S) \cap N_j} f_{r,i}(S) - f_{r,i}(S \cup \{j\}),\label{eqn:supermodular_temp1}
\end{align}
where the first equality uses the fact that for any $h \neq r$, we have $f_h(S) = f_h(S \cup \{j\})$, the third equality uses the fact that if $i \notin E_r(S)$, then $f_{r,i}(S) = f_{r,i}(S \cup \{j\}) = 0$, and the fourth equality uses the fact that $v_i(A_r \setminus S) = v_i(A_r \setminus S \cup \{j\})$ whenever $i \notin N_j$. By a similar reasoning for the set $T$, we get that
\begin{align}
f(T) &- f(T \cup \{j\}) = \sum_{i \in E_r(T) \cap N_j} f_{r,i}(T) - f_{r,i}(T \cup \{j\}).
\label{eqn:supermodular_temp2}
\end{align}

Recall that $E_r(T) \subseteq E_r(S)$. Therefore, \Cref{eqn:supermodular_temp1} can be rewritten as
\begin{align}
    f(S) - f(S \cup \{j\}) & = \sum_{i \in E_r(T) \cap N_j} f_{r,i}(S) - f_{r,i}(S \cup \{j\}) + \nonumber \\
     & \qquad \sum_{i \in E_r(S) \setminus E_r(T) \cap N_j} f_{r,i}(S) - f_{r,i}(S \cup \{j\}) \nonumber \\
     & \geq \sum_{i \in E_r(T) \cap N_j} f_{r,i}(S) - f_{r,i}(S \cup \{j\}),
\label{eqn:supermodular_temp3}
\end{align}
where the inequality follows from the use of the monotonicity of $f_{r,i}$ for all $i \in E_r(S) \setminus E_r(T) \cap N_j$.

Therefore, from \Cref{eqn:supermodular_temp2,eqn:supermodular_temp3}, it suffices to show that for every $i \in E_r(T) \cap N_j$, $f_{r,i}(S) - f_{r,i}(S \cup \{j\}) \geq f_{r,i}(T) - f_{r,i}(T \cup \{j\})$. We will prove this by contradiction.

Suppose, for contradiction, that for some $i \in E_r(T) \cap N_j$, we have $f_{r,i}(S) - f_{r,i}(S \cup \{j\}) < f_{r,i}(T) - f_{r,i}(T \cup \{j\})$. Then, we must have $i \in E_r(S \cup \{j\})$, since otherwise we get $i \notin E_r(T \cup \{j\})$ and therefore $f_{r,i}(S \cup \{j\}) = f_{r,i}(T \cup \{j\}) = 0$. This would imply that $f_{r,i}(S) < f_{r,i}(T)$, which contradicts the monotonicity of $f_{r,i}$. Hence, for any $i \in E_r(T) \cap N_j$, we also have that $i \in E_r(S \cup \{j\})$.

Notice that for any $i \in E_r(S \cup \{j\}) \cap N_j$, we have $f_{r,i}(S) - f_{r,i}(S \cup \{j\}) = v_{i,j}$ by the additivity of valuations. However, this would require that $f_{r,i}(T) - f_{r,i}(T \cup \{j\}) > v_{i,j}$, which is a contradiction. Therefore, the function $f$ must be supermodular.
\end{proof}

We are now ready to prove \Cref{thm:HEFk_Verification_ApproxAlgo}.

\begin{proof} (of \Cref{thm:HEFk_Verification_ApproxAlgo})
Note that allocation $A$ is \HEF{} with respect to a set $S$ if and only if $f(S) \leq 0$. For integral valuations, $f(S) \leq 0$ if and only if $f(S) < 1$. Therefore, it suffices to compute a set $S$ in polynomial time such that $|S| \leq \kopt \cdot \ln E + 1$ and $f(S) < 1$.

Consider the greedy algorithm described in Algorithm~\ref{alg:Greedy_HEFk_ApproxAlgo}. 
\begin{algorithm}[t]
 \DontPrintSemicolon
 \KwIn{An instance $\langle [n], [m], \V \rangle$ and an allocation $A$.}
 \KwOut{A set $S \subseteq [m]$.}
 \BlankLine
 Initialize $S = \emptyset$.\;
 \While{$f(S) \geq 1$}{
 	Set $j' \leftarrow \arg\max_{j \in [m] \setminus S} f(S) - f(S \cup \{j\})$\;
 	 \Comment*[r]{tiebreak lexicographically}
 	Update $S \leftarrow S \cup \{j'\}$
 }
 \KwRet $S$
 \caption{Greedy Approximation Algorithm for \HEFkVerification{}}
\label{alg:Greedy_HEFk_ApproxAlgo}
\end{algorithm}
At each step, the algorithm adds to the current set the good that provides the largest reduction in the residual envy. This process is continued as long as $f(S) \geq 1$. Since there are $m$ goods, it is clear that the algorithm terminates in at most $m$ steps. Furthermore, from the above observation, it follows that the allocation $A$ is \HEF{} with respect to the set $S$ returned by the algorithm. Therefore, all that remains to be shown is a bound on $|S|$.

Observe that $f(\emptyset) = E$. Recall from the proof of \Cref{lem:Supermodularity} that $f$ is a sum of monotone non-increasing set functions, and is therefore itself monotone non-increasing. Define another set function $g: 2^{[m]} \rightarrow \mathbb{R}$ as follows:
\begin{align*}
	g(S) \coloneqq E - f(S).
\end{align*}

Notice that $g$ is a non-negative, monotone non-decreasing, and integer-valued submodular function with $g(\emptyset) = 0$. Therefore, our goal is to find a set $S$ such that $g(S) > E - 1$.

We will now use the result of \citet{NWF78analysis} for submodular maximization stated below as \Cref{prop:Submodular_Greedy_Approx}. In particular, let $p \coloneqq \kopt$ be the size of the optimal hidden set (i.e., the number of goods that must be hidden under $A$). Then, 
$$\max\limits_{S : |S|\leq p} g(S) = E - \min\limits_{S : |S|\leq \kopt} f(S) = E.$$

From the bound in \Cref{prop:Submodular_Greedy_Approx}, we have that
\begin{alignat*}{2}
    & \qquad & (1 - e^{-q/p}) \max\limits_{S : |S|\leq p} g(S) & > E - 1 \\
    & \Longleftrightarrow & (1 - e^{-q/\kopt}) \cdot E & > E - 1 \\
    & \Longleftrightarrow & 1 - e^{-q/\kopt} & > 1 - 1/E \\
    & \Longleftrightarrow & \ln \frac{1}{E} & > -q/\kopt \\
    & \Longleftrightarrow & q & > \kopt \ln E.
\end{alignat*}

Thus, after $q > \kopt \ln E$ steps, any set $S$ constructed by the algorithm satisfies $g(S) > E - 1$, or, equivalently, $f(S) < 1$, giving us the desired bound $|S| \leq \kopt \ln E + 1$. This completes the proof of \Cref{thm:HEFk_Verification_ApproxAlgo}.
\end{proof}

\begin{restatable}[\citet{NWF78analysis}, \citet{KG14submodular}]{prop}{SubmodularGreedyApprox}
 \label{prop:Submodular_Greedy_Approx}
 Let $g: 2^{[m]} \rightarrow \mathbb{R}_{\geq 0}$ be a monotone non-decreasing submodular function, and let $\{S_i\}_{i \geq 0}$ be the sequence of sets constructed in Algorithm~\ref{alg:Greedy_HEFk_ApproxAlgo}. Then, for any positive integers $p$ and $q$, we have that
 \begin{align*}
     g(S_q) \geq (1 - e^{-q/p}) \max\limits_{S : |S|\leq p} g(S).
 \end{align*}
\end{restatable}

\subsection{\MNW{} can have large regret in the worst-case}
\label{subsec:MNW_Large_Regret}

This section presents two results concerning the worst-case regret of \MNW{} solution. In \Cref{prop:NSW_vs_HEFk}, we will provide a family of instances for which the normalized regret of \MNW{} approaches $1$ (i.e., the maximum possible value). In \Cref{prop:NSW_vs_HEFk_Binary_Vals}, we will show a slightly weaker limit ($\nicefrac{1}{2}$ instead of $1$) that holds even for the restricted domain of binary valuations. We will use $\kappaopt(\I) \coloneqq \min_{A} \kappa(A,\I)$ to denote the smallest number of goods that must be hidden under any allocation in the instance $\I$. 

\begin{restatable}{prop}{NSWRegret}
 \label{prop:NSW_vs_HEFk}
There exists a family of instances for which the normalized regret of any Nash optimal allocation approaches $1$ in the limit.
\end{restatable}
\begin{proof}
Consider the fair division instance $\I$ with five agents $a_1,\dots,a_5$ and five goods $g_1,\dots,g_5$ shown in \Cref{tab:Instance_NSW_vs_HEFk}. 
\begin{table}
\centering
    \begin{tabular}{ c|cccccc }
	& $g_1$ & $g_2$ & $g_3$ & $g_4$ & $g_5$\\ \hline
  $a_1$ & $1$ & $0$ & $0$ & $0$ & $0$\\
  $a_2$ & $10$ & $1$ & $0$ & $0$ & $0$\\
  $a_3$ & $0$ & $10$ & $1$ & $0$ & $0$\\
  $a_4$ & $0$ & $0$ & $10$ & $1$ & $0$\\
  $a_5$ & $0$ & $0$ & $0$ & $10$ & $1$\\
	\end{tabular}
	\caption{The instance used in proof of \Cref{prop:NSW_vs_HEFk}.}
	\label{tab:Instance_NSW_vs_HEFk}
\end{table}
The unique Nash optimal allocation (say $A$) for this instance assigns $g_i$ to $a_i$ for every $i \in [5]$. Thus, the goods $g_1$, $g_2$, $g_3$, $g_4$ must be hidden under $A$, i.e., $\kappa(A,\I) = 4$. On the other hand, an allocation (say $B$) that assigns $g_5$ to $a_1$, and $g_{i-1}$ to $a_i$ for every $i \in \{2,\dots,5\}$ only needs to hide the good $g_1$. Indeed, $\kappaopt(\I) = 1$ since any allocation must hide $g_1$ to avoid envy from $a_1$ or $a_2$. The desired family of instances is the natural extension of the above example to $n$ agents and $n$ goods. In the limit, the normalized regret of the Nash optimal allocation is $\lim_{n \rightarrow \infty} \frac{(n-1) - 1}{n-1} = 1$.
\end{proof}

\begin{restatable}{prop}{NSWvsHEFBinary}
 \label{prop:NSW_vs_HEFk_Binary_Vals}
There exists a family of instances with binary valuations for which the normalized regret of any Nash optimal allocation approaches $\nicefrac{1}{2}$ in the limit.
\end{restatable}
\begin{proof}
Fix some $t \in \N$. Consider an instance $\I_n$ with $2t+1$ agents, consisting of $t$ groups of \emph{ordinary agents} $\{a_{i},b_{i}\}_{i \in [t]}$ and one \emph{special agent} $s$. The goods are also classified into $t$ groups, with group $i$ comprising of five goods $g_{i,1},\dots,g_{i,5}$. For each $i \in [t]$, both $a_i$ and $b_i$ approve all five goods in group $i$ and have zero value for all the other goods. The special agent $s$ approves all the goods.

The above instance admits an envy-free allocation $A$ in which $s$ gets one good from each group, and the other goods are allocated evenly among the group members. That is, for each $i \in [t]$, $a_{i}$ gets $\{g_{i,1},g_{i,2}\}$, $b_{i}$ gets $\{g_{i,3},g_{i,4}\}$, and $s$ gets $g_{i,5}$. Thus, $\kappaopt(\I_n) = 0$.

Let $B$ denote any Nash optimal allocation. It is easy to see that $B$ is of one of the following two canonical forms:
\begin{itemize}
    \item Either $s$ gets two goods from two different groups and the rest of the goods are assigned `evenly,' i.e., for each $i \in [t-2]$, $a_i$ gets $\{g_{i,1},g_{i,2},g_{i,3}\}$ and $b_i$ gets $\{g_{i,4},g_{i,5}\}$, and for $i \in \{t-1,t\}$, $a_i$ gets $\{g_{i,1},g_{i,2}\}$, $b_i$ gets $\{g_{i,3},g_{i,4}\}$ and $s$ gets $g_{i,5}$,
    \item or, $s$ gets three goods from three different groups and the other goods are assigned `evenly,' i.e., for each $i \in [t-3]$, $a_i$ gets $\{g_{i,1},g_{i,2},g_{i,3}\}$ and $b_i$ gets $\{g_{i,4},g_{i,5}\}$, and for $i \in \{t-2,t-1,t\}$, $a_i$ gets $\{g_{i,1},g_{i,2}\}$, $b_i$ gets $\{g_{i,3},g_{i,4}\}$ and $s$ gets $g_{i,5}$.
\end{itemize}

Either way, $B$ must hide at least $t-3$ goods (one good in each of the groups $1,\dots,t-3$ to avoid envy from $b_i$). Thus, $\reg(B,\I_n) = \kappa(B,\I_n) = t-3$. 

The desired family of instances can be obtained by choosing an arbitrarily large $t$. In the limit, the normalized regret of $B$ is $\lim_{t \rightarrow \infty} \frac{t-3}{2t} = \frac{1}{2}$.
\end{proof}

\subsection{Additional Experimental Results}
\label{subsec:Additional_Experiments}

\Cref{tab:Expt_BinaryVals_bias_0.7_Part2} presents additional results for the synthetic data used in \Cref{subsec:Expt_Synthetic} (i.e., binary valuations with $v_{i,j} \sim \Ber(0.7)$ i.i.d.). This time, we compare the algorithms in terms of their (a) normalized worst-case regret (over the $100$ instances), (b) the frequency with which the algorithms compute envy-free outcomes, and (c) the worst-case number of goods that must be hidden by each algorithm. The trend is similar to that in \Cref{subsec:Expt_Synthetic}, with \Market{} and \RR{} outperforming \MNW{} and \Envygraph{}. \Cref{tab:Expt_BinaryVals_bias_0.5} presents similar results for Bernoulli parameter $0.5$. Finally, \Cref{fig:Spliddit_data_distribution} illustrates the distribution of the Spliddit data. As can be seen, a large fraction of instances have between $3$ and $6$ agents and between $3$ and $15$ goods, with a sharp spike at $n=3$ and $m=6$.

\begin{table*}
\centering
 \begin{tabular}{|cccc|}
 \multicolumn{4}{c}{}\\
 \multicolumn{4}{c}{\textbf{Normalized worst-case regret}}\\
 \hline
 \footnotesize{\Market{}} & \footnotesize{\RR{}} & \footnotesize{\MNW{}} & \footnotesize{\Envygraph{}}\\
 \includegraphics[width=0.22\textwidth]{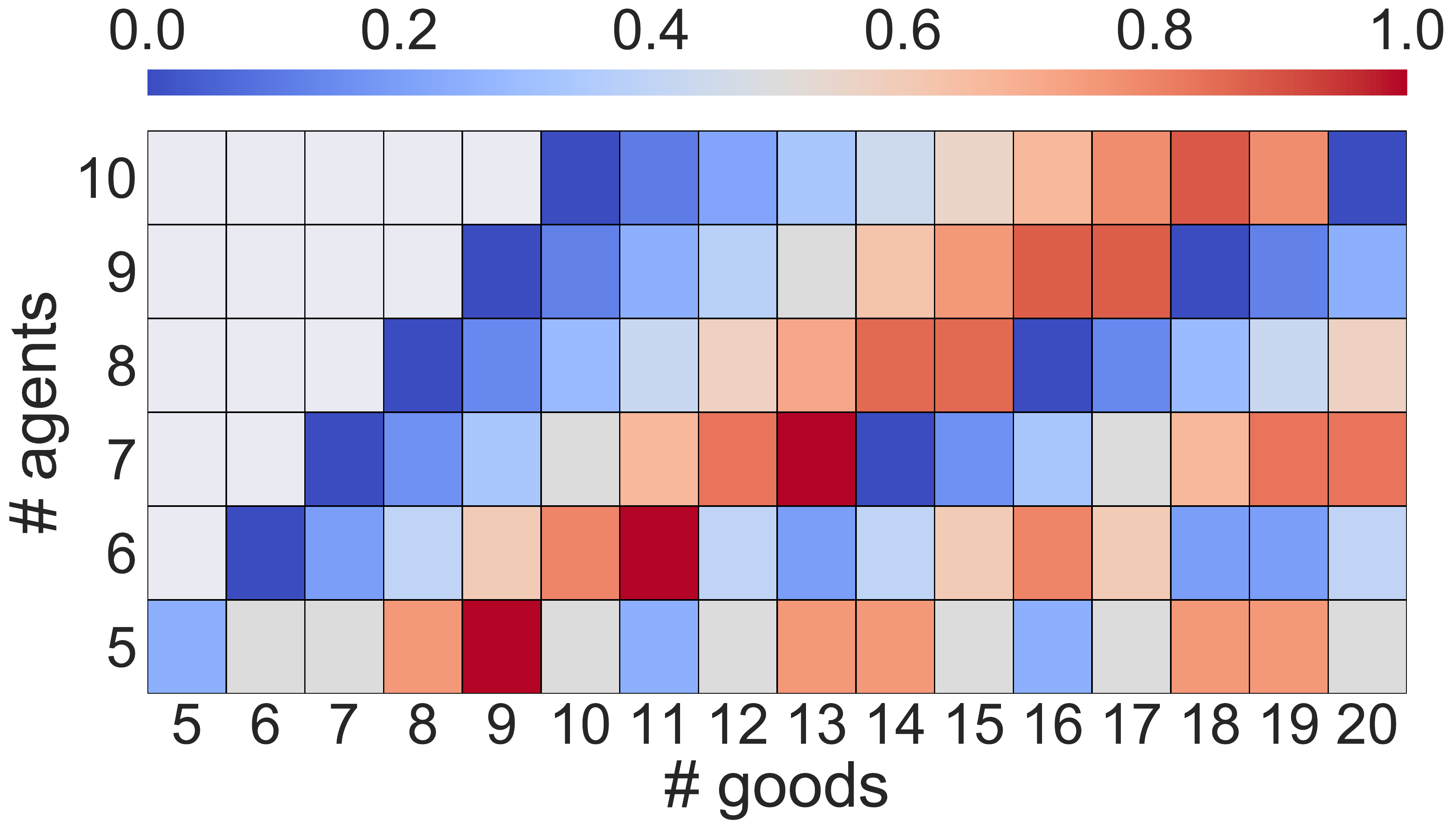} & \includegraphics[width=0.22\textwidth]{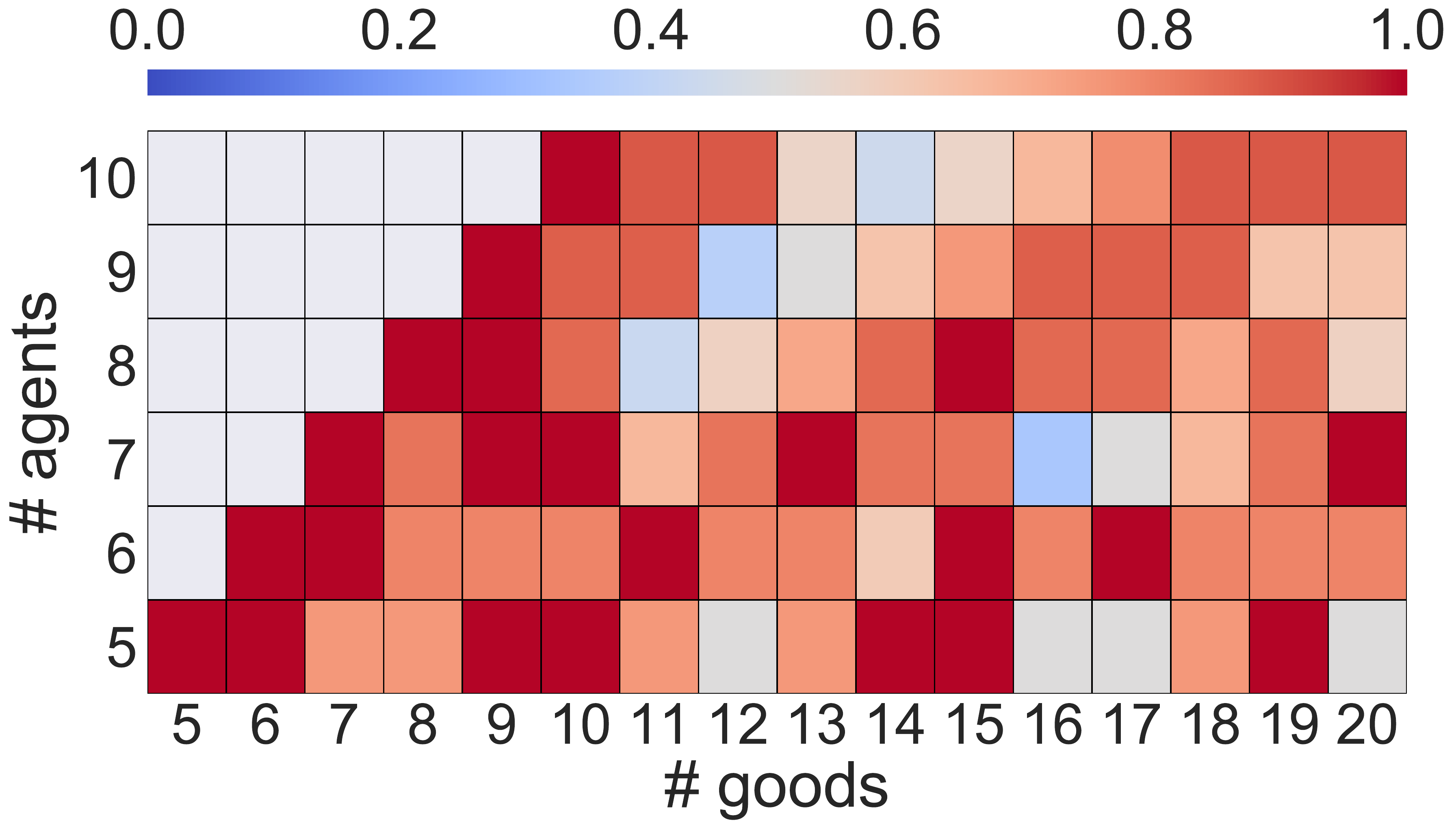} &
 \includegraphics[width=0.22\textwidth]{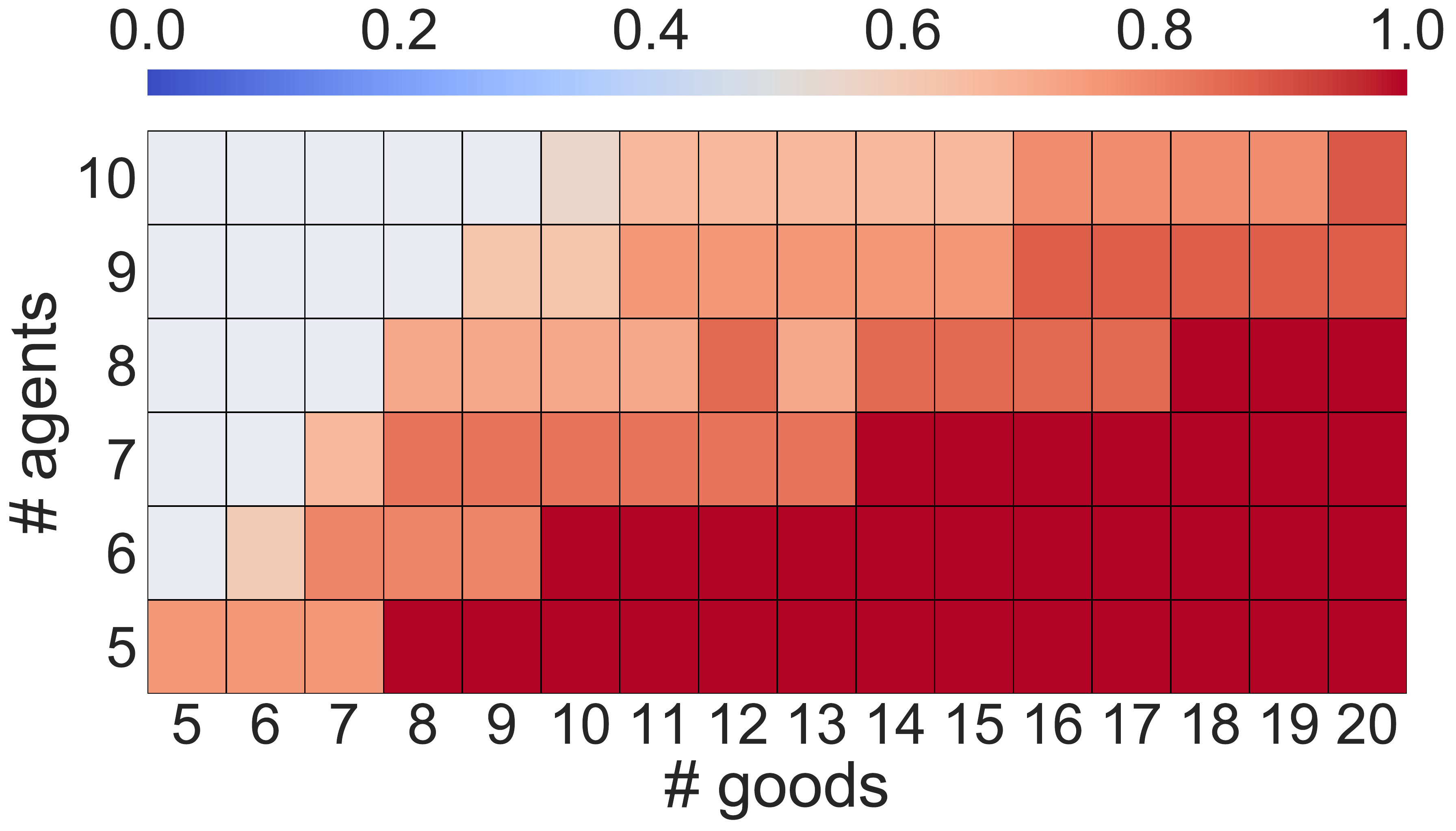} &
 \includegraphics[width=0.22\textwidth]{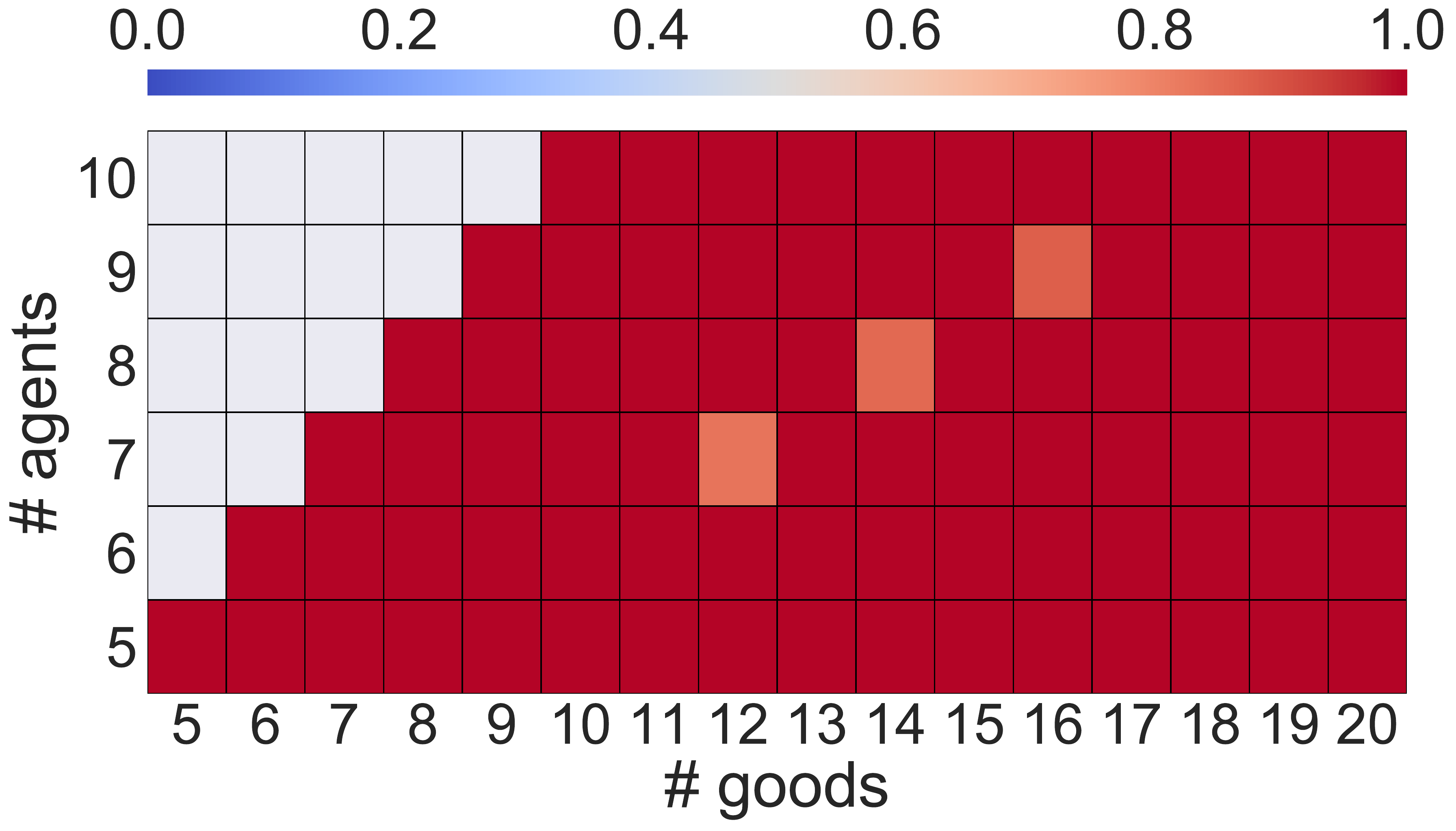} \\
 \hline
 \multicolumn{4}{c}{}\\
 \multicolumn{4}{c}{\textbf{Frequency of envy-freeness}}\\
 \hline
 \footnotesize{\Market{}} & \footnotesize{\RR{}} & \footnotesize{\MNW{}} & \footnotesize{\Envygraph{}}\\
 \includegraphics[width=0.22\textwidth]{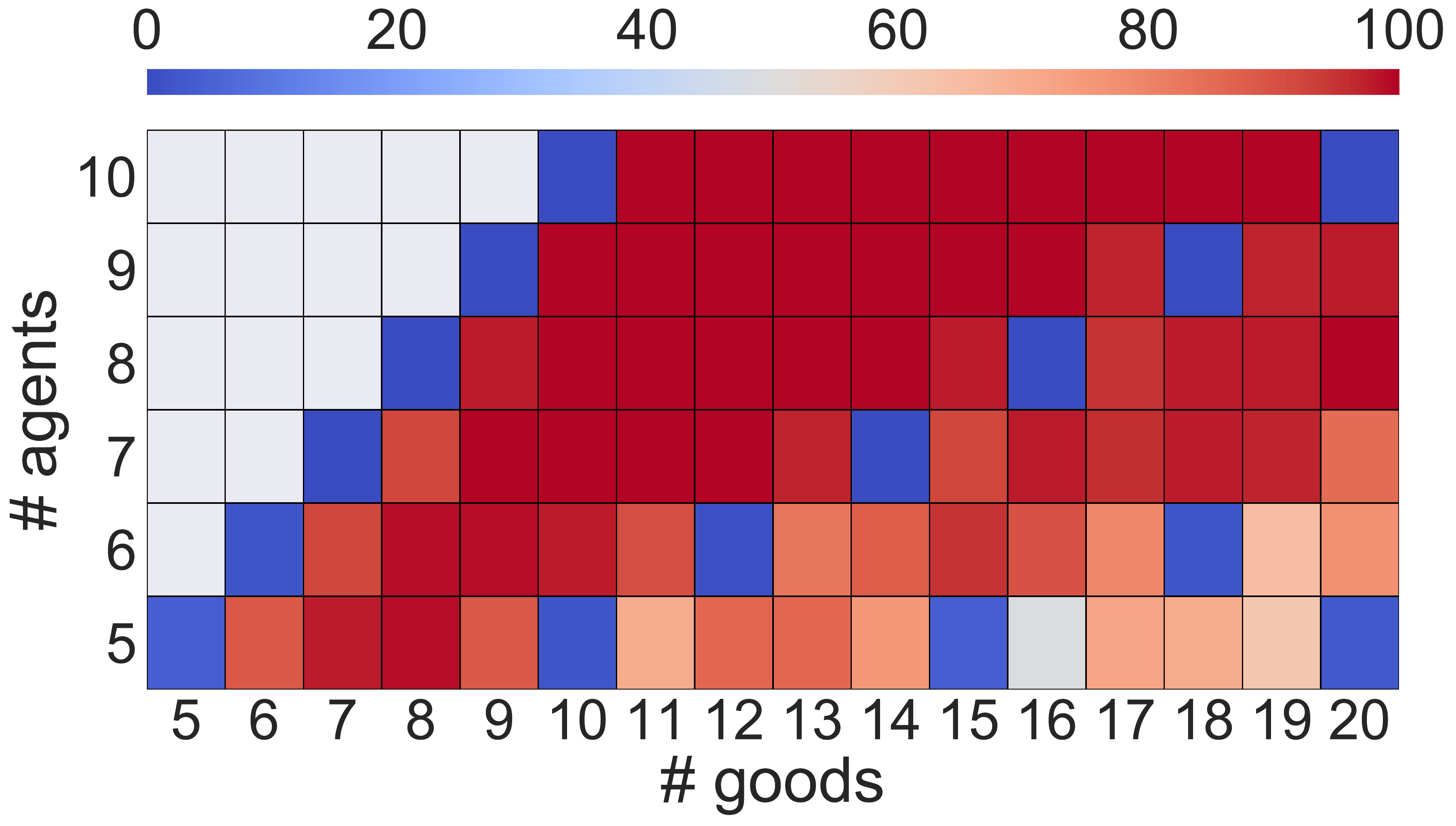} & \includegraphics[width=0.22\textwidth]{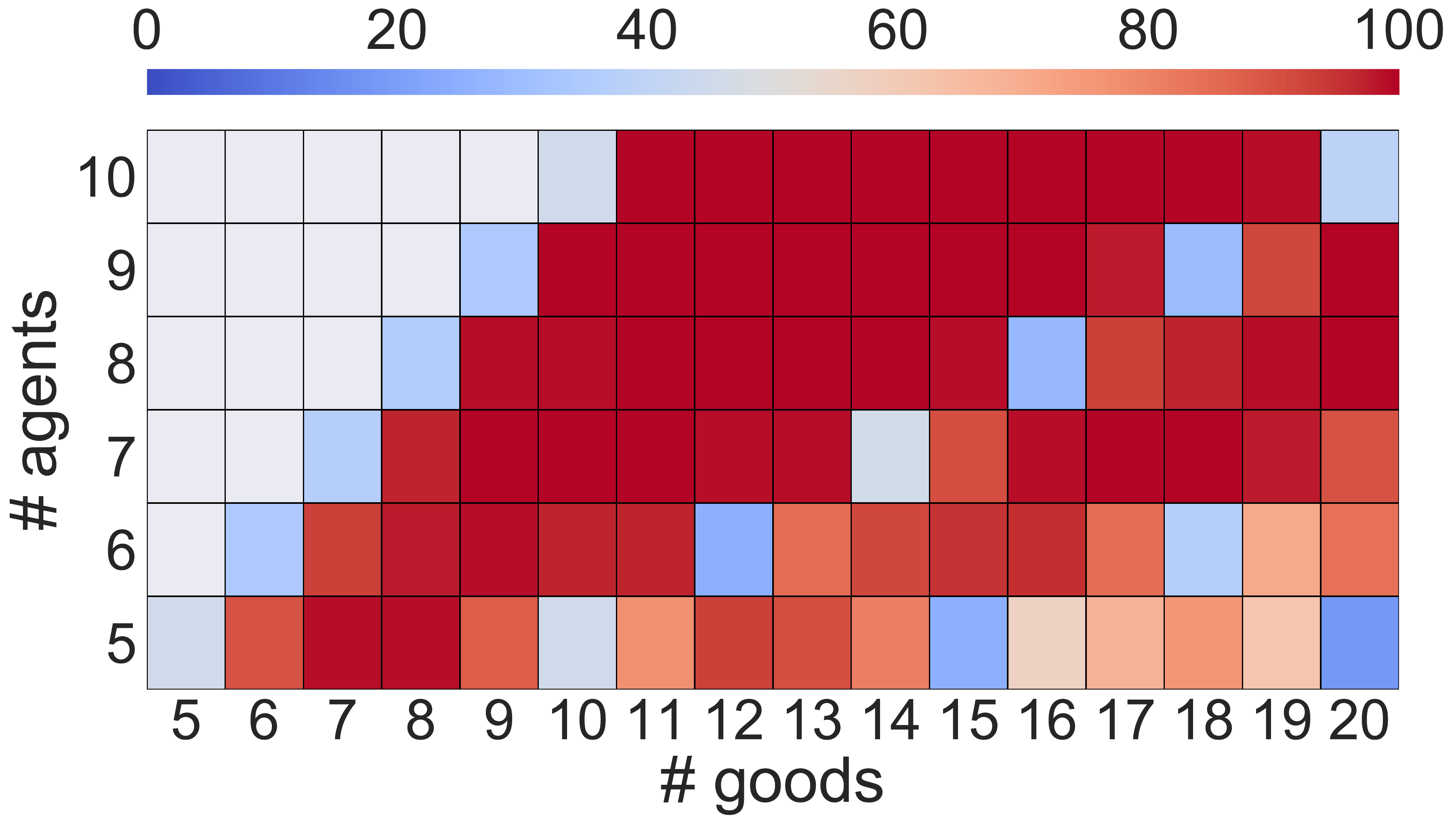} &
 \includegraphics[width=0.22\textwidth]{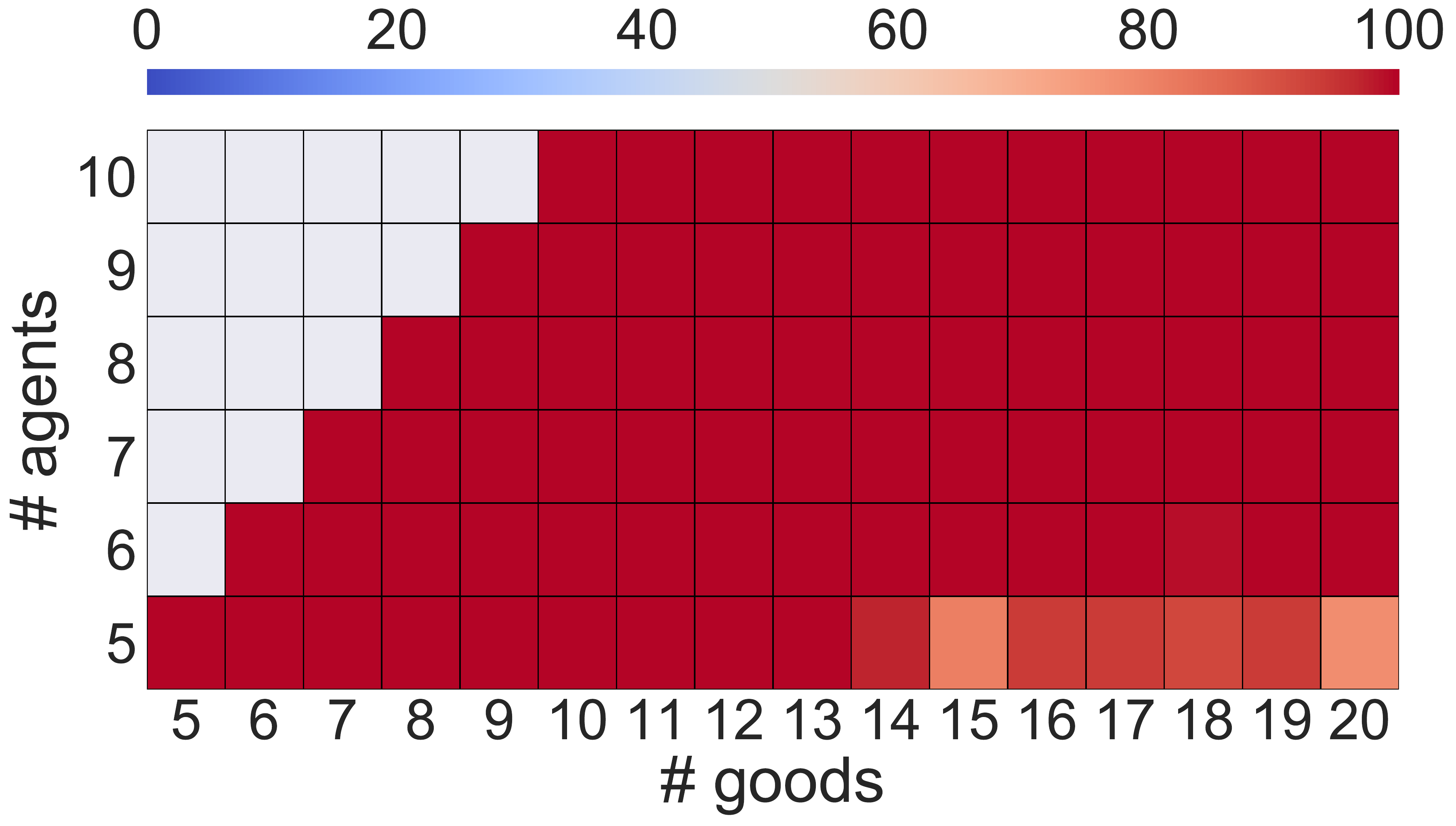} &
 \includegraphics[width=0.22\textwidth]{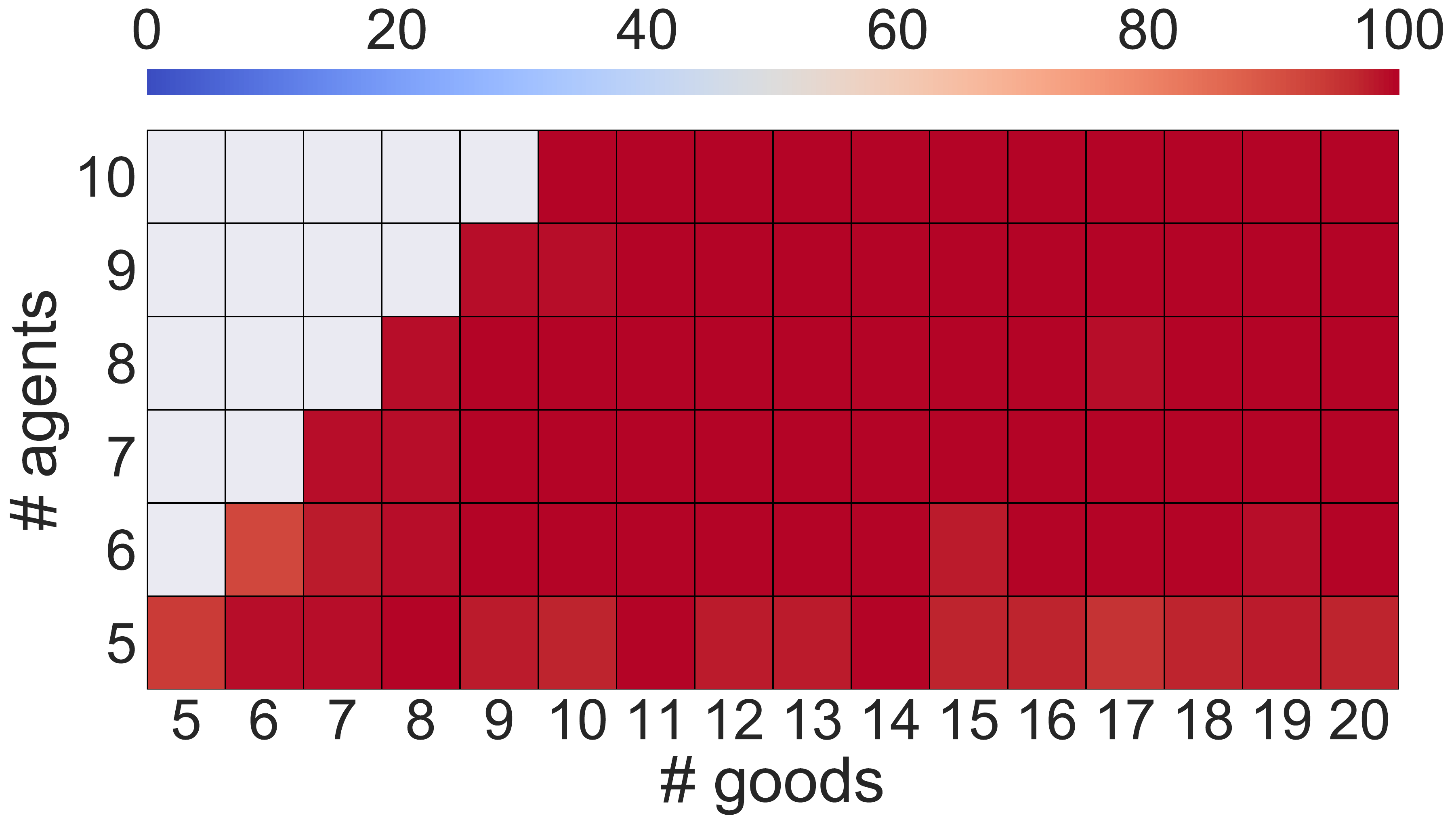} \\
 \hline
 \multicolumn{4}{c}{}\\
 \multicolumn{4}{c}{\textbf{Number of goods that must be hidden in the worst-case} (max over all $100$ instances)}\\
 \hline
 \footnotesize{\Market{}} & \footnotesize{\RR{}} & \footnotesize{\MNW{}} & \footnotesize{\Envygraph{}}\\
 \includegraphics[width=0.22\textwidth]{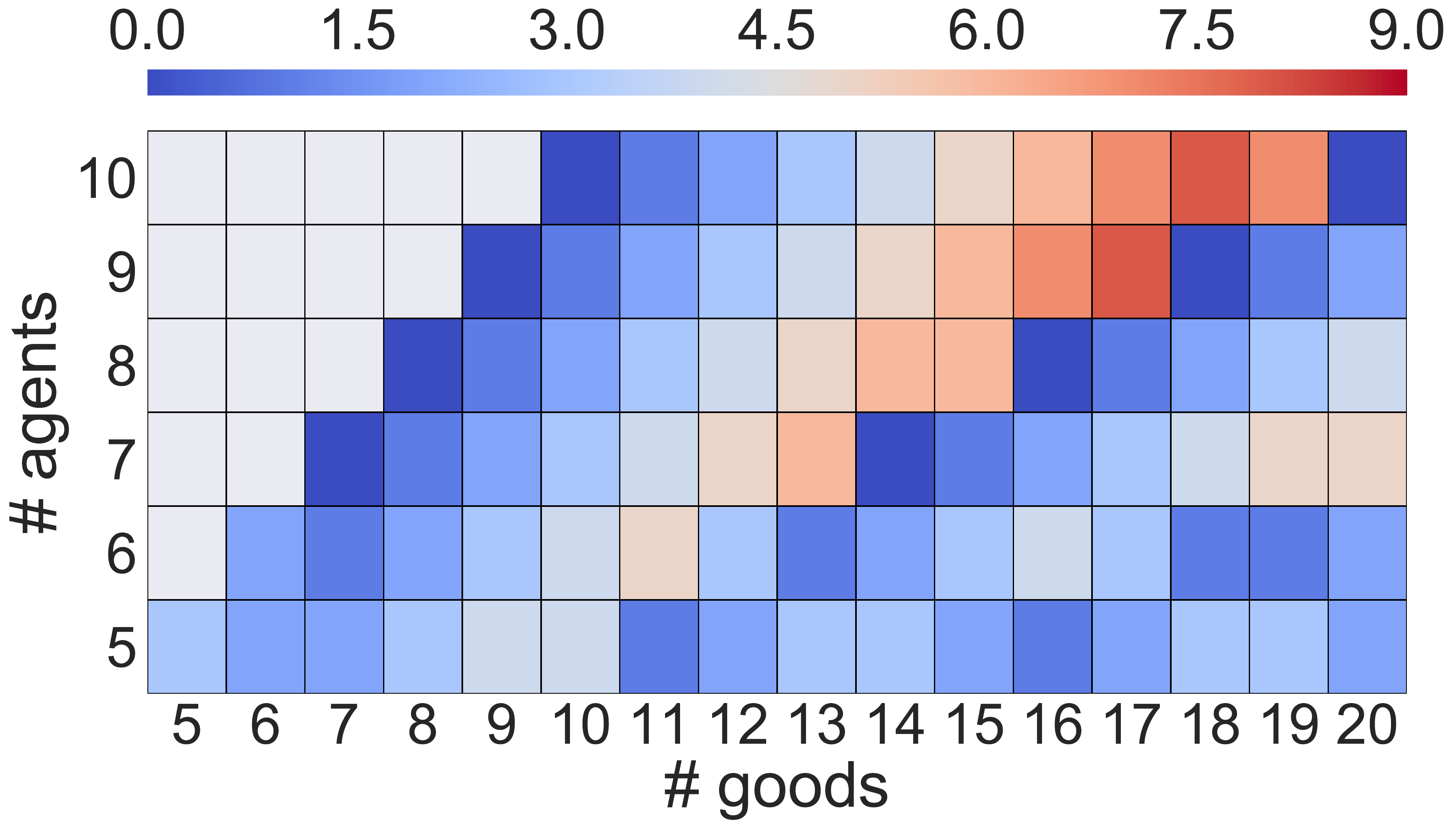} & \includegraphics[width=0.22\textwidth]{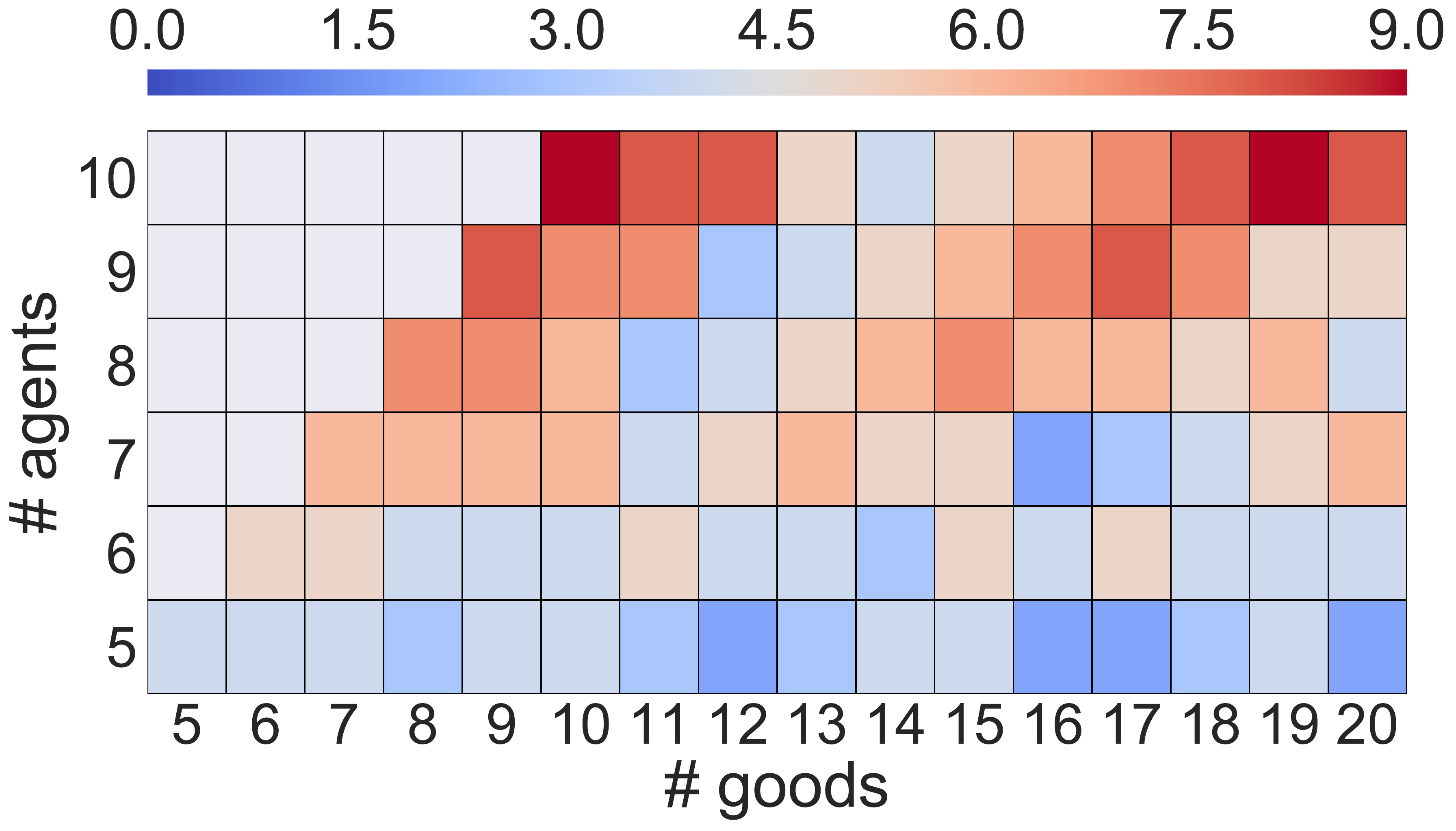} &
 \includegraphics[width=0.223\textwidth]{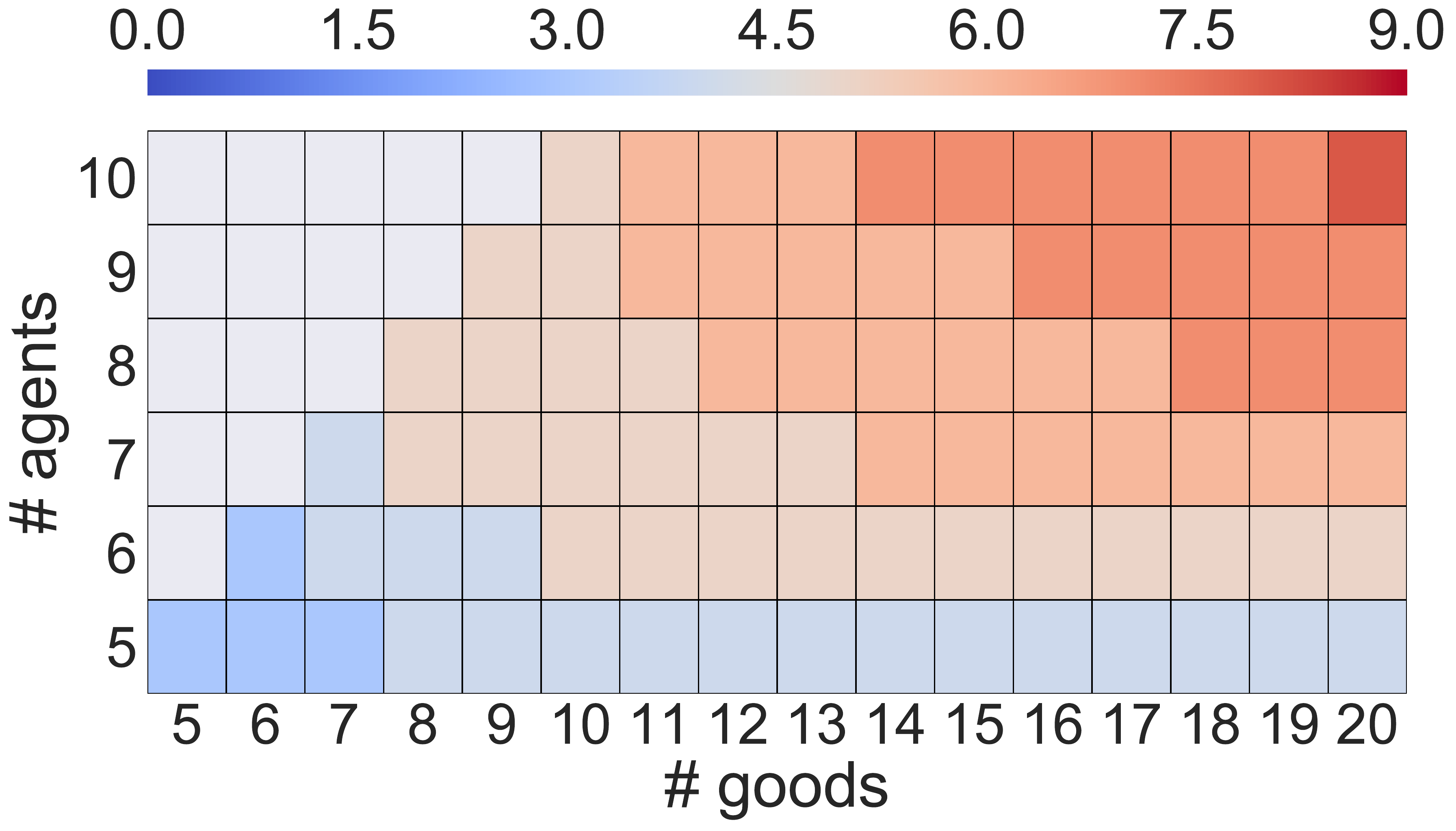} &
 \includegraphics[width=0.22\textwidth]{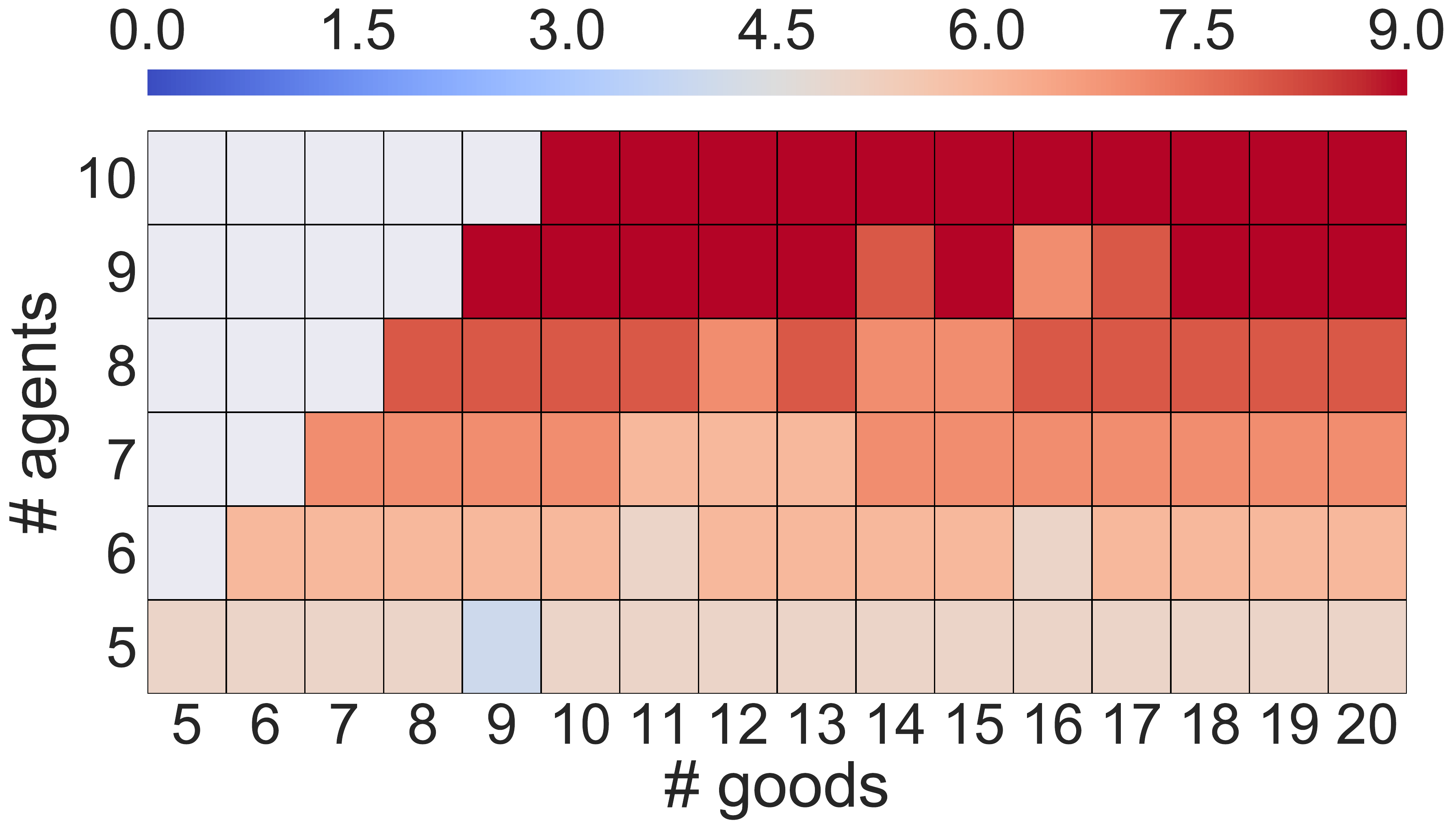} \\
 \hline
 \end{tabular}
 \caption{Comparing various \EF{1} algorithms over synthetically generated binary instances with $v_{i,j} \sim \Ber(0.7)$ i.i.d.}
 \label{tab:Expt_BinaryVals_bias_0.7_Part2}
\end{table*}

\begin{figure*}
    \centering
    \includegraphics[width=\linewidth]{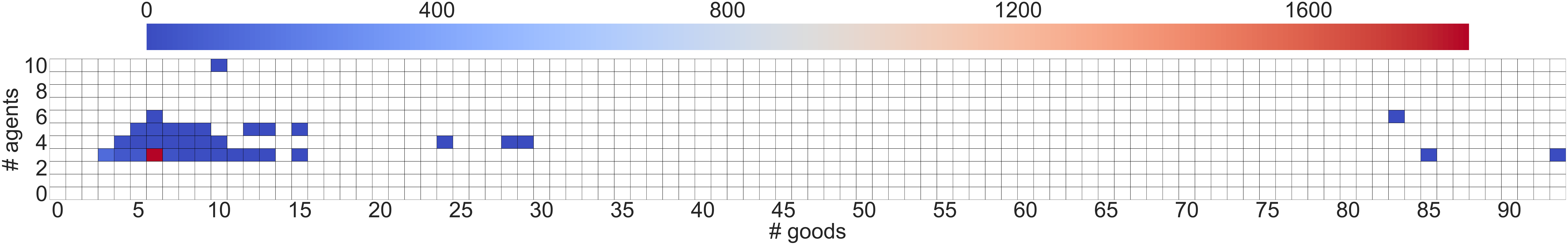}
    \caption{Distribution of the Spliddit data. The color of each cell denotes the number of instances in the dataset with the corresponding number of goods, $m$, on the X axis, and number of agents, $n$, on the Y axis.}
    \label{fig:Spliddit_data_distribution}
\end{figure*}

\begin{table*}
\centering
 \begin{tabular}{|cccc|}
 \multicolumn{4}{c}{\textbf{Normalized average-case regret}}\\
 \hline
 \footnotesize{\Market{}} & \footnotesize{\RR{}} & \footnotesize{\MNW{}} & \footnotesize{\Envygraph{}}\\
 \includegraphics[width=0.22\textwidth]{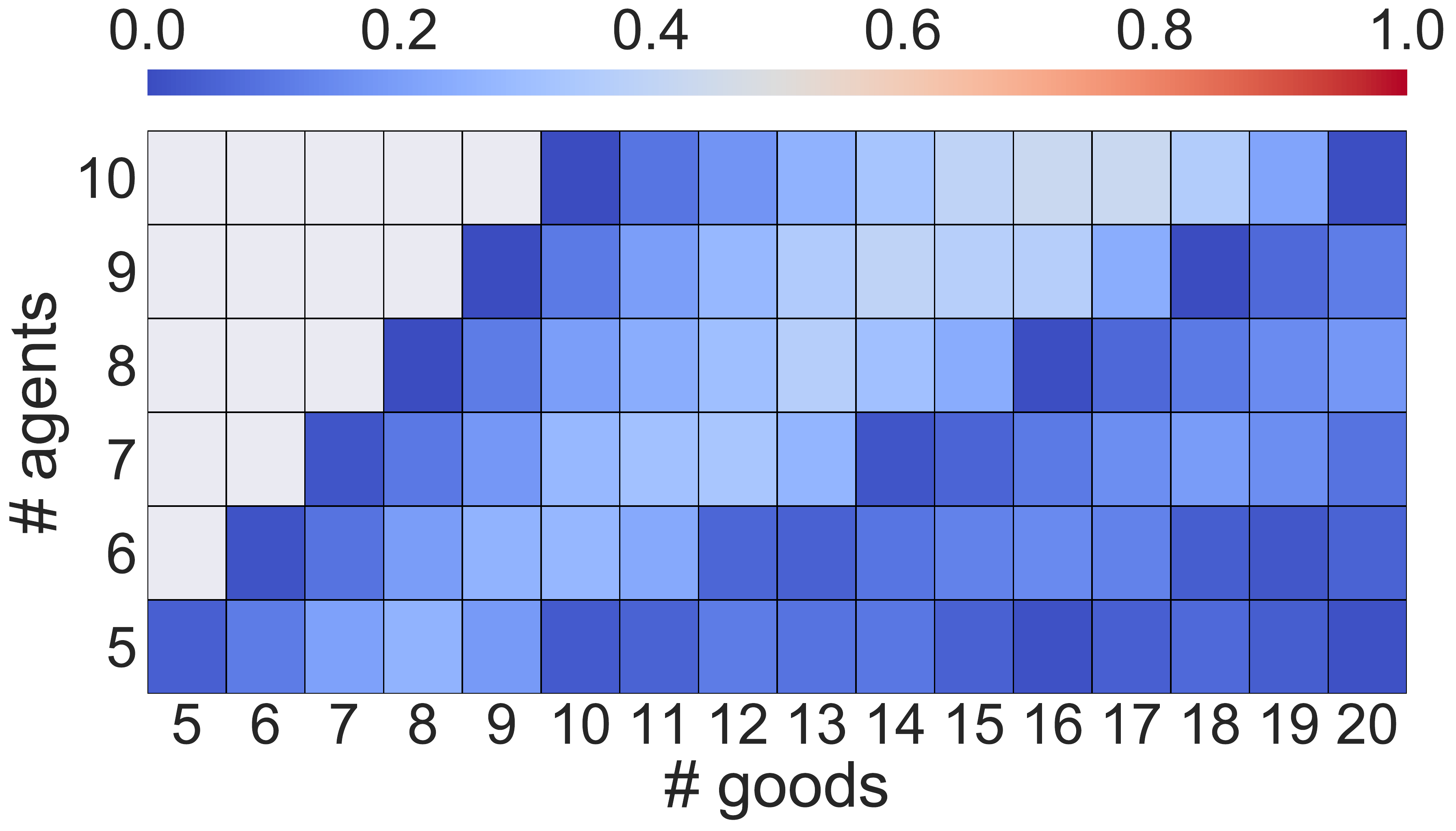} &
 \includegraphics[width=0.22\textwidth]{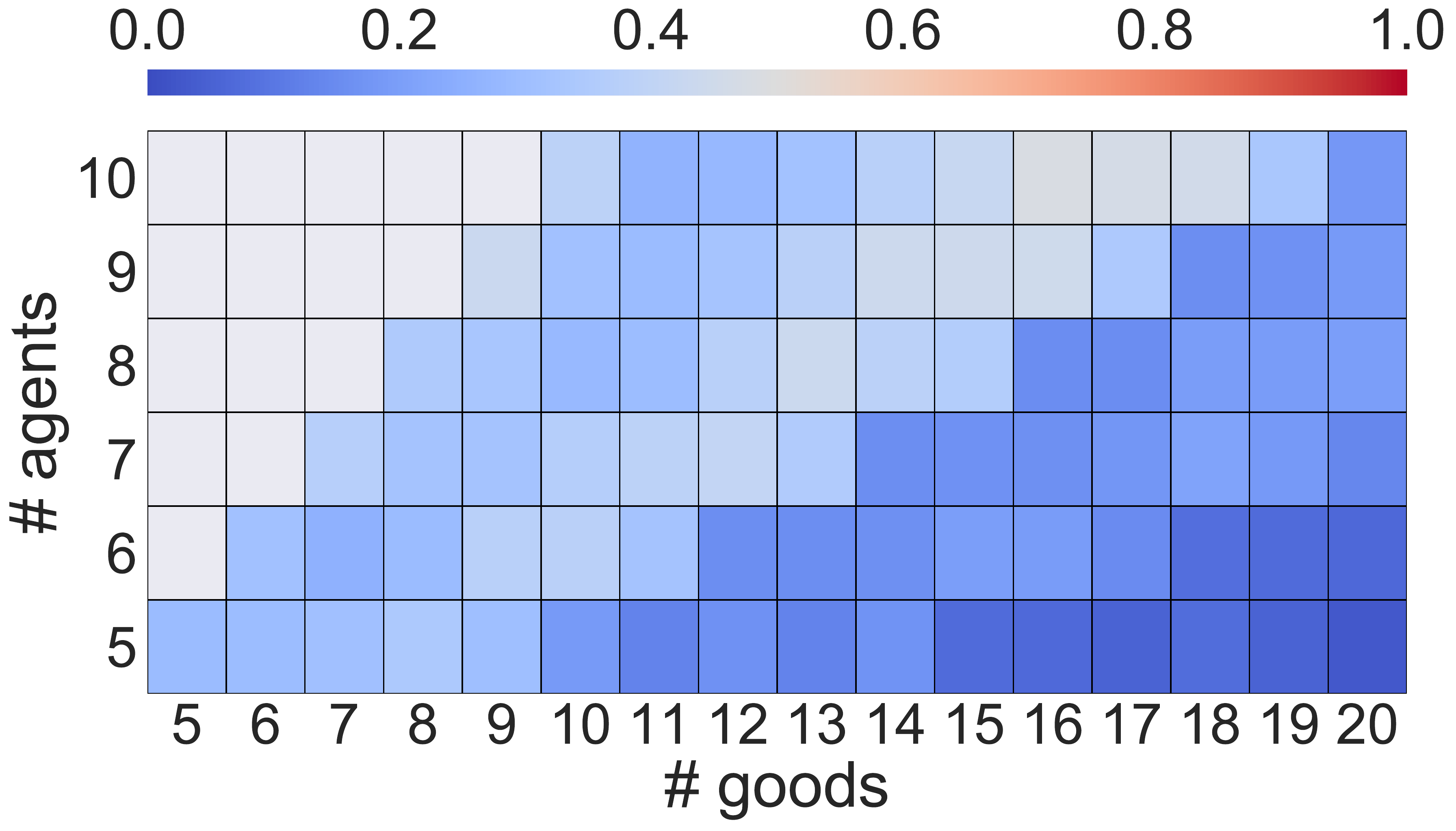} &
 \includegraphics[width=0.22\textwidth]{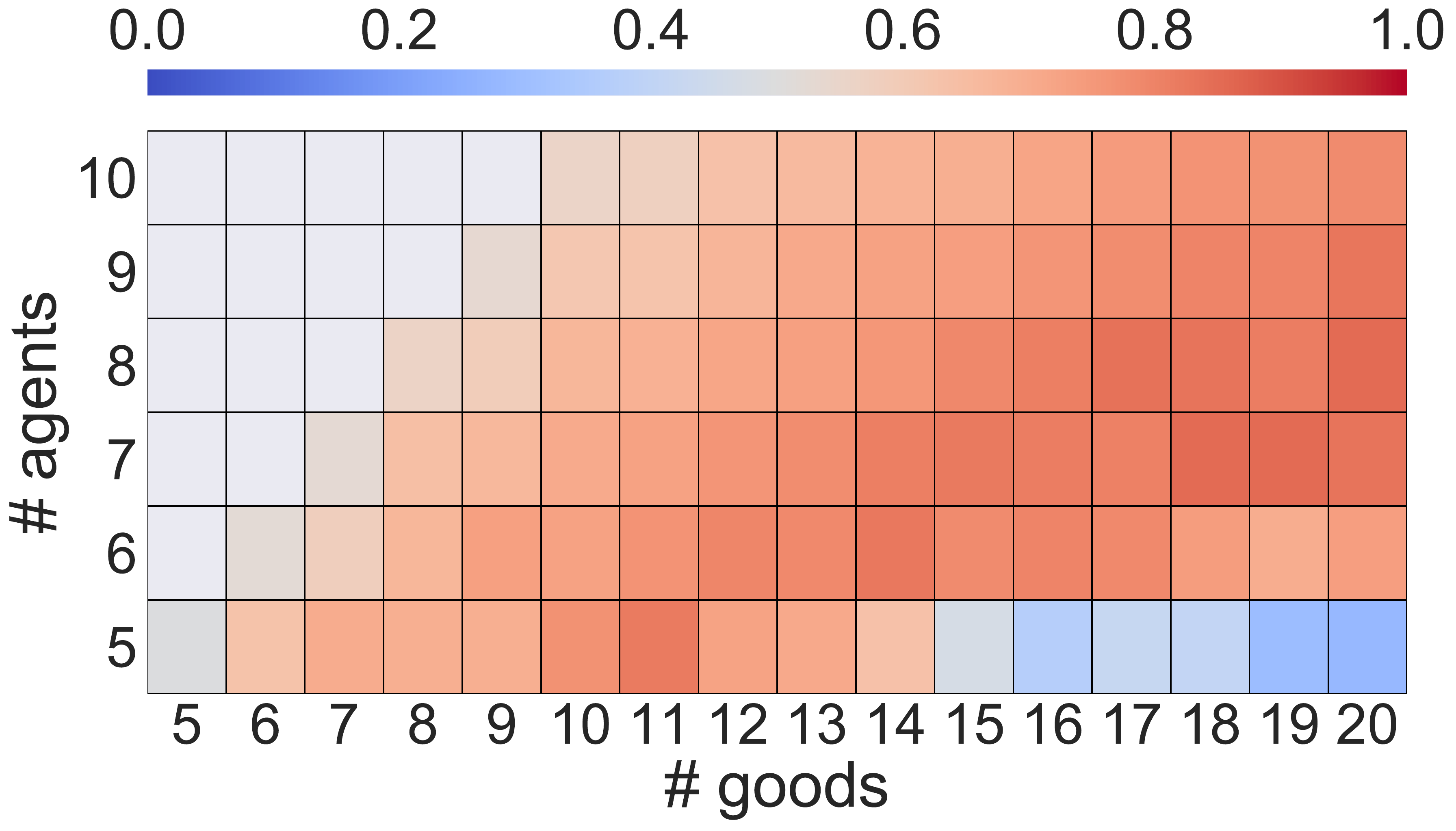} &
 \includegraphics[width=0.22\textwidth]{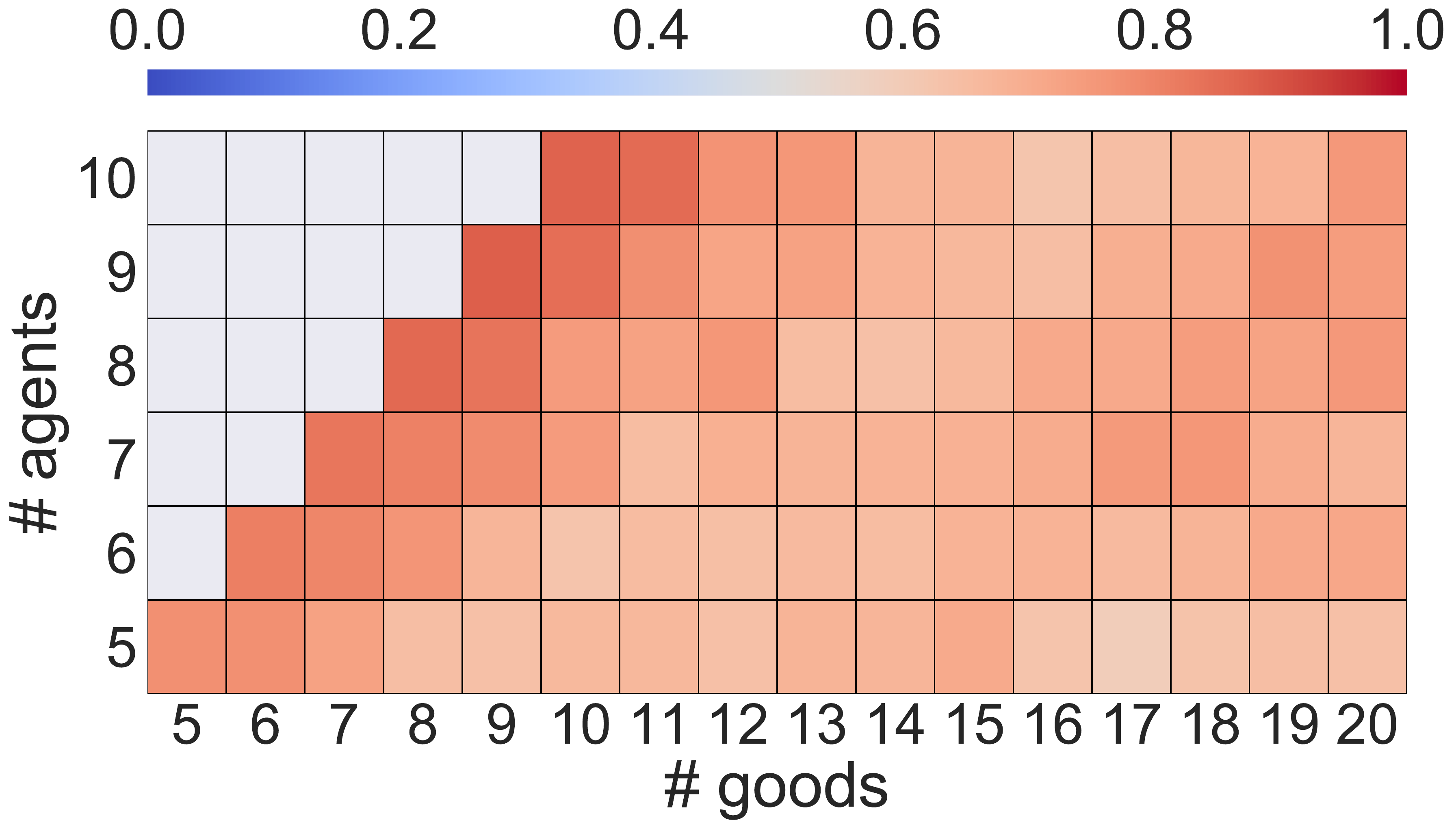} \\
 \hline
 \multicolumn{4}{c}{}\\
 \multicolumn{4}{c}{\textbf{Normalized worst-case regret}}\\
 \hline
 \footnotesize{\Market{}} & \footnotesize{\RR{}} & \footnotesize{\MNW{}} & \footnotesize{\Envygraph{}}\\
 \includegraphics[width=0.22\textwidth]{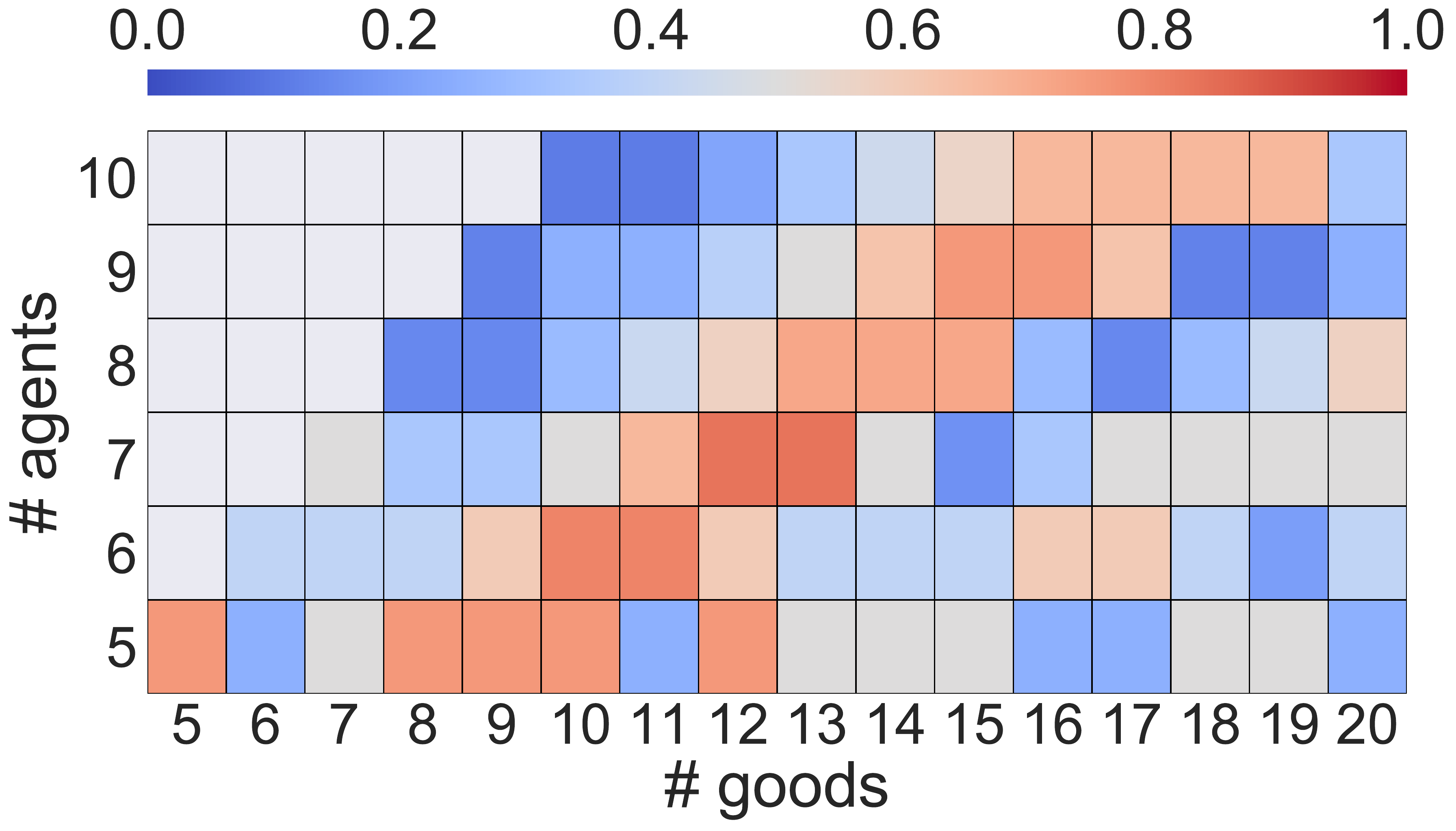} & \includegraphics[width=0.22\textwidth]{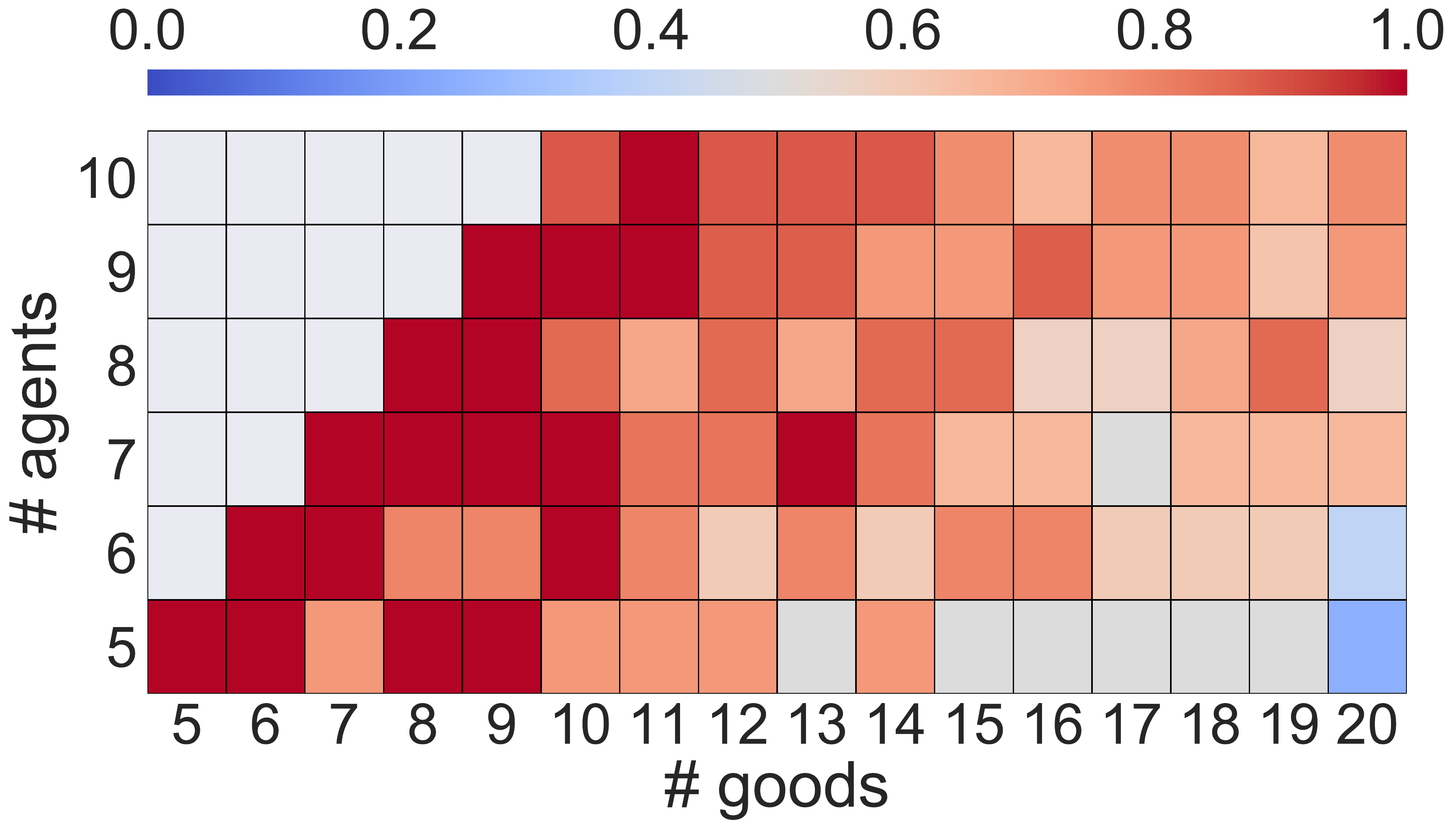} &
 \includegraphics[width=0.22\textwidth]{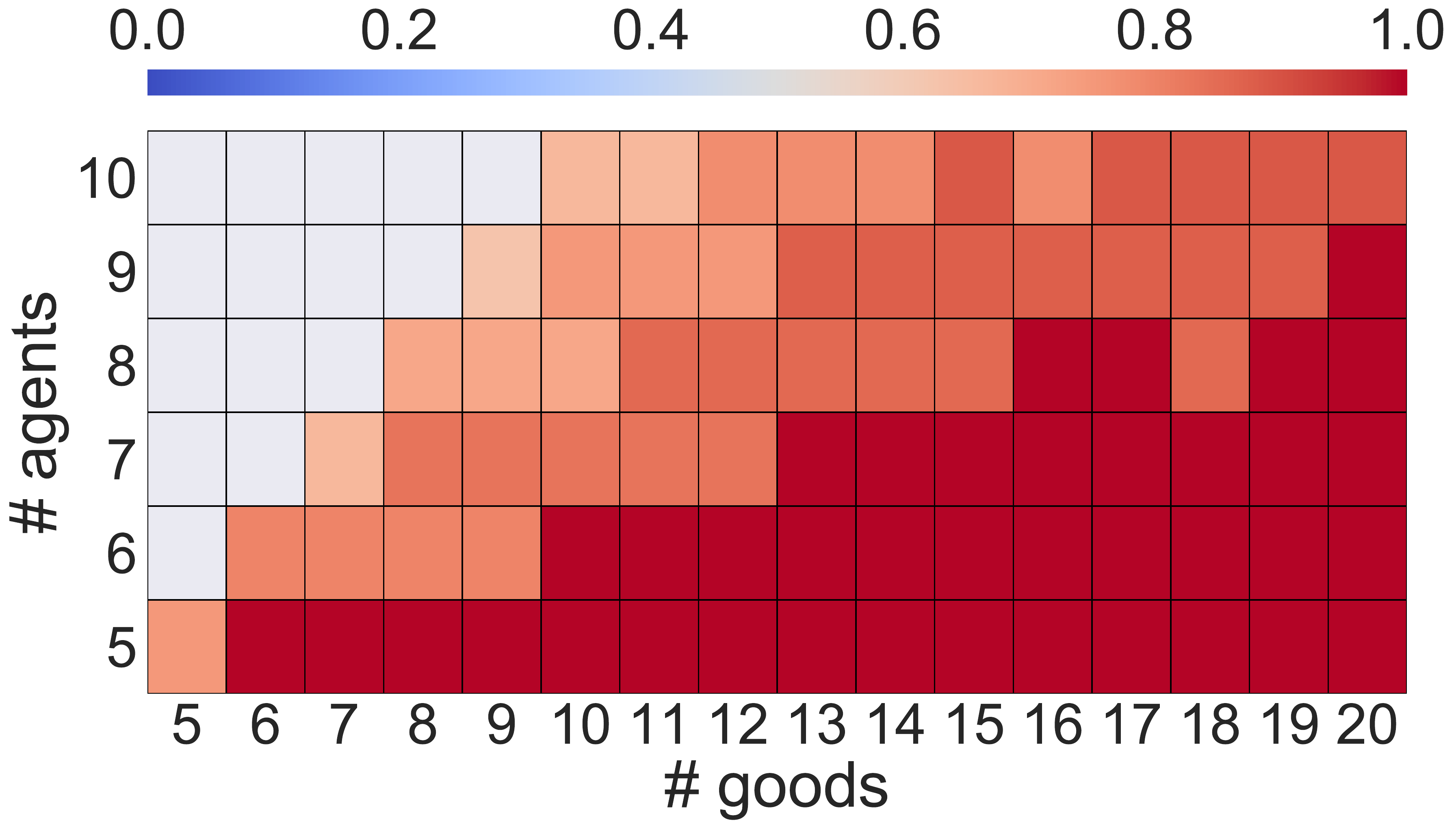} &
 \includegraphics[width=0.22\textwidth]{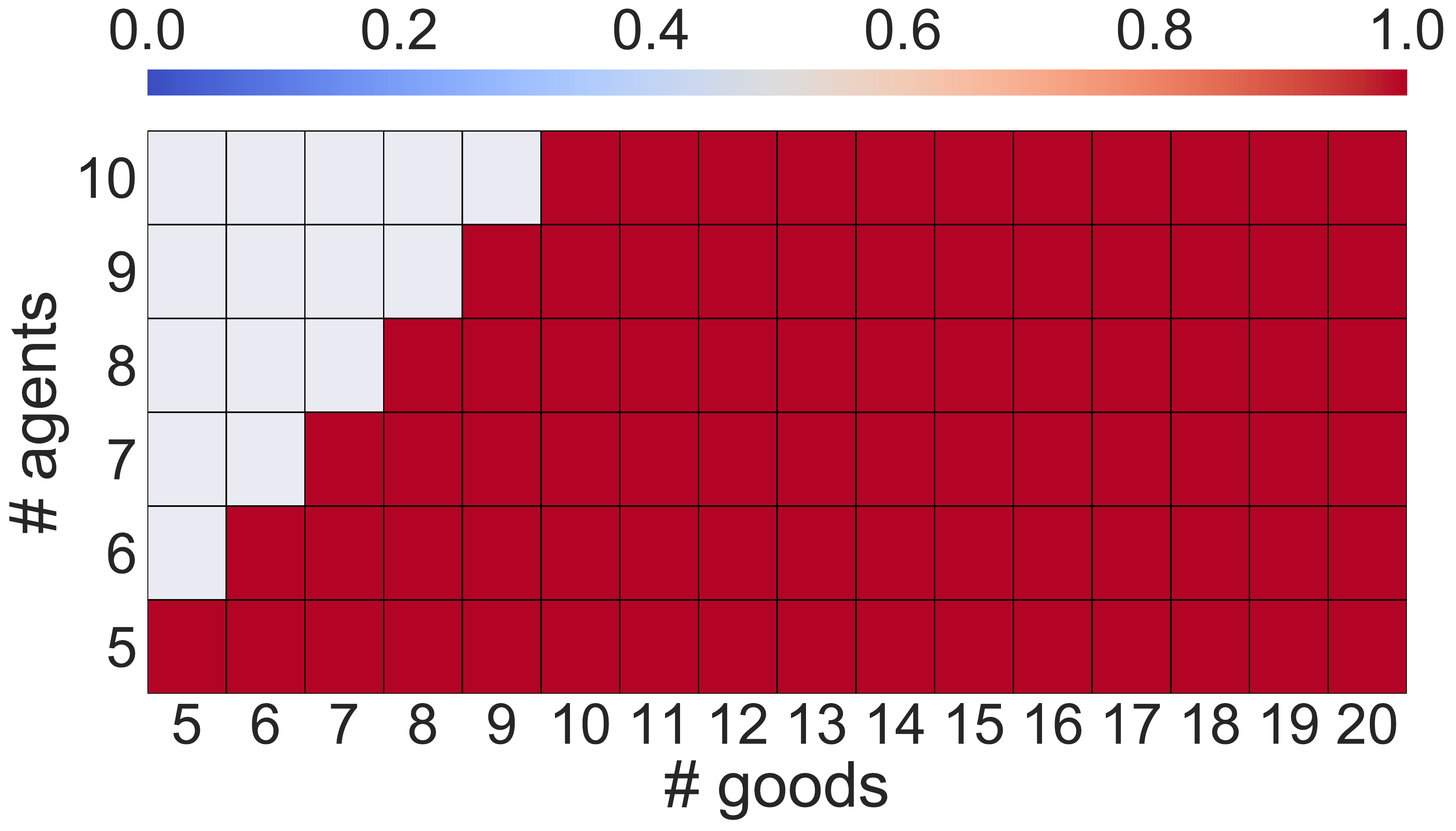} \\
 \hline
 \multicolumn{4}{c}{}\\
 \multicolumn{4}{c}{\textbf{Frequency of envy-freeness}}\\
 \hline
 \footnotesize{\Market{}} & \footnotesize{\RR{}} & \footnotesize{\MNW{}} & \footnotesize{\Envygraph{}}\\
 \includegraphics[width=0.22\textwidth]{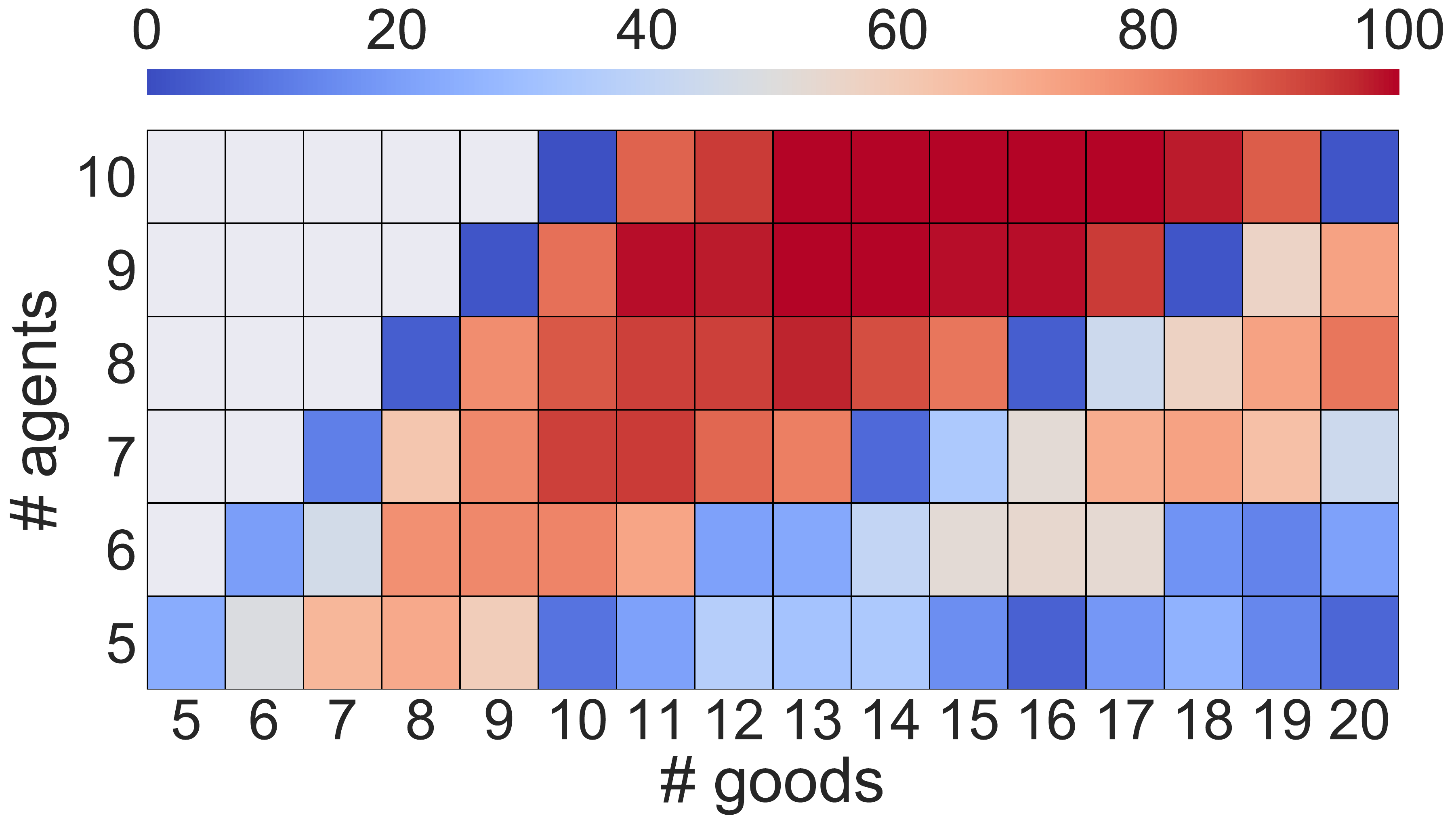} & \includegraphics[width=0.22\textwidth]{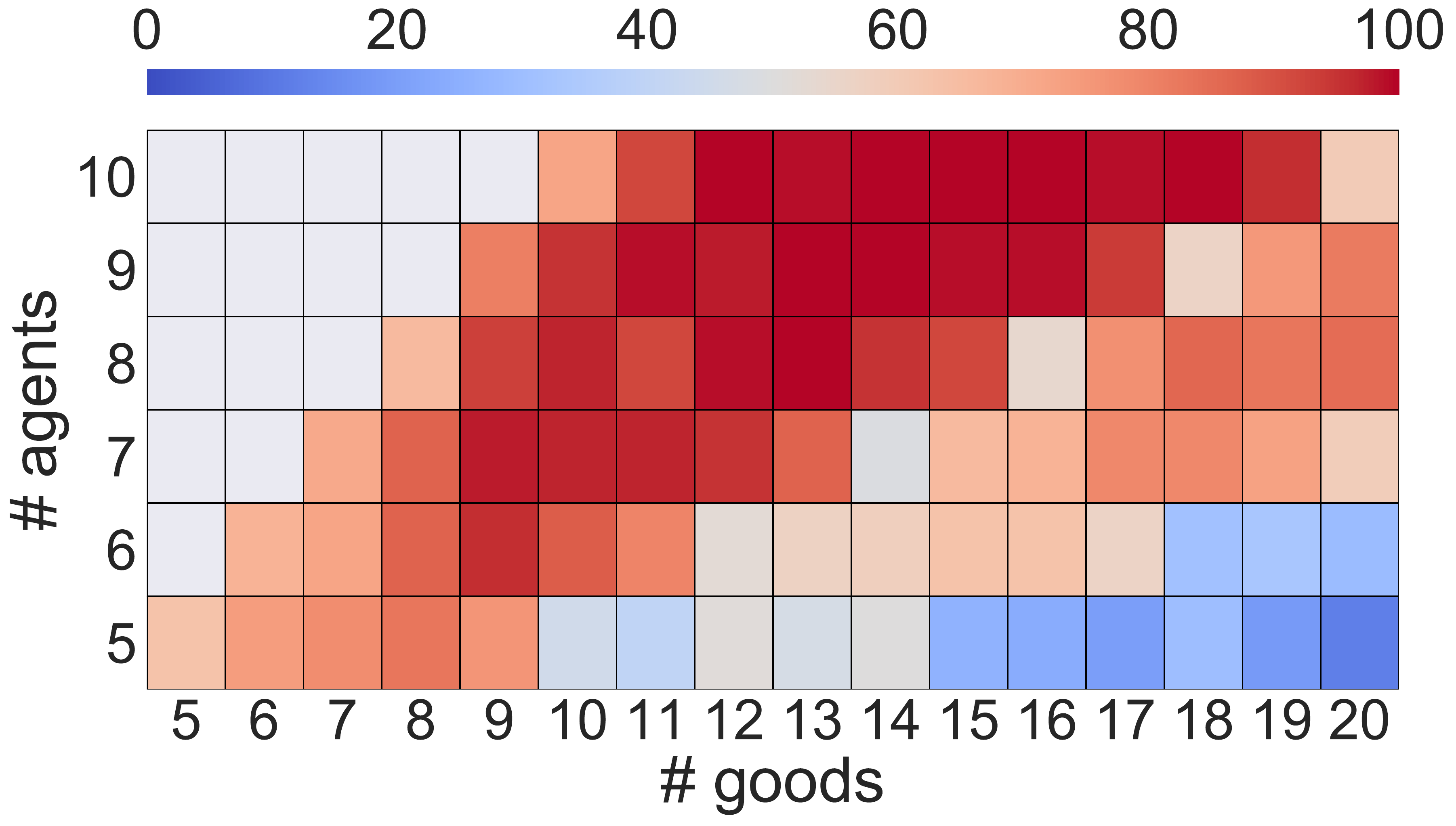} &
 \includegraphics[width=0.22\textwidth]{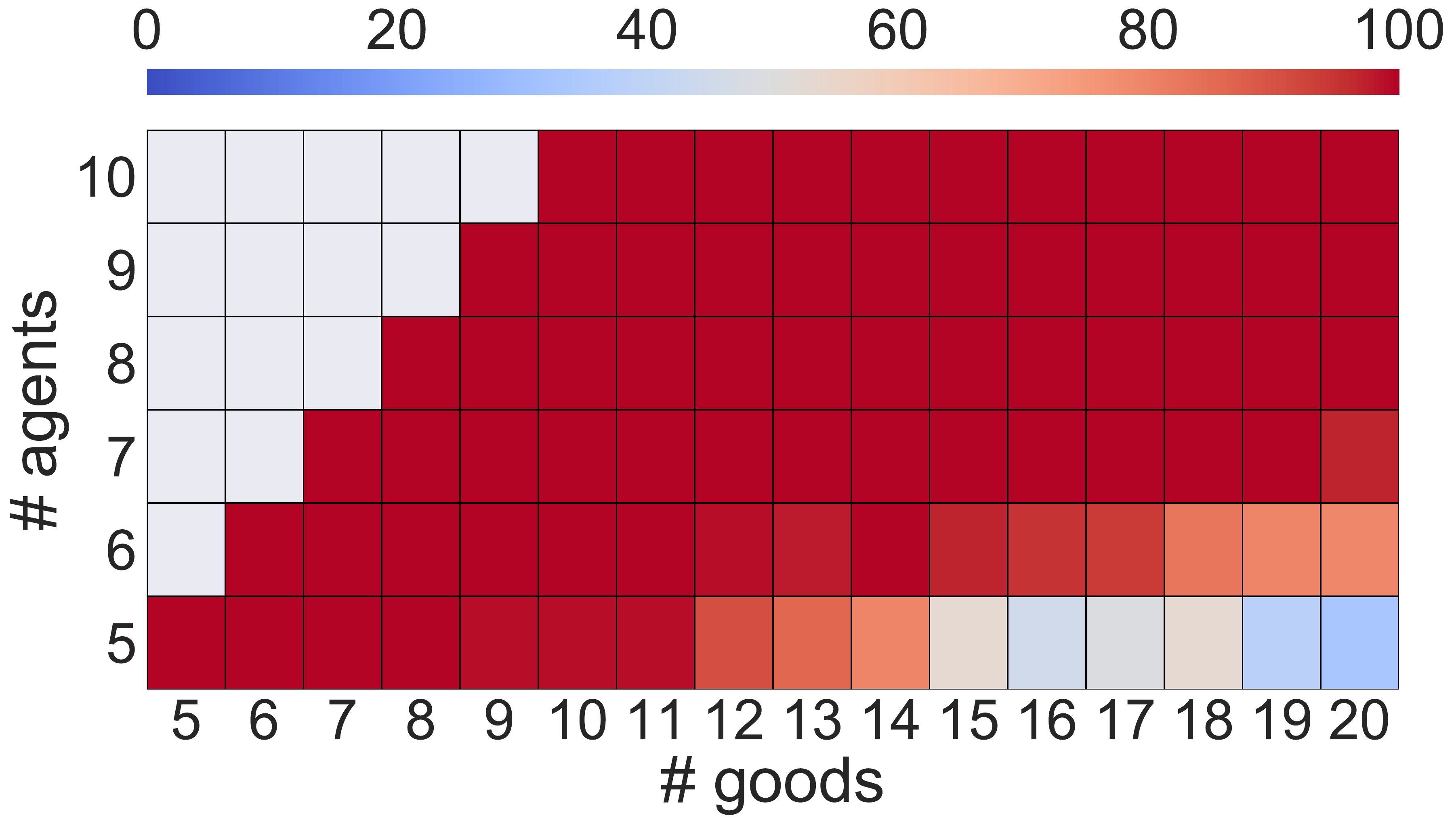} &
 \includegraphics[width=0.22\textwidth]{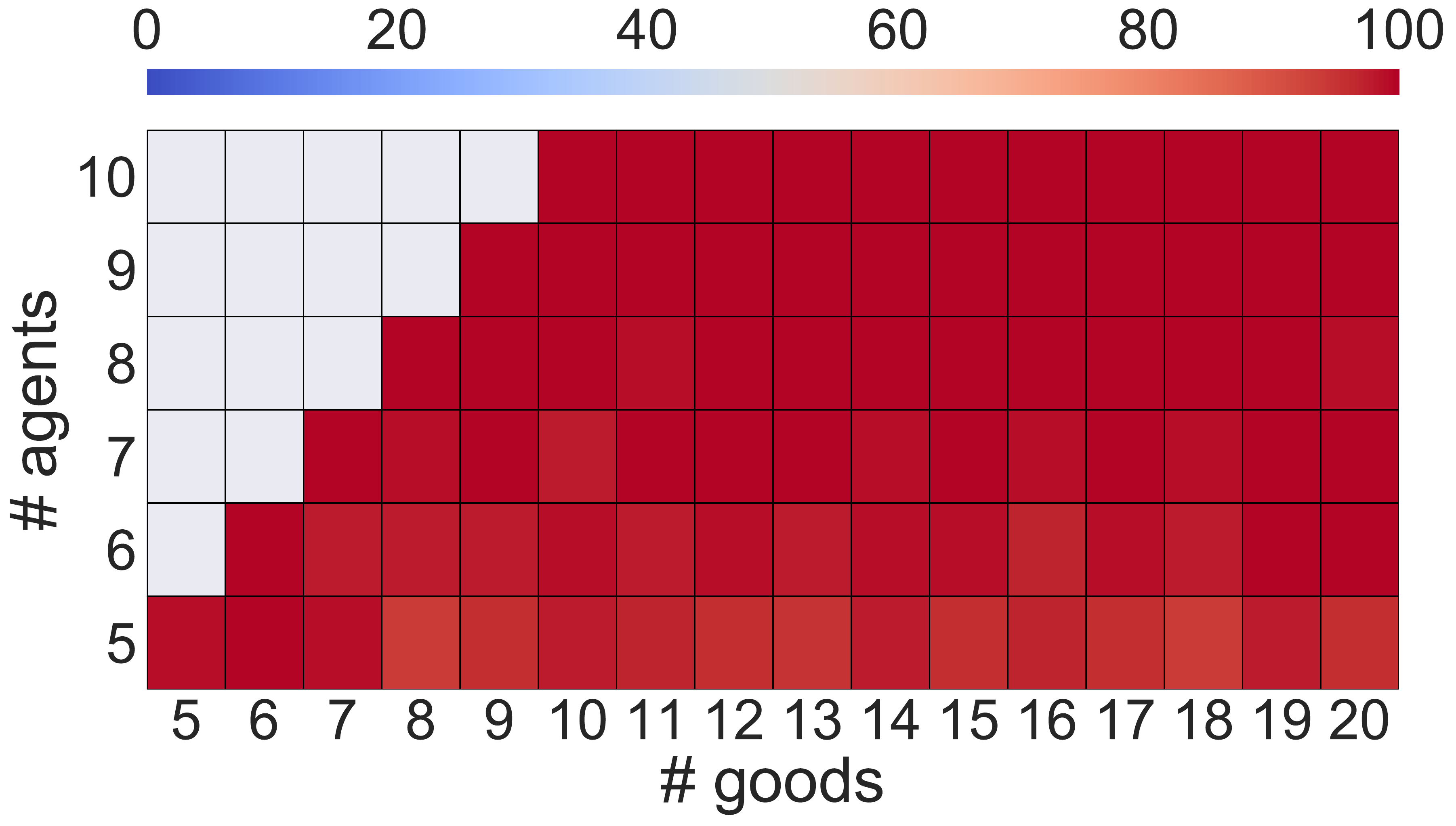} \\
 \hline
 \multicolumn{4}{c}{}\\
 \multicolumn{4}{c}{\textbf{Number of goods that must be hidden on average} (averaged over non-\EF{} instances only)}\\
 \hline
 \footnotesize{\Market{}} & \footnotesize{\RR{}} & \footnotesize{\MNW{}} & \footnotesize{\Envygraph{}}\\
 \includegraphics[width=0.22\textwidth]{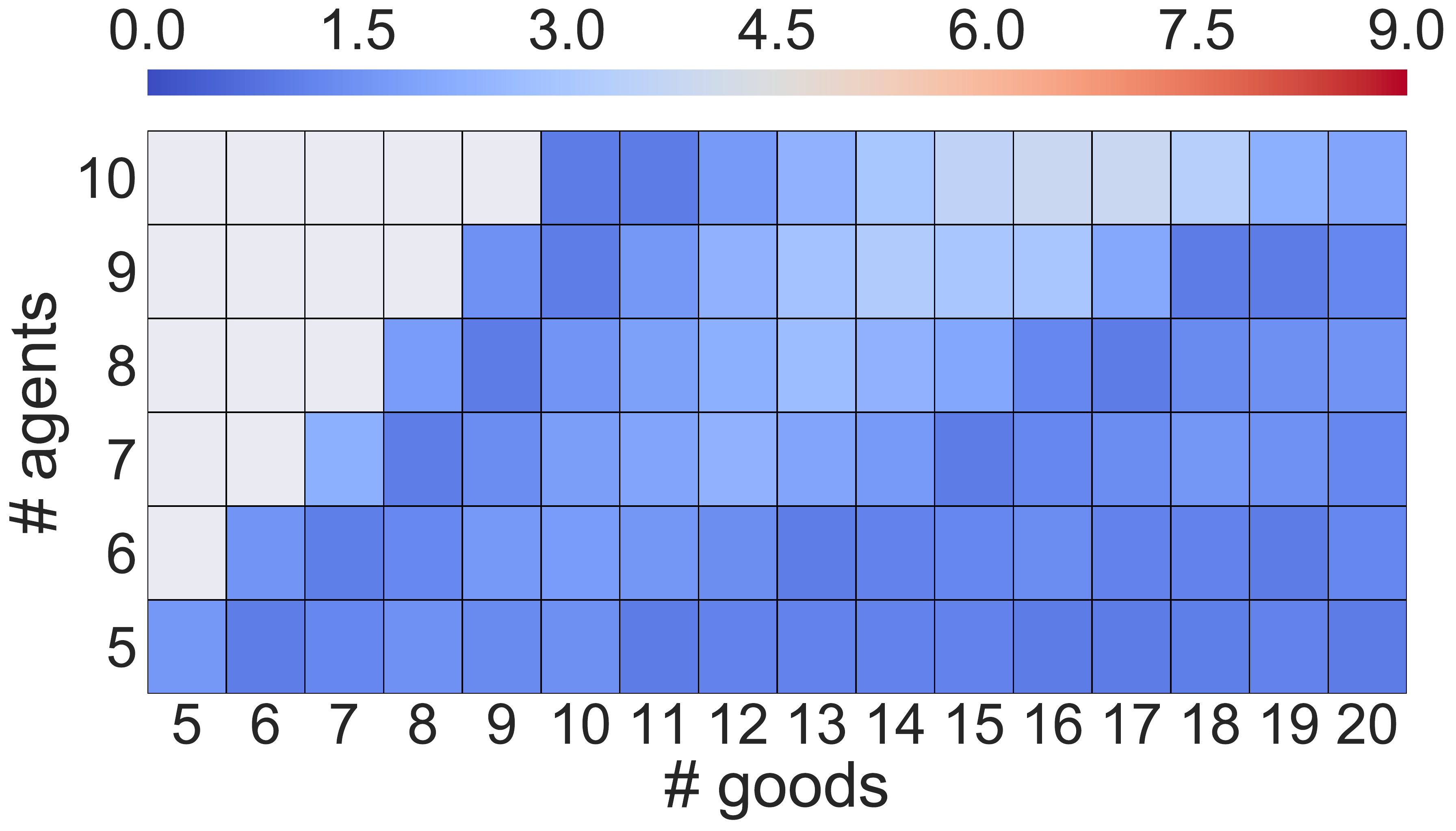} & \includegraphics[width=0.22\textwidth]{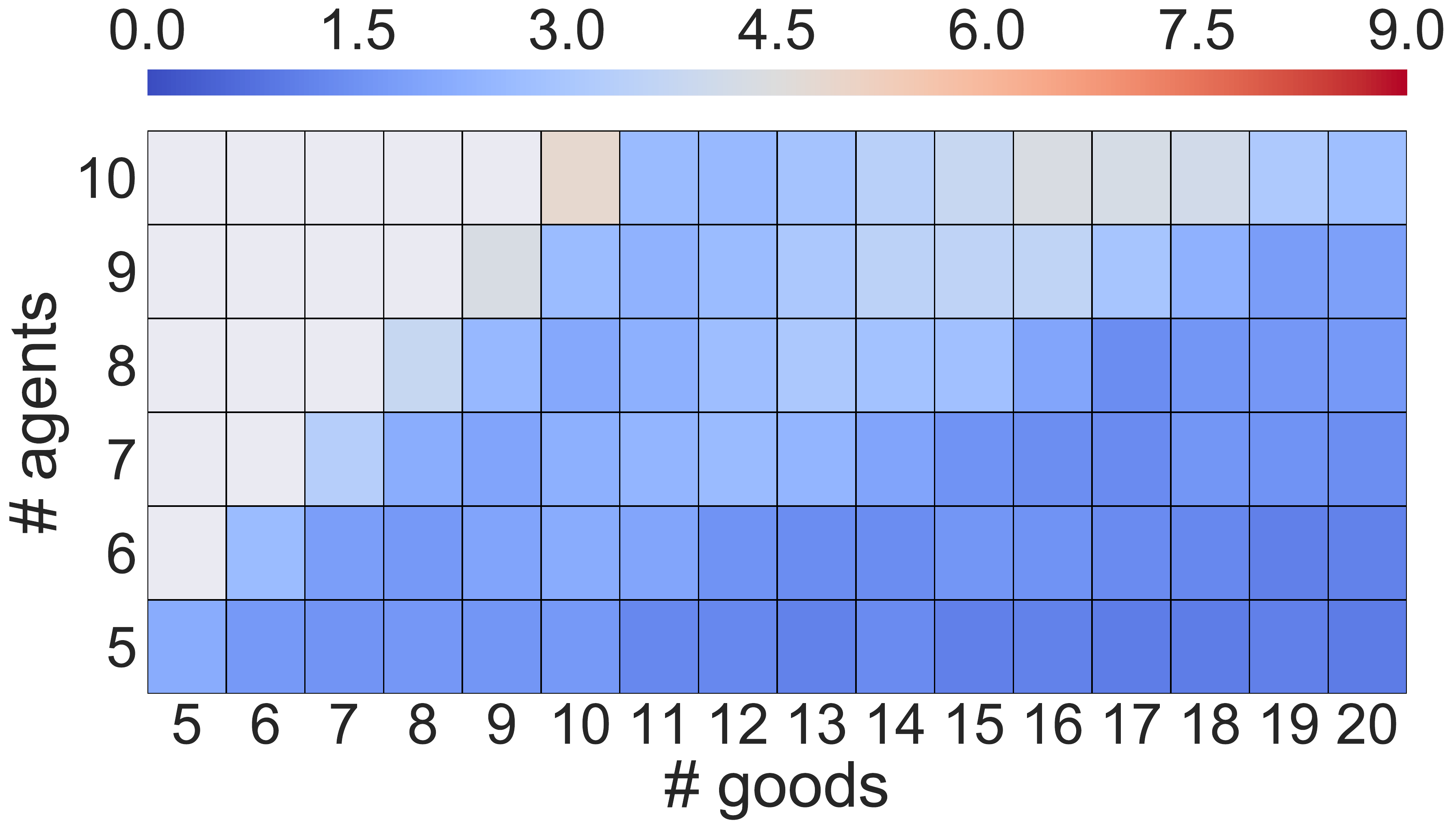} &
 \includegraphics[width=0.22\textwidth]{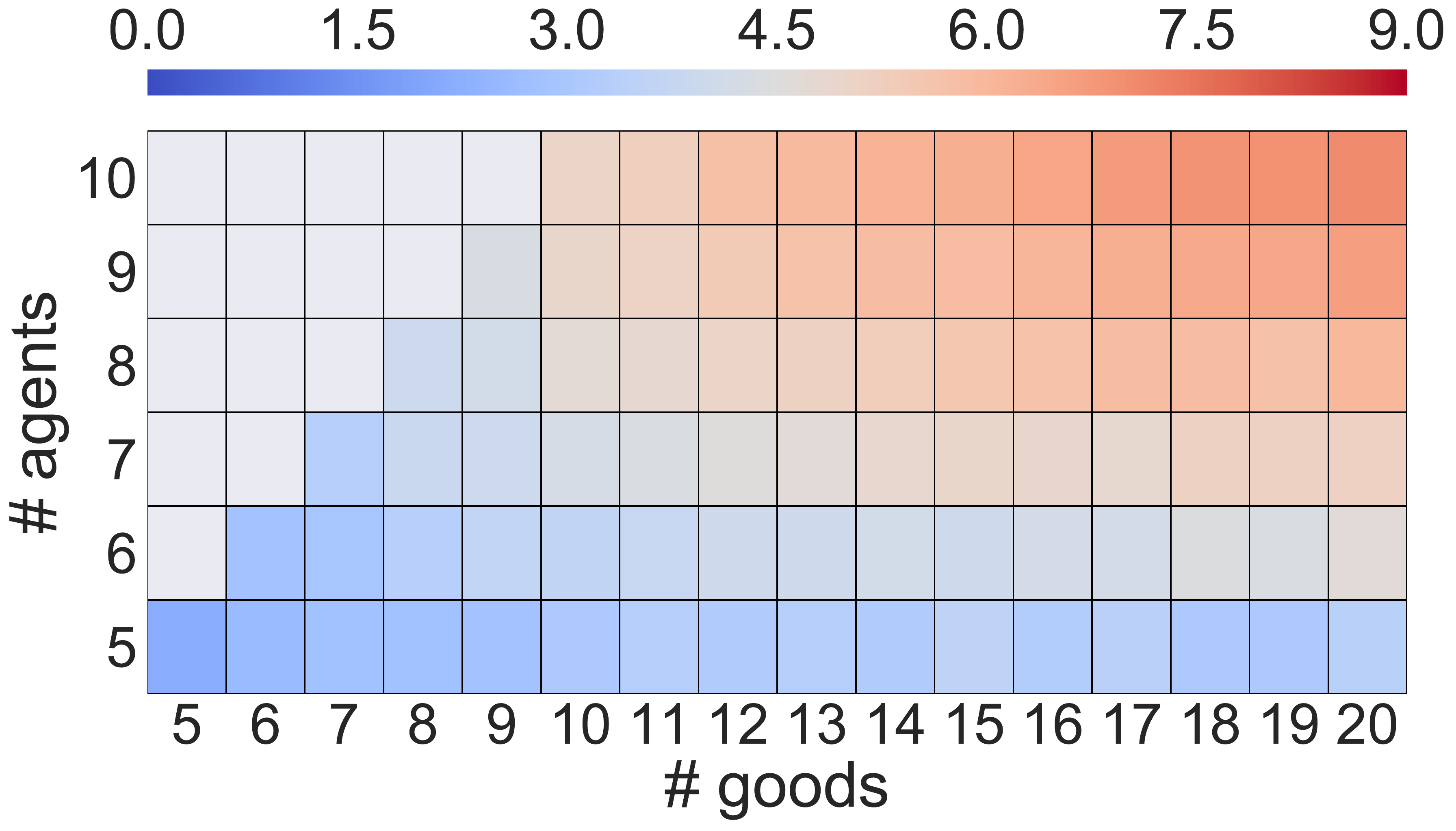} &
 \includegraphics[width=0.22\textwidth]{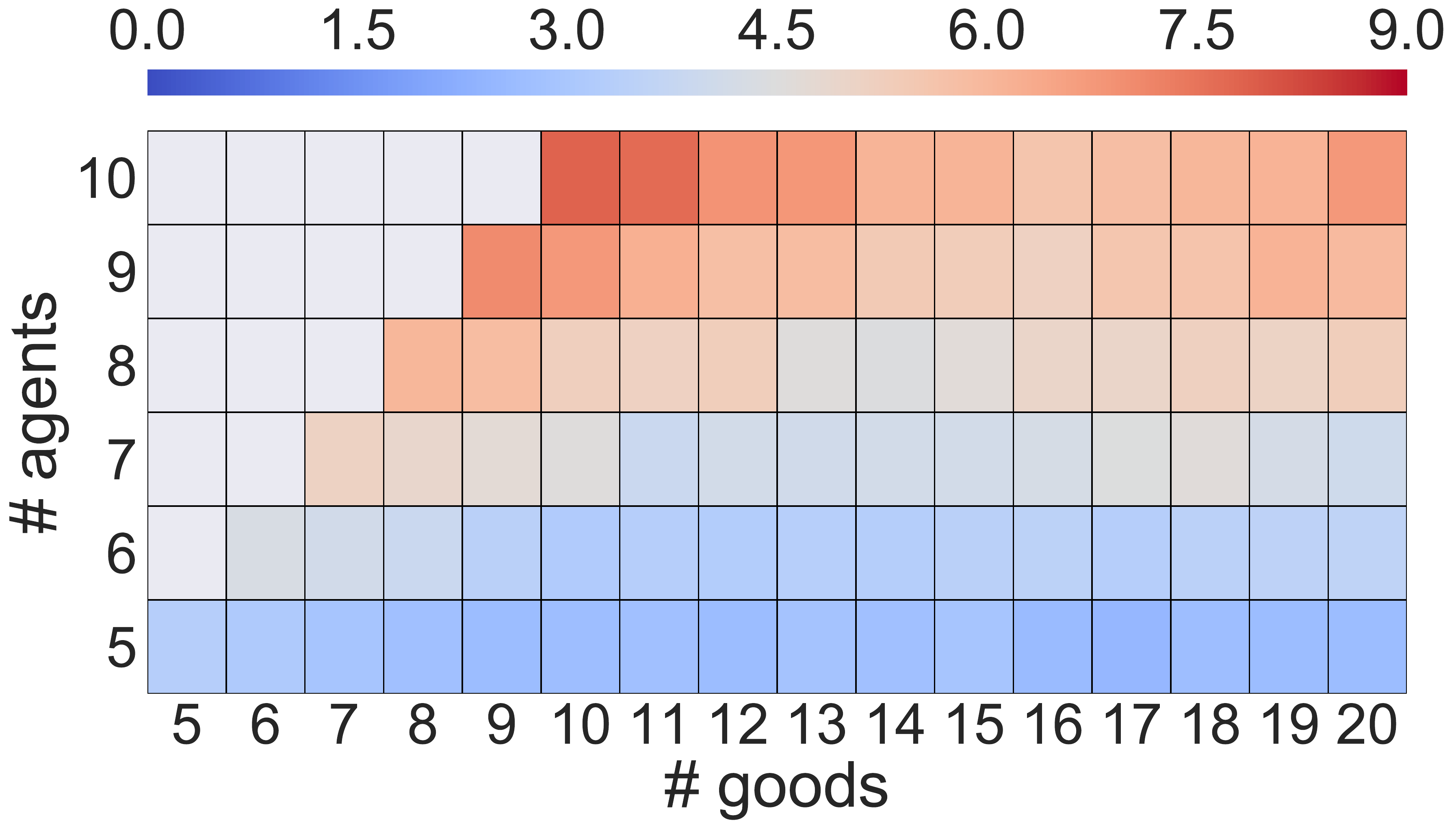} \\
 \hline
 \multicolumn{4}{c}{}\\
 \multicolumn{4}{c}{\textbf{Number of goods that must be hidden in the worst-case} (max over all $100$ instances)}\\
 \hline
 \footnotesize{\Market{}} & \footnotesize{\RR{}} & \footnotesize{\MNW{}} & \footnotesize{\Envygraph{}}\\
 \includegraphics[width=0.22\textwidth]{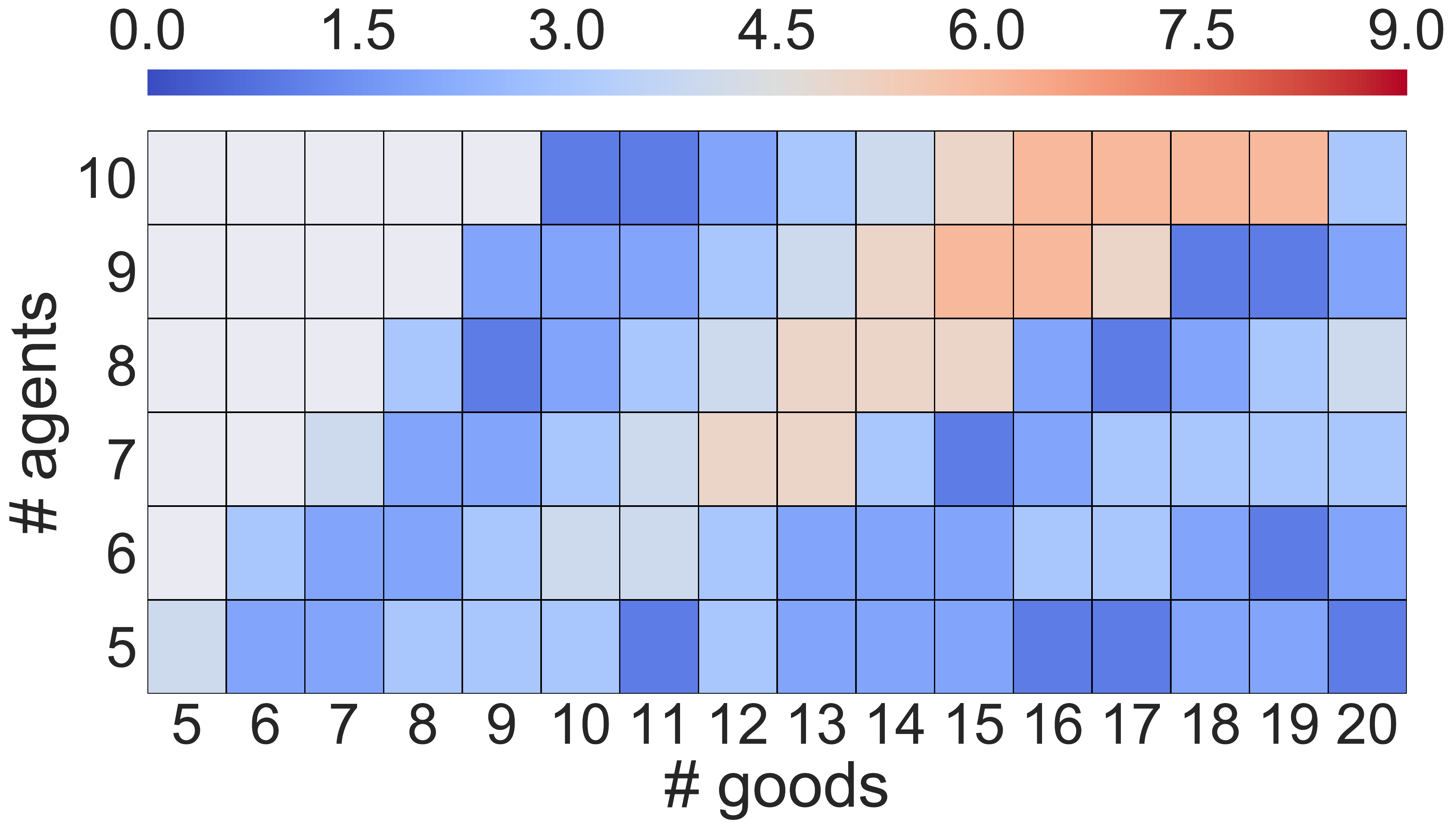} & \includegraphics[width=0.22\textwidth]{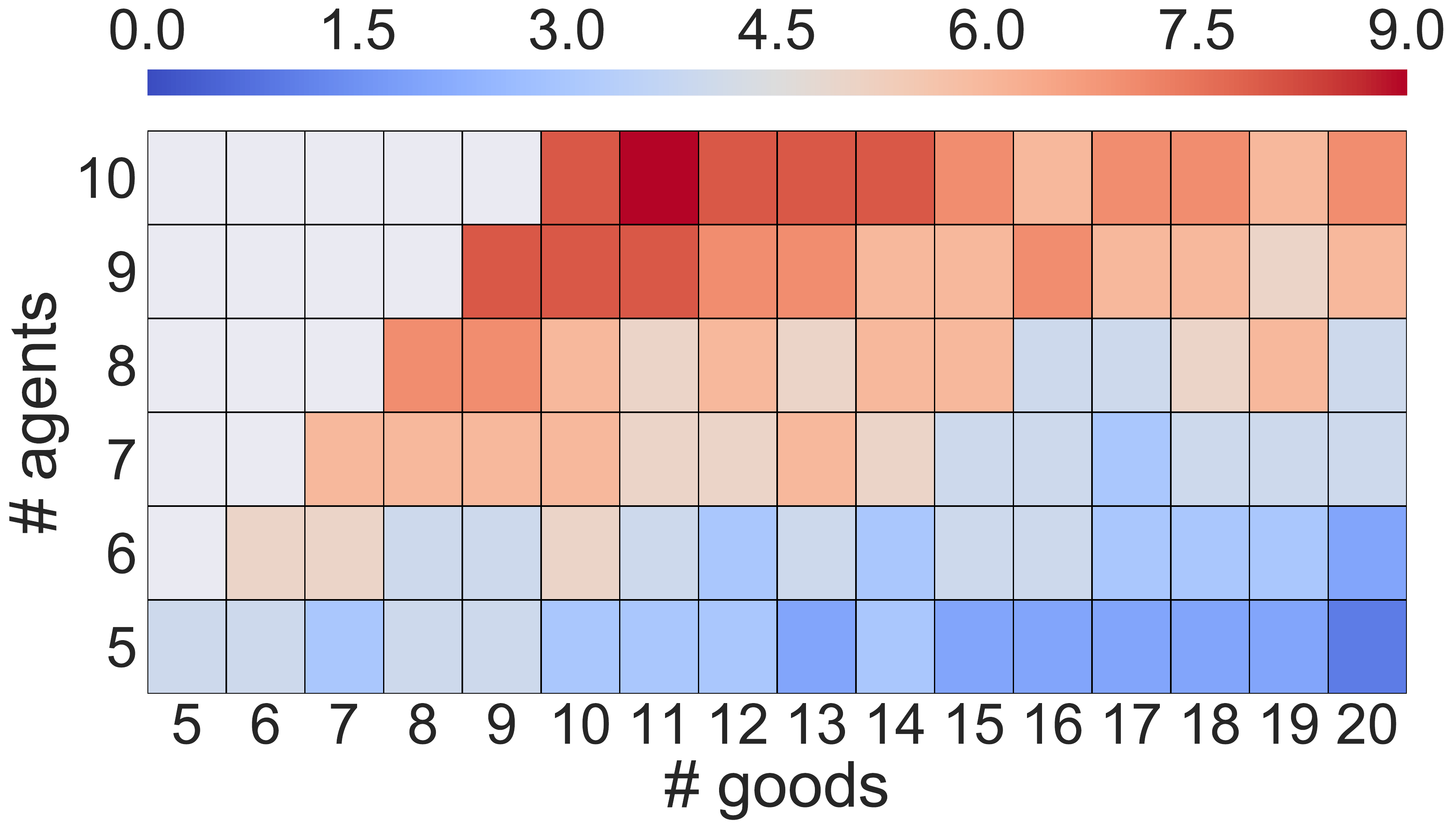} &
 \includegraphics[width=0.22\textwidth]{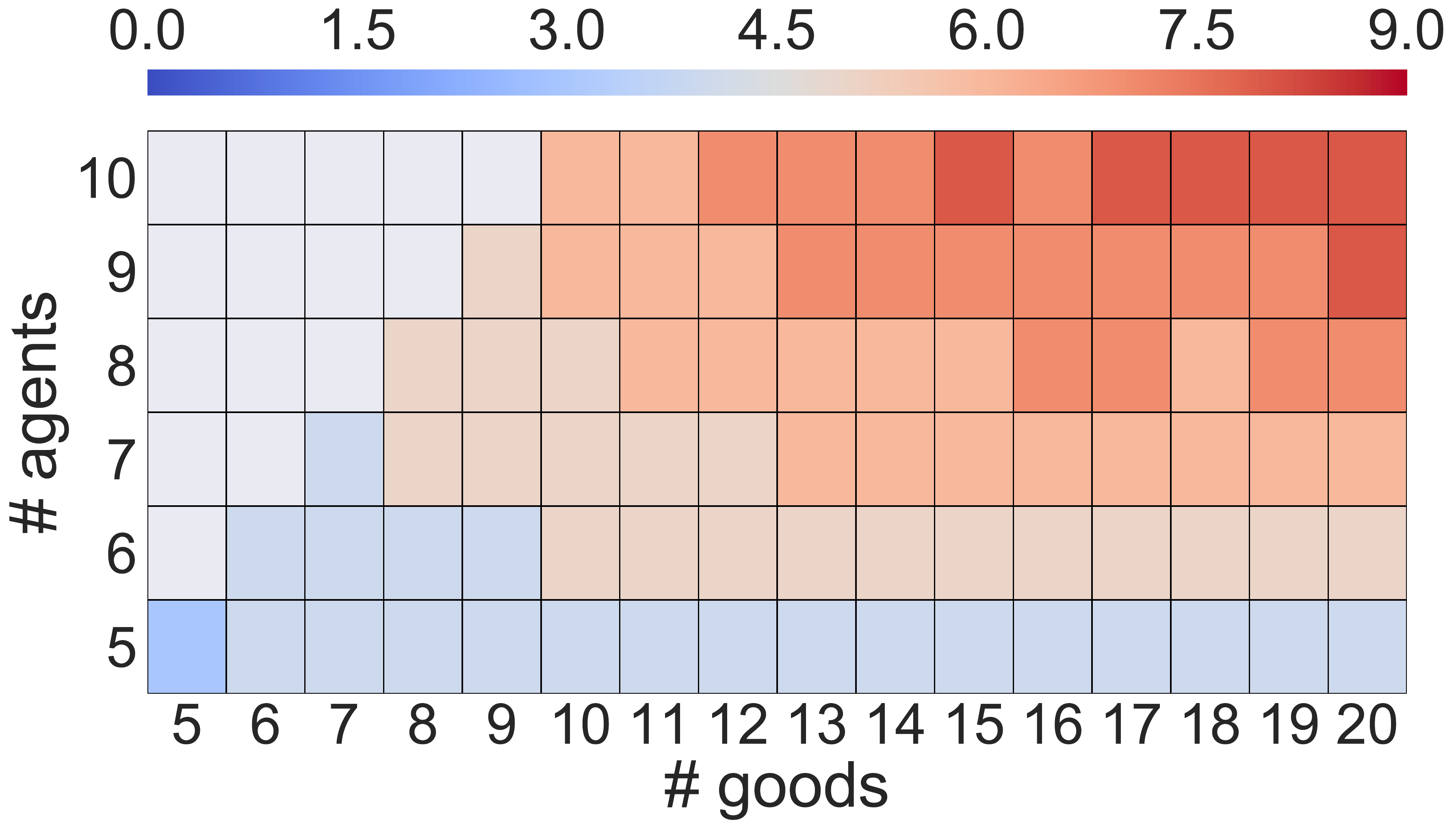} &
 \includegraphics[width=0.22\textwidth]{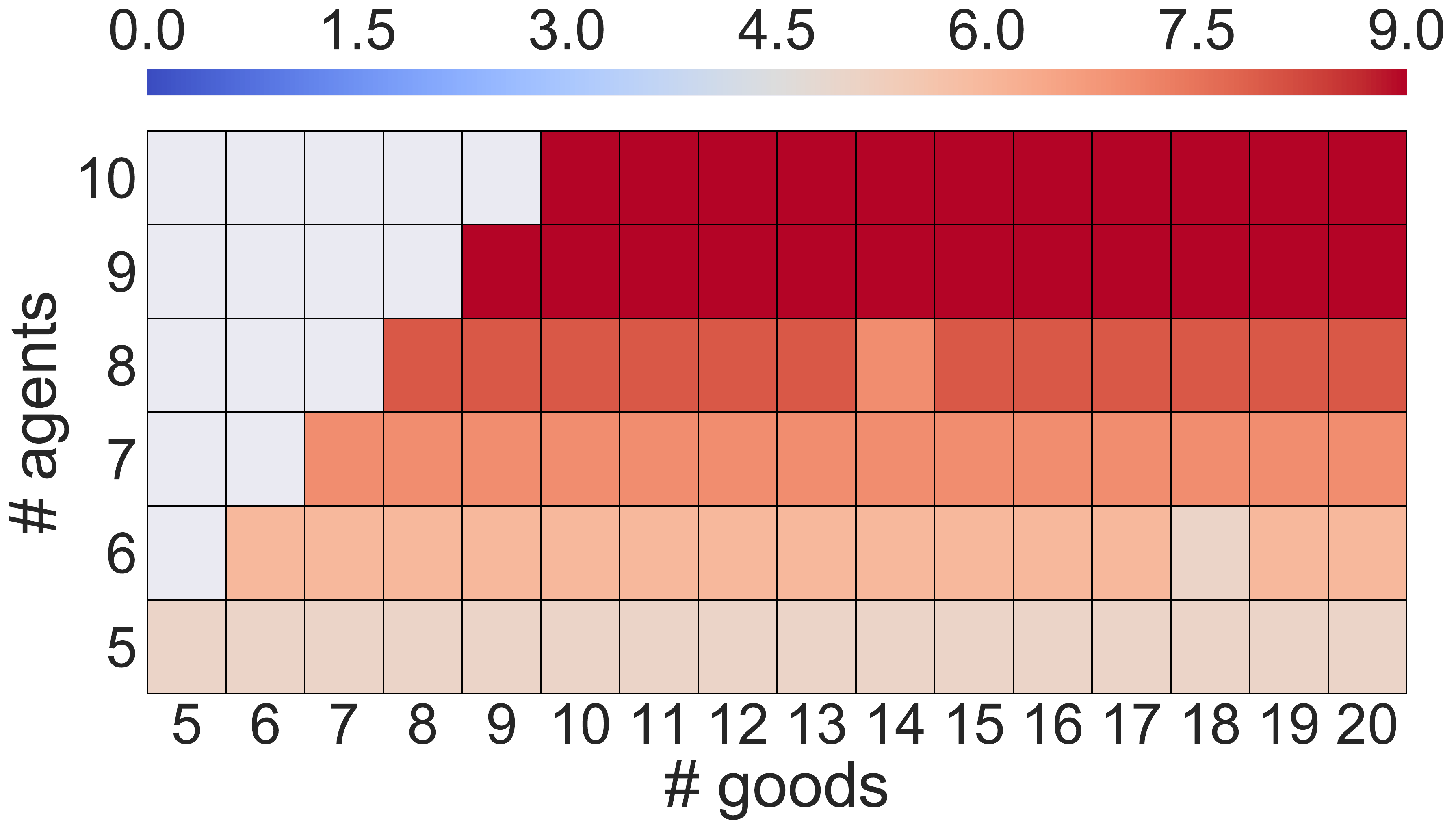} \\
 \hline
 \end{tabular}
 \caption{Comparing various \EF{1} algorithms over synthetically generated binary instances with $v_{i,j} \sim \Ber(0.5)$ i.i.d.}
 \label{tab:Expt_BinaryVals_bias_0.5}
\end{table*}

\end{document}